\documentclass[fleqn]{elsarticle}
\usepackage{amsmath,amsfonts,amssymb,amsthm}
\usepackage{graphicx}
\usepackage{mathtools}
\usepackage{url}
\usepackage[dvipsnames]{xcolor}
\usepackage{booktabs}
\usepackage{tabularx}
\usepackage{multirow}
\usepackage{pdfpages}
\usepackage{color}
\usepackage{cases}
\usepackage{natbib}
\usepackage{tikz}
\usetikzlibrary{calc}
\usepackage{newpxtext}
\let\mathscr\relax
\usepackage{newpxmath}
\usepackage{subfigure}
\usepackage{caption}
\usepackage[english]{babel}
\usepackage{soul}
\theoremstyle{plain}
\newtheorem{theorem}{Theorem}

\newtheorem{lemma}{Lemma}
\newtheorem{remark}{Remark}

\newtheoremstyle{note}{\topsep}{\topsep}{\slshape}{}{\scshape}{}{ }{}
\theoremstyle{note}

\biboptions{sort&compress}
\numberwithin{equation}{section}
\numberwithin{definition}{section}

\mathtoolsset{mathic,centercolon}
\newcommand\mvector{\boldsymbol}

\newcommand\vc{\mvector{c}}

\newcommand\vGamma{\mvector{\Gamma}}

\newcommand\ord{\operatorname{ord}}

\newcommand\rmd{\mathrm{d}}

\newcommand\rmi{\mathrm{i}\mspace{1mu}}

\begin{document}

\begin{frontmatter}
\title{Dynamics and non-integrability of the Swinging Atwood Machine with a massive string: chaos, periodic orbits and resonance structures}

\author[inst1]{Wojciech Szumi\'nski\corref{cor1}}
\ead{w.szuminski@if.uz.zgora.pl}

\author[inst1,inst2]{Jakub Bembenek}
\ead{j.bembenek@stud.uz.zgora.pl}

\cortext[cor1]{Corresponding author}

\address[inst1]{Institute of Physics, University of Zielona G\'ora,
Licealna 9, PL-65-407, Zielona G\'ora, Poland}

\address[inst2]{The Doctoral School of Exact and Technical Sciences,
University of Zielona Góra, al. Wojska Polskiego 69,
65-762, Zielona Góra, Poland}
\begin{abstract}
Building upon our previous studies on nonlinear variable-length pendulum systems
(Szumi\'nski and Maciejewski, 2024; Szumi\'nski and Kapitaniak, 2025), we complete this line of research by investigating the Swinging Atwood Machine with
		a massive string. In contrast to the classical model, the inclusion of string
		inertia introduces an additional configuration-dependent moment of inertia,
		leading to a fundamentally modified Hamiltonian structure and substantially
		richer dynamics.

		To uncover the global organization of the phase space, we combine
		Poincar\'e sections, bifurcation diagrams, and Lyapunov exponent maps with our
		recently developed numerical framework, referred to as \emph{Lyapunov Refined
			Maps}. This approach provides a unified visualization of periodic,
		quasi-periodic, chaotic, and terminating motions, while revealing intricate
		resonance networks and high-order periodic structures hidden inside regular
		regions of the phase space. We further investigate the influence of the string
		mass, the system parameters, and the energy level by constructing Lyapunov maps in the
		parameter space, the initial-condition space, and on fixed-energy surfaces.
		These computations reveal the formation of complex chaotic layers, resonance
		webs, and energy-dependent transitions associated with the geometry of the Hill
		regions.

		The problem of Liouville integrability is addressed within the framework of the
		Morales--Ramis theory based on differential Galois theory. By analysing the
		normal variational equations along explicit non-stationary radial solutions and
		applying the Kovacic algorithm, we prove that the differential Galois group is
		generically $\operatorname{SL}(2,\mathbb{C})$, providing a rigorous obstruction
		to meromorphic Liouville integrability for every nonzero string mass.
		Consequently, the exceptional integrable case of the classical Swinging Atwood
		Machine is shown to be structurally unstable and destroyed by any physically
		realistic inclusion of string inertia.

		Besides establishing these analytical results, the paper completes a systematic
		study of nonlinear variable-length pendulum systems and demonstrates the
		effectiveness of the proposed numerical framework for investigating
		high-dimensional Hamiltonian systems. The methodology is applicable to a broad
		class of nonlinear mechanical models, including variable-length pendula,
		cable-driven mechanisms, adaptive robotic systems, and other systems with
		distributed mass.
	\end{abstract}


\end{frontmatter}

\section{Introduction  and motivation}

Pendulum systems belong to the most classical models of mechanics and dynamical systems, yet they continue to play an important role in contemporary research on nonlinear dynamics and chaos. Despite their simple mechanical construction, they exhibit a broad spectrum of behaviors ranging from regular oscillations to strongly chaotic motion. Over the years, numerous variants have been investigated, including the double pendulum~\cite{Shinbrot:92::,Stachowiak:06::,Stachowiak:15::,Szuminski:25::JSV_VLDP}, the spring pendulum~\cite{Broucke:73::,Lee:97::,Maciejewski:04::c,Szuminski:24::}, systems of coupled pendulums~\cite{Huynh2010,Huynh2013,Elmandouh:16::,Szuminski:20::}, and the swinging Atwood machine~\cite{Tufillaro:84::,Tufillaro:90::,Szuminski:22::,Szuminski:23::}. Owing to their rich dynamics, these models have become standard benchmarks for studying chaos, resonances, and stability transitions~\cite{Levien:93::,Pujol:10::}.

Among them, the swinging Atwood machine occupies a special position. It combines the classical Atwood setup with pendular motion, giving rise to a Hamiltonian system that is both physically transparent and dynamically rich. The classical model assumes a massless string and has been extensively investigated~\cite{Tufillaro:84::,Tufillaro:85::,Tufillaro:90::,Szuminski:22::,Szuminski:23::}. One of its most remarkable properties is the existence of an exceptional Liouville integrable case for the mass ratio $\mu=3$, although generic parameter values exhibit chaotic dynamics. This illustrates the exceptional and fragile nature of integrability in Hamiltonian systems.

From a physical perspective, however, the assumption of a massless string is a strong idealization. Real ropes and cables possess non-negligible mass, and their inertia modifies the dynamics in a natural way. The heavy swinging Atwood machine has therefore attracted increasing attention~\cite{Lemos:17::,Pujol:10::}, since the inclusion of string inertia leads to a configuration-dependent kinetic energy and a substantially more intricate Hamiltonian structure.

Such models belong to the broader class of variable-length and distributed-mass mechanical systems. They arise in numerous scientific and engineering applications~\cite{Cveticanin:12::,IrschikBelyaev:14}. Variable geometry introduces additional degrees of freedom, richer resonance mechanisms, and stronger nonlinear coupling between different modes of motion~\cite{Pesce:03,Olejnik:23b::,Lee:11::}. Beyond their theoretical interest, these systems serve as simplified models of cranes, lifting devices, robotic mechanisms with adaptive geometry, and vibration energy harvesters. Understanding the transition from regular to chaotic motion is therefore essential for ensuring stability, control, and efficient operation in such applications~\cite{JU2006376,MR4459645,Freundlich:20::,PLAUT20133768,YANG2022116727,SHARGHI2022117036,MARSZAL2017251,He:22::,ABOHAMER2023377}.
In this context, the swinging Atwood machine with a massive string can be seen
as a minimal but physically realistic model. It keeps the main ingredients of
more complex variable-length systems, while still being simple enough to allow
for a detailed analytical study. For this reason, it provides a convenient
framework for investigating how physically natural modifications affect the
dynamics and integrability of Hamiltonian systems.

The inclusion of string inertia naturally raises the question of how this
physically motivated modification affects the integrability of the system.
In particular, one may ask whether the exceptional integrable case of the
classical swinging Atwood machine persists once a nonzero linear mass density
of the string is introduced. Our previous studies, together with the present numerical investigations, suggest  that even very small deviation from the massless-string limit, corresponding to $\alpha\ll 1$, leads to a rapid destruction of the regular structures associated with the
integrable case $\mu=3$ and give rise to chaotic dynamics. Similar behaviour has
been observed in other variable-length and coupled pendulum systems%
~\cite{Yakubu:21::,Yakubu:22::,Olejnik:23b::,Szuminski:23::}.

To investigate these phenomena, we first employ global numerical tools based on
Poincar\'e sections, Lyapunov exponent maps, and the newly introduced
Lyapunov Refined Maps (LRM), which combine Lyapunov exponents with periodic-orbit
detection to reveal the resonance organization of the phase space.
These methods provide a detailed global picture of the dynamics and allow one to
identify parameter regions associated with regular, chaotic, and periodic
motions. Nevertheless, numerical techniques alone cannot establish or exclude
Liouville integrability. In parameter-dependent Hamiltonian systems,
integrable cases may occur only on thin subsets of the parameter space, making
them difficult to identify solely by numerical exploration.

For this reason, a rigorous analytical approach is required. In the present
work, we investigate the swinging Atwood machine with a massive string from the viewpoint of integrability theory. Our analysis is based on the Morales--Ramis theory~\cite{Morales:99::,Morales:00::}, which provides necessary
conditions for Liouville integrability in terms of the differential Galois
group of the variational equations along a particular solution. This framework
has been successfully applied to a wide class of Hamiltonian and
non-Hamiltonian systems~\cite{Maciejewski:18::,Yagasaki:18::,Acosta:18::,
	Huang:18::,Combot:18::}, leading to the discovery of numerous integrable and
superintegrable models~\cite{Elmandouh:18::,Szuminski:18a::,Szuminski:18b::}.

The principal objective of this paper is to determine whether the
exceptional Liouville integrability of the classical swinging Atwood
machine with mass ratio $\mu=3$ persists after the inclusion of string
inertia. To this end, we combine a comprehensive numerical analysis
based on Poincar\'e sections, Lyapunov exponent maps, and the proposed
Lyapunov Refined Maps  with a rigorous analytical investigation
using the Morales--Ramis theory and differential Galois methods.
Our main result proves that any nonzero linear mass density of the
string destroys Liouville integrability. Consequently, the classical
integrable case is shown to be structurally unstable with respect to
physically realistic perturbations arising from the inertia of the
string.

The rest of the paper is organized as follows. In Section~2 we introduce the Hamiltonian formulation of the swinging Atwood machine with a massive string and derive the corresponding equations of motion. We also discuss the influence of the string inertia on the structure of the model.
Section~3 is devoted to the construction of an invariant manifold and the derivation of explicit non-stationary particular solutions, which constitute the basis for the subsequent analytical non-integrability analysis.
In Section~4 we investigate the phase-space structure by means of Poincar\'e sections, illustrating the coexistence of regular and chaotic dynamics and the gradual destruction of invariant tori.
Section~5 presents a detailed numerical study based on Lyapunov exponent maps, providing a global characterization of the transition from regular to chaotic motion for different values of the system parameters.
In Section~6, we introduce the Lyapunov Refined Maps methodology, describing both the underlying algorithm and its application to the heavy swinging Atwood machine. This new numerical framework combines Lyapunov exponents with phase-space topology, enabling the distinction between periodic and quasi-periodic dynamics while revealing the organization of periodic orbit families.
Section~7 contains the analytical proof of non-integrability. By deriving the normal variational equations along the explicit particular solutions and applying the Morales--Ramis theory together with the Kovacic algorithm, we prove that the system is not Liouville integrable for any nonzero string mass.
Finally, Section~8 summarizes the main results of the paper and discusses their implications for the dynamics of Hamiltonian systems as well as possible directions for future research.

\section{The heavy Swinging Atwood Machine}



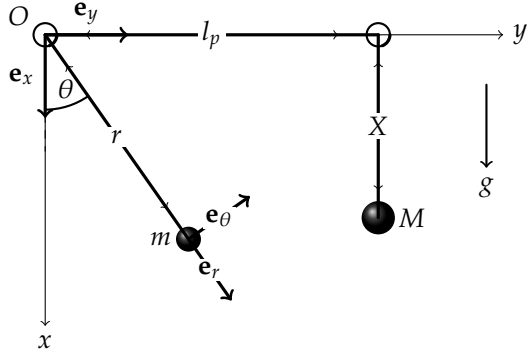
\begin{figure}[t]
	\centering
	\begin{tikzpicture}[scale=1.1, line cap=round, line join=round]

		\def\r{3.0}
		\def\X{2.2}
		\def\lp{1.0}
		\def\LPO{4*\lp}   
		\def\th{35}
		\def\Rpul{0.14}
		\def\e{1.0}       
		\def\elen{0.9}    

		\coordinate (O) at (0,0);

		\draw[->,thin] (O) -- (-0,-3.5) node[below] {$x$};
		\draw[->,thin] (O) -- (5.5,0)  node[right] {$y$};

		\draw[->,very thick] (O) -- (0,-\e) node[midway,left] {$\mathbf e_x$};
		\draw[->,very thick] (O) -- (\e,0)  node[midway,above] {$\mathbf e_y$};

		\node[above left,xshift=-2] at (O) {$O$};

		\draw[thick] (O) circle (\Rpul);
		\coordinate (Q) at (O);

		\coordinate (P) at ({\LPO},0);
		\draw[thick ] (P) circle (\Rpul);

		\coordinate (m) at ({\r*sin(\th)}, {-\r*cos(\th)});
		\shade[ball color=black] (m) circle (0.15);
		\node[right,xshift=-18] at (m) {$m$};

		\coordinate (M) at ({\LPO},{-\X});
		\shade[ball color=black] (M) circle (0.2);
		\node[right,xshift=4] at (M) {$M$};

		\draw[very thick] (Q) -- (m);   
		\draw[very thick] (Q) -- (P);   
		\draw[very thick] (P) -- (M);   

		\draw[very thick] ({\Rpul},0) arc[start angle=0,end angle=-70,radius=\Rpul];
		\draw[very thick] ($(P)+(\Rpul,0)$) arc[start angle=0,end angle=-70,radius=\Rpul];

		\draw[<->] ($(Q)!0.15!(m)$) -- ($(Q)!0.85!(m)$)
		node[midway,fill=white,inner sep=1.2pt] {$r$};

		\draw[<->] ($(P)!0.15!(M)$) -- ($(P)!0.85!(M)$)
		node[midway,fill=white,inner sep=1.2pt] {$X$};

		\draw[<->] ($(Q)!0.12!(P)$) -- ($(Q)!0.88!(P)$)
		node[midway,fill=white,inner sep=1.2pt] {$l_p$};

		\draw[dashed] (Q) -- (0,-1.4);
		\draw[thick] (0,-0.9) arc[start angle=-90,end angle=-90+\th,radius=0.9];
		\node at (0.25,-0.65) {$\theta$};

		\coordinate (er_end) at ($(m)+({\elen*sin(\th)},{-\elen*cos(\th)})$);
		\draw[->,very thick] (m) -- (er_end)
		node[midway,fill=white,inner sep=1pt] {$\mathbf e_r$};

		\coordinate (eth_end) at ($(m)+({\elen*cos(\th)},{\elen*sin(\th)})$);
		\draw[->,very thick] (m) -- (eth_end)
		node[midway,fill=white,inner sep=1pt] {$\mathbf e_\theta$};

		\draw[->,thick] (5.3,-0.6) -- ++(0,-1.0) node[below] {$ g$};

	\end{tikzpicture}
	\caption{Planar geometry of the heavy Swinging Atwood Machine.
		The fixed Cartesian unit vectors $\mathbf e_x,\mathbf e_y$ are attached at the origin,
		while $\mathbf e_r,\mathbf e_\theta$ form the moving polar basis attached to the
		swinging mass $m$. The second pulley is shifted to the right by $l_p$.}
\end{figure}

\subsection{Geometry, kinematics and constraints}

We consider the motion of the Swinging Atwood Machine in a fixed vertical plane.
Let $\{\mathbf{e}_x,\mathbf{e}_y\}$ be a Cartesian basis, where $\mathbf{e}_x$ is
directed vertically downward and $\mathbf{e}_y$ horizontally to the right.
The reference pulley is fixed at the origin $O=(0,0)$.
The system consists of two point masses $m$ and $M$ connected by a uniform,
inextensible string of total length $l$ and constant linear mass density $\lambda$.
The string passes over two ideal pulleys, which are assumed to be massless,
frictionless, and of negligible radius.
The portion of the string  between  the pulleys has a constant length $l_p$,
while the remaining straight segments have total length $l - l_p$.

The string is assumed to slide freely through the pulleys, so that material points
can move from one straight branch to another.
Only the straight portions of the string contribute to the kinetic energy of the
system; the contribution of the string segments wrapped around the pulleys is
neglected.
Under these assumptions, the string undergoes no bending or elastic deformation,
and its motion is fully determined by the motion of the point masses and the imposed
geometric constraints.

The position vector of the swinging mass $m$ is
\begin{align*}
	\mathbf{r}(t)=x(t)\,\mathbf{e}_x+y(t)\,\mathbf{e}_y
	= r(t)\,\mathbf{e}_r(\theta),\qquad \text{where}\qquad
	\mathbf{e}_r(\theta)
	=
	\cos\theta\,\mathbf{e}_x+\sin\theta\,\mathbf{e}_y,
	\qquad
	r(t)=\|\mathbf{r}(t)\|.
\end{align*}
The counterweight $M$ moves along the vertical direction, and its position is
\begin{align*}
	\mathbf{r}_M(t)=X(t)\,\mathbf{e}_x.
\end{align*}
The total length of the string decomposes as
\begin{align*}
	l = r + X + l_p,\qquad \text{where}\qquad
	\qquad l,l_p=\mathrm{const}.
\end{align*}
Introducing
$
	h := l - l_p,
$
the holonomic constraint reads
$
	r + X = h.
$
Hence,
\begin{align*}
	X = h - r,
	\qquad
	\dot X = -\dot r.
\end{align*}
After imposing the constraint, the configuration space is two-dimensional and
parametrized by $(r,\theta)$.

\subsection{Energy decomposition}
The Hamiltonian formulation of the heavy swinging Atwood machine requires an explicit computation of both the kinetic and potential energies of all moving components. In contrast to the classical model with a massless string, the finite linear mass density introduces additional contributions associated with the translational and rotational motion of the string itself. In particular, the swinging branch possesses a nontrivial moment of inertia that depends on its instantaneous length, leading to a configuration-dependent kinetic energy. In this subsection, we derive the complete expressions for the kinetic and potential energies by summing the contributions of the swinging mass, the counterweight, and the three string segments.
\subsubsection{The kinetic energy}
Differentiating $\mathbf{r}(t)=r\,\mathbf{e}_r$ gives
\begin{align}
	\label{eq:dot_r}
	\dot{\mathbf{r}}
	=
	\dot r\,\mathbf{e}_r
	+
	r\dot\theta\,\mathbf{e}_\theta,\qquad \text{where}\qquad 	\mathbf{e}_\theta(\theta)
	=
	-\sin\theta\,\mathbf{e}_x+\cos\theta\,\mathbf{e}_y,
	\qquad
	\mathbf{e}_r\cdot\mathbf{e}_\theta=0.
\end{align}
Therefore,
$
	\|\dot{\mathbf{r}}\|^2
	=
	\dot r^{\,2}+r^2\dot\theta^{\,2}.
$
The kinetic energy of the swinging mass is
\begin{align*}
	T_m
	=
	\frac12 m\left(\dot r^{\,2}+r^2\dot\theta^{\,2}\right).
\end{align*}
The counterweight moves with velocity $\dot{\mathbf{r}}_M=\dot X\,\mathbf{e}_x$,
hence
\begin{align*}
	T_M
	=
	\frac12 M\dot X^{\,2}
	=
	\frac12 M\dot r^{\,2}.
\end{align*}

The vertical branch attached to the counterweight has length $X=h-r$.
Since the subsequent points on the string are connected to each other and the string itself undergoes no bending or elastic deformation, all these material points move with the same absolute speed, hence $|\dot X|=|\dot r|$).
Therefore, its kinetic energy is
\begin{align*}
	T_{\mathrm{str,vert}}
	 & =
	\frac12\int_0^{h-r} \lambda\,\dot r^{\,2}\,ds
	=
	\frac12\lambda(h-r)\dot r^{\,2}.
\end{align*}

The straight horizontal segment between the pulleys has constant length $l_p$ and
undergoes pure translation with the same material speed $|\dot r|$ (the string slides
through both pulleys).
Hence, its kinetic energy equals
\begin{align*}
	T_{\mathrm{str,hor}}
	 & =
	\frac12\int_0^{l_p} \lambda\,\dot r^{\,2}\,ds
	=
	\frac12\lambda l_p\,\dot r^{\,2}.
\end{align*}

The swinging branch is a straight segment of instantaneous length $r$,
which at each time undergoes the same radial motion and rigid rotation
as the swinging mass $m$.
Consequently, the velocity field along this branch follows directly
from the point-mass velocity~\eqref{eq:dot_r}
by noting that: (i) the radial component $\dot r$ is common to all points of the branch,
and
(ii) the transverse component scales linearly with the distance from the pulley.

Thus, for a string element located at arc-length $s\in[0,r]$ from the pulley,
the velocity is
\begin{align*}
	\dot{\mathbf r}_{\mathrm{swing}}(s,t)
	=
	\dot r\,\mathbf e_r
	+
	s\dot\theta\,\mathbf e_\theta,\qquad \text{and}\qquad 	\|\dot{\mathbf r}_{\mathrm{swing}}(s,t)\|^2
	=
	\dot r^{\,2}
	+
	s^2\dot\theta^{\,2}.
\end{align*}
The kinetic energy of the swinging branch is therefore
\begin{align*}
	\begin{aligned}
		T_{\mathrm{str,swing}}
		=
		\frac12\int_0^r
		\lambda\,
		\|\dot{\mathbf r}_{\mathrm{swing}}(s,t)\|^2\,ds  =
		\frac12\int_0^r
		\lambda\left(
		\dot r^{\,2}
		+
		s^2\dot\theta^{\,2}
		\right)ds.
	\end{aligned}
\end{align*}
The direct integration gives
\begin{align*}
	T_{\mathrm{str,swing}}
	=
	\frac{1}{2}\lambda\, r\,\dot r^{\,2}
	+
	\frac{1}{2}I_{\mathrm{str}}(r)\dot\theta^{\,2},\qquad \text{where}\qquad 	I_{\mathrm{str}}(r)
	:=
	\int_0^r s^2\,\lambda\,ds
	=
	\frac{\lambda r^3}{3},
\end{align*}
is the moment of inertia of the swinging string segment
about the pulley.

Collecting all contributions, the total kinetic energy of the system reads
\begin{align}
	T
	=
	\frac12\left(m+M+\lambda l\right)\dot r^{\,2}
	+
	\frac12\,r^2\left(m+\frac{\lambda r}{3}\right)\dot\theta^{\,2}.
\end{align}
\subsubsection*{Potential energy}

The vertical height is measured opposite to the $x$--axis.
Therefore, the
gravitational potential is
\begin{align*}
	V_m = -mg\,r\cos\theta.
\end{align*}

The counterweight $M$ has vertical coordinate $X=h-r$, hence we have\begin{align*}
	V_M = -Mg(h-r).
\end{align*}

The potential energy of the massive string is obtained by integrating the height
of its material points.
For the vertical branch,
\begin{align*}
	V_{\mathrm{str,vert}}
	=
	\int_0^{h-r}\lambda g(-s)\,ds
	=
	-\frac12\lambda g(h-r)^2.
\end{align*}
For the swinging branch,
\begin{align*}
	V_{\mathrm{str,swing}}
	=
	\int_0^r \lambda g\,(-s\cos\theta)\,ds
	=
	-\frac12\lambda g r^2\cos\theta.
\end{align*}
The pulley segment contributes only a constant and is omitted.
Up to an additive constant, the total potential energy is therefore
\begin{align}
	V(r,\theta)
	=
	g\,r\Bigl[M+\lambda h - m\cos\theta\Bigr]
	-\frac12\lambda g r^2\bigl(1+\cos\theta\bigr).
\end{align}

\subsection{Lagrangian and rescaling}

Having specified the geometry of the system together with
energy contributions, we are now in a position to formulate the Lagrangian.
The total Lagrangian $L=T-V$ of the planar Swinging Atwood Machine with a massive
string and a constant pulley length, which completely defines the dynamical model,
is given by
\begin{align}
	\begin{aligned}
		\label{eq:lag}
		L
		=
		\frac12\left(m+M+\lambda l\right)\dot r^{\,2}
		+
		\frac12\,r^2\left(m+\frac{\lambda r}{3}\right)\dot\theta^{\,2}
		- g\,r\Bigl[M+\lambda h - m\cos\theta\Bigr]
		+\frac12\lambda g r^2\bigl(1+\cos\theta\bigr).
	\end{aligned}\end{align}

To simplify the formulation and avoid a singular nondimensionalisation in the
degenerate limit $h=l-l_p\to 0$, we scale the length by the total string
length $l$ and use the corresponding gravitational time scale. We introduce
dimensionless variables
\begin{align*}
	r & = l\,R,
	\qquad
	\theta = \Theta,
	\qquad
	t = \sqrt{\frac{l}{g}}\,\tau,
	\qquad
	\mathcal{L}:=\frac{L}{mgl},\qquad \text{with}\qquad mgl\neq 0.
\end{align*}
We also define the dimensionless, positive  parameters
\begin{align*}
	\mu := \frac{M}{m},
	\qquad
	\alpha := \frac{\lambda l}{m},
	\qquad
	\eta := \frac{h}{l}=1-\frac{l_p}{l}.
\end{align*}
Note that the limit $h\to 0$ corresponds to $\eta\to 0$ and is now admissible at the
level of dimensionless variables.

For convenience, we introduce
\begin{equation}
	\begin{aligned}
		\label{dupa}
		\mathcal{M}  := 1+\mu+\alpha,                          \qquad
		D(R)         := R^{2}\Bigl(1+\frac{\alpha R}{3}\Bigr), \qquad
		Q(R)         := R+\frac{\alpha}{2}R^{2}.
	\end{aligned}
\end{equation}
Note that, by construction $
	D'(R)=2\,Q(R)$,
which will be   used to simplify intermediate expressions.

Thus,  starting from the Lagrangian~\eqref{eq:lag}, we obtain the dimensionless Lagrangian written in a compact form
\begin{align}
	\begin{aligned}
		\label{eq:nlag}
		\mathcal{L}
		 & =
		\frac{1}{2}\,\mathcal{M}\,\dot R^{\,2}
		+
		\frac{1}{2}\,D(R)\,\dot \Theta^{\,2}
		-
		\mathcal{V}(R,\Theta),
	\end{aligned}
\end{align}
Here $\mathcal{V}(R,\Theta)$ denotes the dimensionless potential, which can be  written as a sum of two separated parts as follows
\begin{align}
	\label{eq:V_Q}
	\mathcal{V}(R,\Theta)
	=
	V_0(R)-Q(R)\cos\Theta,\qquad\text{with}\qquad V_0:=R(\mu+\alpha\eta)-\frac{\alpha}{2}R^{2},\quad Q(R)         := R+\frac{\alpha}{2}R^{2}
\end{align}

\subsection{Hamiltonian formulation}

Starting from the dimensionless Lagrangian~\eqref{eq:nlag}, we introduce the conjugate momenta
\begin{align*}
	\begin{aligned}
		P_R
		:=
		\frac{\partial \mathcal{L}}{\partial \dot R} & =
		\mathcal{M}\,\dot R,\qquad
		P_\Theta
		:=
		\frac{\partial \mathcal{L}}{\partial \dot \Theta}=
		D(R)\,\dot \Theta.
	\end{aligned}
\end{align*}
Hence,
\begin{align*}
	\dot R
	=
	\frac{1}{\mathcal{M}}\,P_R,
	\qquad
	\dot \Theta
	=
	\frac{1}{D(R)}\,P_\Theta.
\end{align*}
The (dimensionless) Hamiltonian is obtained by the Legendre transform
\begin{align*}
	\mathcal{H}(R,\Theta,P_R,P_\Theta)
	:=
	P_R\dot R+P_\Theta\dot\Theta-\mathcal{L},
\end{align*}
which yields
\begin{align}
	\label{eq:HH}
	\mathcal{H}
	=
	\frac{P_R^{2}}{2\mathcal{M}}
	+
	\frac{P_\Theta^{2}}{2D(R)}
	+
	\mathcal{V}(R,\Theta),\qquad 	\mathcal{V}(R,\Theta)
	=
	V_0(R)-Q(R)\cos\Theta,
\end{align}
where the components $V_0(R)$ and $Q(R)$ are previously defined in~\eqref{eq:V_Q} and in~\eqref{dupa}.

The Hamiltonian equations of motion define a four-dimensional autonomous system
of first-order ordinary differential equations of the form
\begin{align}
	\label{eq:Ham_sys}
	\dot R
	=
	\dfrac{1}{\mathcal{M}}\,P_R,                         \quad
	\dot \Theta
	=
	\dfrac{1}{D(R)}\,P_\Theta,                             \quad		\dot P_R
	=
	\dfrac{1}{2}\,\dfrac{D'(R)}{D(R)^{2}}\,P_\Theta^{2}
	-\dfrac{\partial \mathcal{V}}{\partial R}(R,\Theta),    \quad
	\dot P_\Theta
	=
	-\,Q(R)\sin\Theta.
\end{align}
In the subsequent sections, we analyze the Hamiltonian system
\eqref{eq:HH}--\eqref{eq:Ham_sys} with respect to its qualitative dynamics and
Liouville integrability. Special attention is devoted to the weak string--mass
regime $\alpha<1$, in which the system may be viewed as a perturbation of the
classical swinging Atwood machine. This regime is particularly suitable for the
application of analytical methods and differential--Galois integrability tests.

\section{Invariant manifold and particular solutions}
We consider the Hamiltonian function~\eqref{eq:HH} and its corresponding equations of motion~\eqref{eq:Ham_sys}.
In what follows we linearise the system~\eqref{eq:Ham_sys} about the invariant
manifold
\begin{align}
	\label{eq:M_invariant_linstab}
	\mathcal{N}
	:=
	\{(R,\Theta,P_R,P_\Theta):\ \Theta=0,\ P_\Theta=0\},
\end{align}
which corresponds to purely radial motion.
Along $\mathcal{N}$ the reduced dynamics is governed by the variables $(R,P_R)$,
while the transverse variables $(\Theta,P_\Theta)$ decouple at the linear level.

Restricting the dynamics to $\mathcal{N}$ yields the reduced one-degree-of-freedom
Hamiltonian subsystem
\begin{align}
	\label{eq:v_M}
	\dot R
	=
	\frac{1}{\mathcal{M}}\,P_R,
	\qquad
	\dot P_R
	=
	-\,\frac{\partial \mathcal{V}}{\partial R}(R,0).
\end{align}
Using the decomposition~\eqref{eq:V_Q}, we obtain
\begin{equation}
	\begin{split}
		\label{eq:rel}
		\mathcal V(R,0)
		=
		R(\mu+\alpha\eta-1)-\alpha R^{2}, \qquad
		\mathcal V_{R}(R,0)
		=
		\mu+\alpha\eta-1-2\alpha R.
	\end{split}
\end{equation}

\subsection{Stationary solutions and their stability}
\label{subsec:stationary_particular}

A stationary solution on the invariant manifold $\mathcal{N}$ corresponds to an
equilibrium of the reduced subsystem~\eqref{eq:v_M}. It is characterized by
\begin{align*}
	\dot R=0,\quad \dot P_R=0
	\quad\Longleftrightarrow\quad
	P_R=0,\quad \mathcal V_{R}(R,0)=0.
\end{align*}
For $\alpha\neq0$, this condition yields a unique stationary radial position
\begin{align}
	\label{eq:delta}
	R_*=\delta
	:=
	\frac{\mu+\alpha\eta-1}{2\alpha},
\end{align}
with $\Theta=0$ and $P_\Theta=0$. Consequently, the full Hamiltonian
system~\eqref{eq:Ham_sys} admits the equilibrium point
\begin{align}
	\label{eq:equilibrium_full}
	\vGamma_*
	=
	(R_*,\Theta,P_R,P_\Theta)
	=
	(\delta,0,0,0),
	\qquad (\alpha\neq0).
\end{align}

To classify $\vGamma_*$, we compute the Jacobian matrix of the Hamiltonian vector
field associated with~\eqref{eq:Ham_sys} at $\vGamma_*$. Introducing
\begin{align*}
	\mathcal{M}_*:=\mathcal{M},
	\qquad
	D_*:=D(\delta)
	=\delta^{2}\Bigl(1+\frac{\alpha\delta}{3}\Bigr),
\end{align*}
and noting that $P_\Theta=0$ on $\mathcal{N}$, the linearisation decouples into
radial and angular blocks. The relevant second derivatives of the potential at
$\vc:=(\delta,0)$ read
\begin{align*}
	\mathcal V_{RR}(\vc)=-2\alpha,
	\quad
	\mathcal V_{\Theta\Theta}(\vc)
	=\delta+\frac{\alpha}{2}\delta^{2},
	\quad
	\mathcal V_{R\Theta}(\vc)=0.
\end{align*}
Hence, the Jacobian matrix takes the block--diagonal form
\begin{align*}
	J(\vGamma_*)
	=
	\begin{pmatrix}
		0                     & 0                               & \dfrac{1}{\mathcal{M}_*} & 0              \\[0.15cm]
		0                     & 0                               & 0                        & \dfrac{1}{D_*} \\[0.15cm]
		-\mathcal V_{RR}(\vc) & 0                               & 0                        & 0              \\[0.15cm]
		0                     & -\mathcal V_{\Theta\Theta}(\vc) & 0                        & 0
	\end{pmatrix}.
\end{align*}

The spectrum therefore splits into a radial pair and an angular pair. The radial
block yields
\begin{align}
	\lambda^{2}-\frac{2\alpha}{\mathcal{M}_*}=0
	\quad\Longrightarrow\quad
	\lambda_{1,2}=\pm\omega_0,\qquad \text{with}\qquad 	\label{eq:omega}
	\omega_0
	=
	\sqrt{\frac{2\alpha}{\mathcal{M}}}
	=
	\sqrt{\frac{2\alpha}{1+\mu+\alpha}}.
\end{align}
Thus, for physical parameters $\alpha>0$, the equilibrium is hyperbolic in the
radial directions. The angular block gives
\begin{align}
	\lambda^{2}+\Omega_0^{2}=0,	\quad\Longrightarrow\quad 	\lambda_{3,4}=\pm\rmi\Omega_0, \qquad \text{with}
	\qquad
	\Omega_0^{2}
	=
	\frac{Q(\delta)}{D_*}
	=
	\frac{1+\frac{\alpha}{2}\delta}
	{\delta\left(1+\frac{\alpha\delta}{3}\right)},
\end{align}
Consequently, the equilibrium $\vGamma_*$ is of saddle--centre type
(hyperbolic$\times$elliptic) and therefore linearly unstable for $\alpha>0$.

In the massless--string limit $\alpha=0$, the reduced potential simplifies to
$\mathcal V(R,0)=R(\mu-1)$. In this case the condition
$\partial_R\mathcal V(R,0)=0$ is satisfied only for $\mu=1$. Hence, for
$\alpha=0$ and $\mu\neq1$ the reduced subsystem exhibits constant acceleration
and admits no equilibrium, whereas for $\alpha=0$ and $\mu=1$ it possesses a
neutrally stable continuum of stationary solutions $R=\mathrm{const}$ on~$\mathcal{N}$.

\subsection{Non-stationary particular solutions}

The restriction of the Hamiltonian system~\eqref{eq:HH} to the invariant
manifold $\mathcal{N}$ yields a reduced Hamiltonian subsystem with
Hamiltonian
\begin{align}
	\label{eq:H_M}
	\widetilde{	\mathcal{H}}=		\mathcal{H}_{|\mathcal{N}}
	=
	\frac{P_R^{2}}{2\mathcal{M}}+\mathcal{V}(R,0).
\end{align}
Fixing the energy level $\widetilde{	\mathcal{H}}=E$, the radial motion is governed by
the first integral
\begin{align}
	\label{eq:ee}
	\dot R^{2}
	=
	\frac{2\bigl(E-\mathcal{V}(R,0)\bigr)}{\mathcal{M}},
\end{align}
which defines a one-parameter family of non-stationary solutions of the reduced
system~\eqref{eq:v_M}.

Let $R(\tau)$ be any non-constant solution of~\eqref{eq:v_M}. Then
\begin{align}
	\label{eq:non_stat_part_sol}
	\vGamma(\tau)
	:=
	\bigl(R(\tau),\,0,\,P_R(\tau),\,0\bigr),
\end{align}
defines a non-stationary particular solution of the full Hamiltonian system,
along which the Morales--Ramis theory will be applied. The explicit form of
$R(\tau)$ depends on whether the parameter $\alpha$ vanishes or not.

For $\alpha\neq 0$, i.e.\ for a string with non-vanishing linear mass density
$\lambda\neq 0$, the amount of mass involved in the motion depends explicitly on
the configuration. As a result, the effective inertia of the system increases
with the instantaneous length $R(\tau)$ of the swinging branch, which
qualitatively modifies the radial dynamics. In this case, solutions of the
reduced equations~\eqref{eq:v_M} take the hyperbolic form
\begin{align}
	\label{eq:rr_1}
	R(\tau)
	=
	A\cosh\!\bigl[\omega_0 (\tau-\tau_0)\bigr]+\delta,
	\quad
	A=\sqrt{\delta^2-\frac{E}{\alpha}},
\end{align}
where $\omega_0>0$ is the characteristic growth rate defined
in~\eqref{eq:omega}, and $\delta$, given by~\eqref{eq:delta}, denotes the
location of the stationary (but unstable) radial equilibrium. The constant $A$
is determined by the initial conditions and measures the initial deviation from
the equilibrium~$R=\delta$.

In contrast, for $\alpha=0$ (the massless-string limit), the reduced
equation~\eqref{eq:v_M} coincides with the radial equation of the classical
Atwood machine with constant acceleration. Integrating twice, one obtains
\begin{align}
	\label{eq:alpha0}
	R(\tau)
	=
	-\frac{1}{2}a\,\tau^2+V_0\,\tau+R_0,
	\qquad
	a=\frac{\mu-1}{\mu+1},
\end{align}
where $V_0$ and $R_0$ are integration constants corresponding to the initial
velocity and position, respectively.

It is worth emphasizing that the limit $\alpha\to 0$ is singular at the level of
the reduced dynamics. For $\alpha\neq 0$ the radial motion is governed by an
effective quadratic potential with an unstable stationary point $R=\delta$,
whereas for $\alpha=0$ the potential becomes purely linear and the stationary
point disappears. In particular, the characteristic hyperbolic time scale
$\omega_0^{-1}\sim\sqrt{\frac{\mathcal{M}}{2\alpha}}$ diverges as $\alpha\to 0$,
and the hyperbolic solution~\eqref{eq:rr_1} does not converge uniformly on long
time intervals to the parabolic law~\eqref{eq:alpha0}.

\section{Poincar\'e sections}
\begin{figure}[t]
	\centering
	\subfigure[$\alpha=0,\, E=1$]{
		\includegraphics[width=0.32\linewidth]{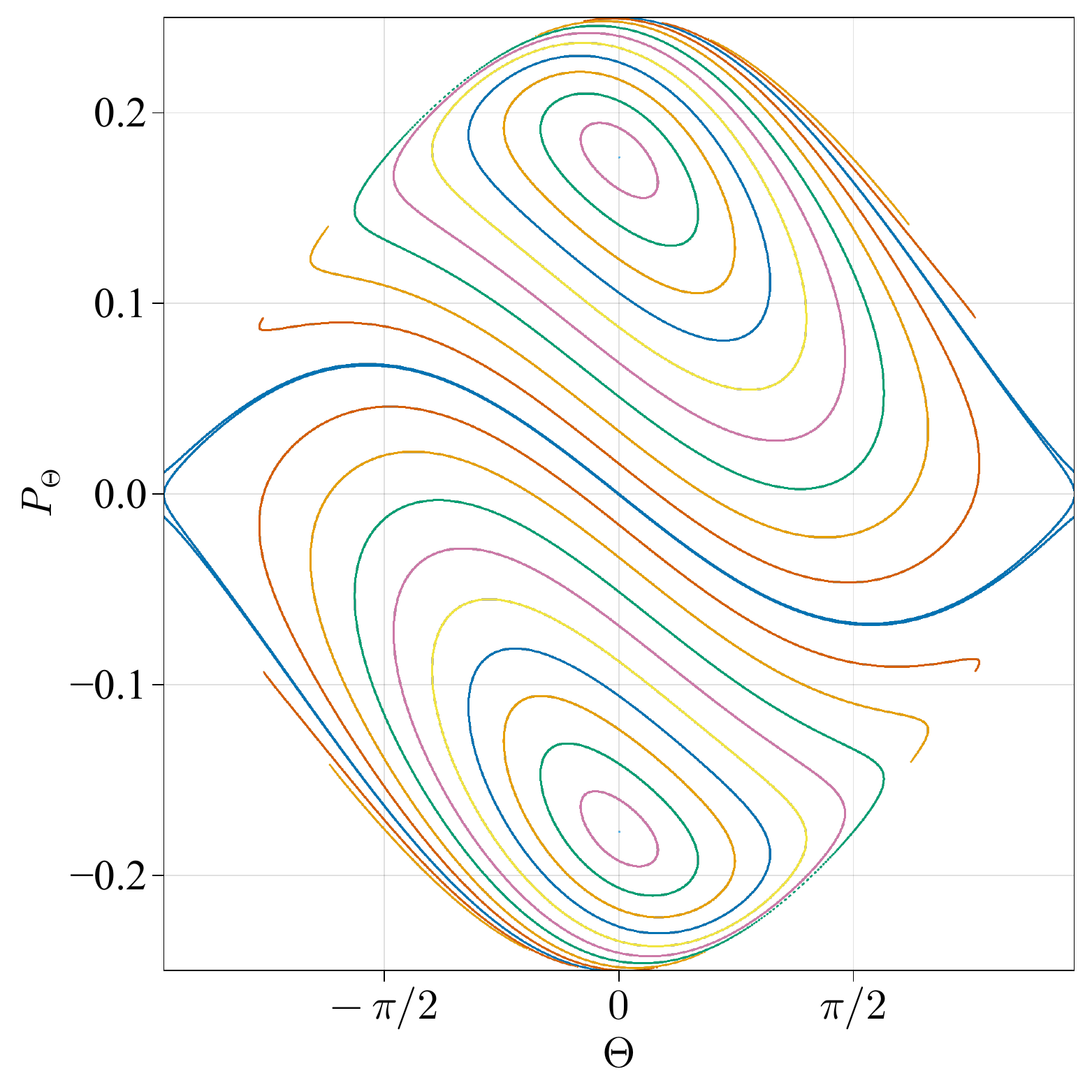}  }
	\subfigure[$\alpha=0.01\, E=1$]{
		\includegraphics[width=0.32\linewidth]{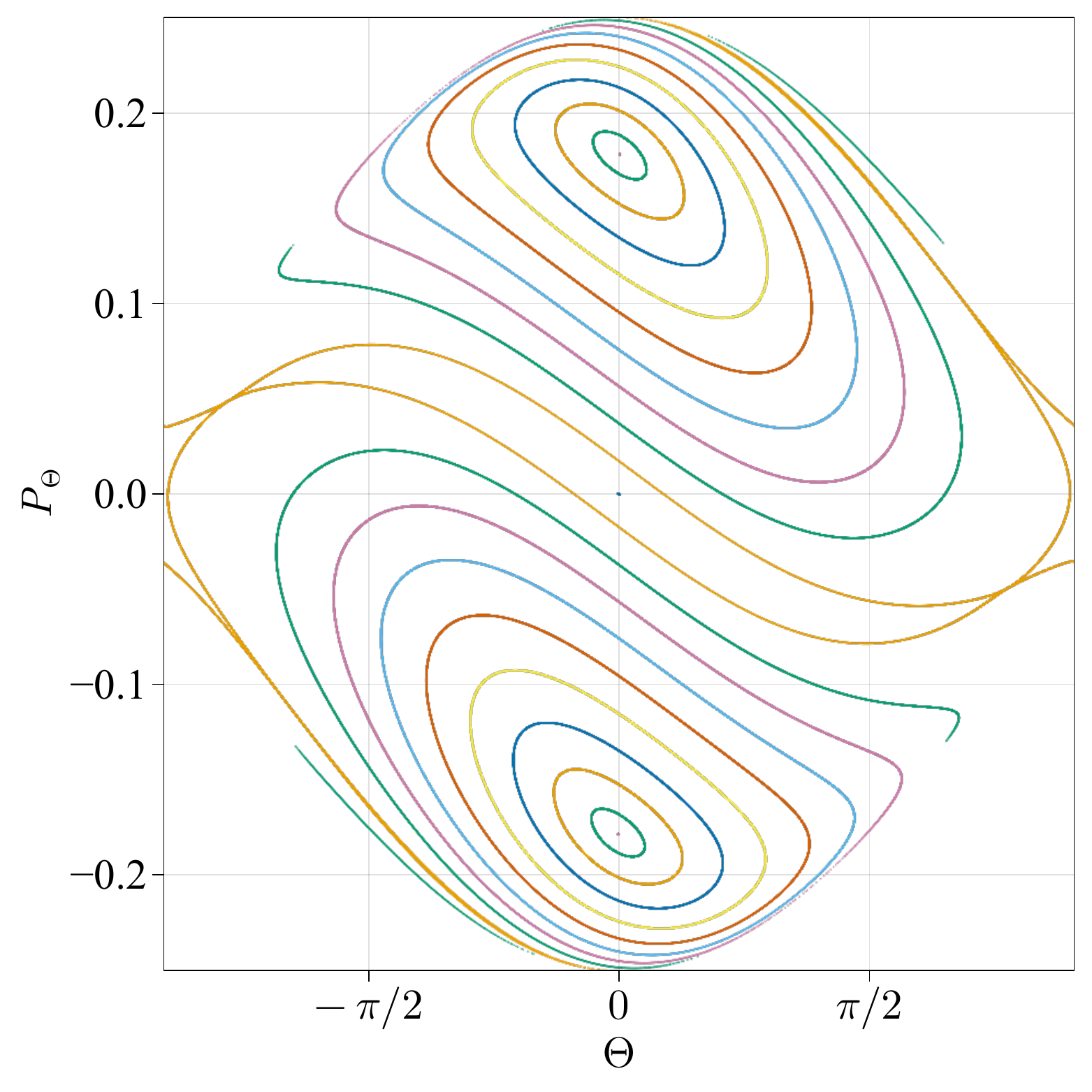}  }
	\subfigure[$\alpha=0.015,\, E=1.002$]{
		\includegraphics[width=0.32\linewidth]{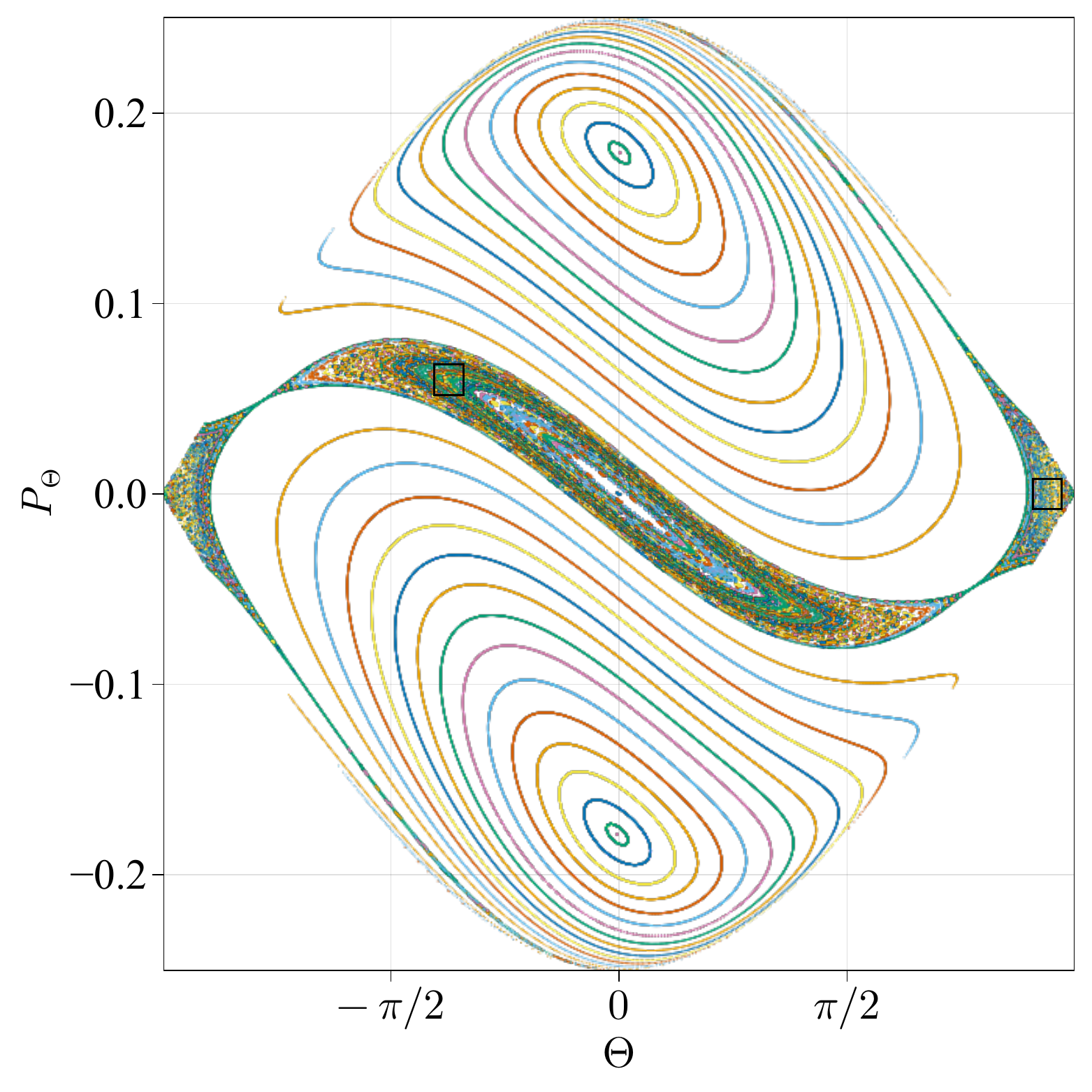}}
	\\
	\subfigure[$\alpha=0.5\, E=1.006$]{\includegraphics[width=0.31\linewidth]{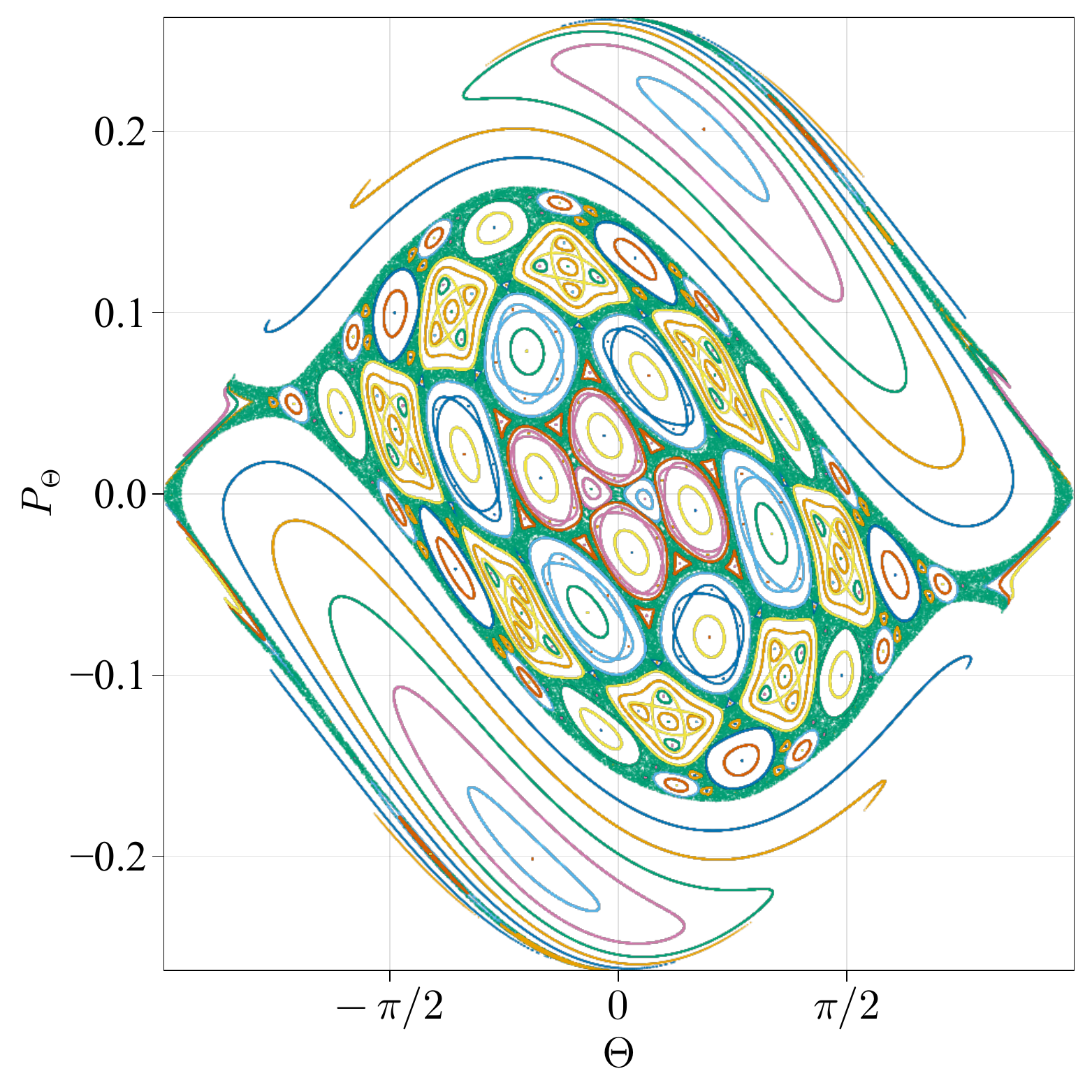}  }\hspace{5pt}
	\subfigure[$\alpha=1.5,\, E=1.19$]{
		\includegraphics[width=0.31\linewidth]{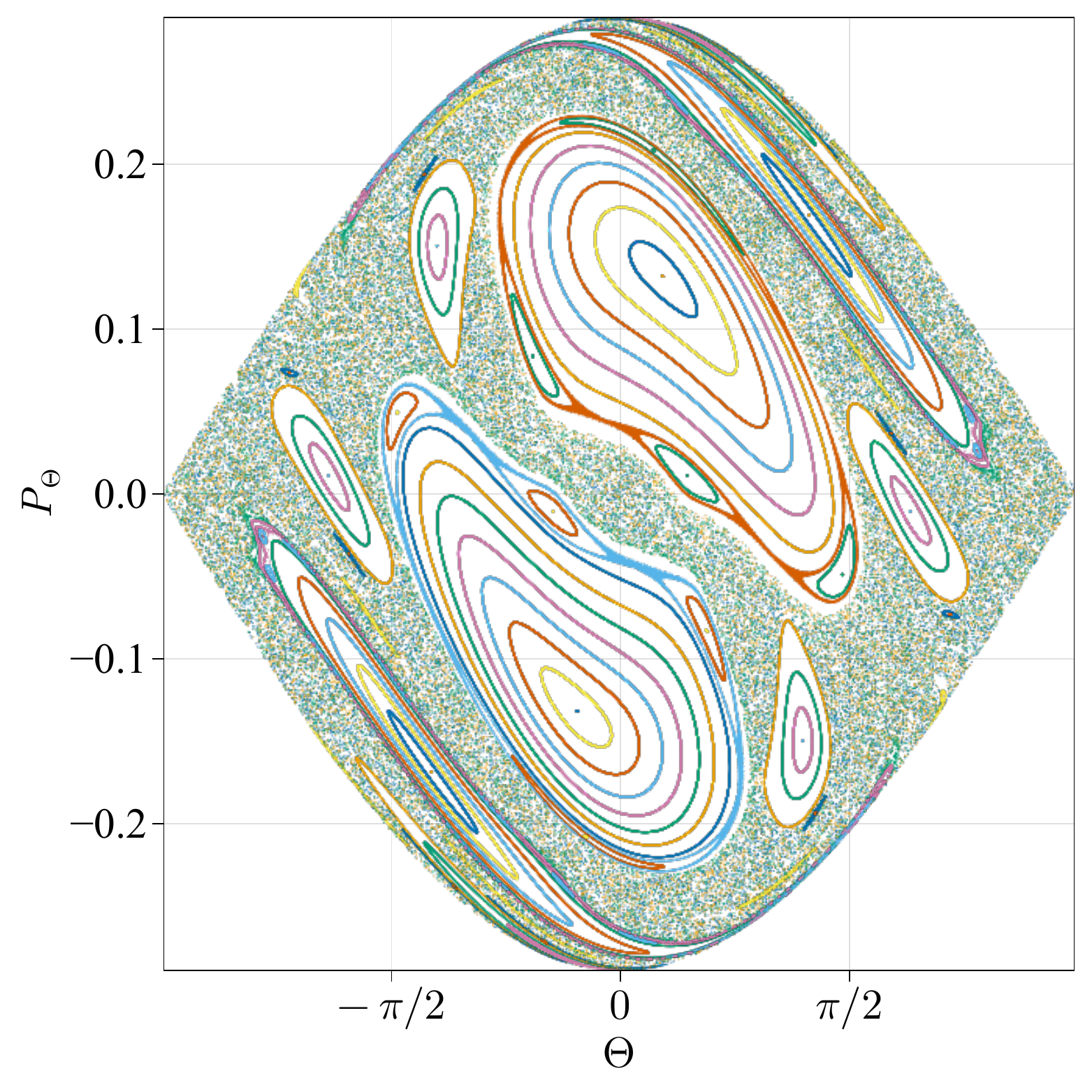}
	}\subfigure[$\alpha=1.656,\, E=1.207$]{
		\includegraphics[width=0.31\linewidth]{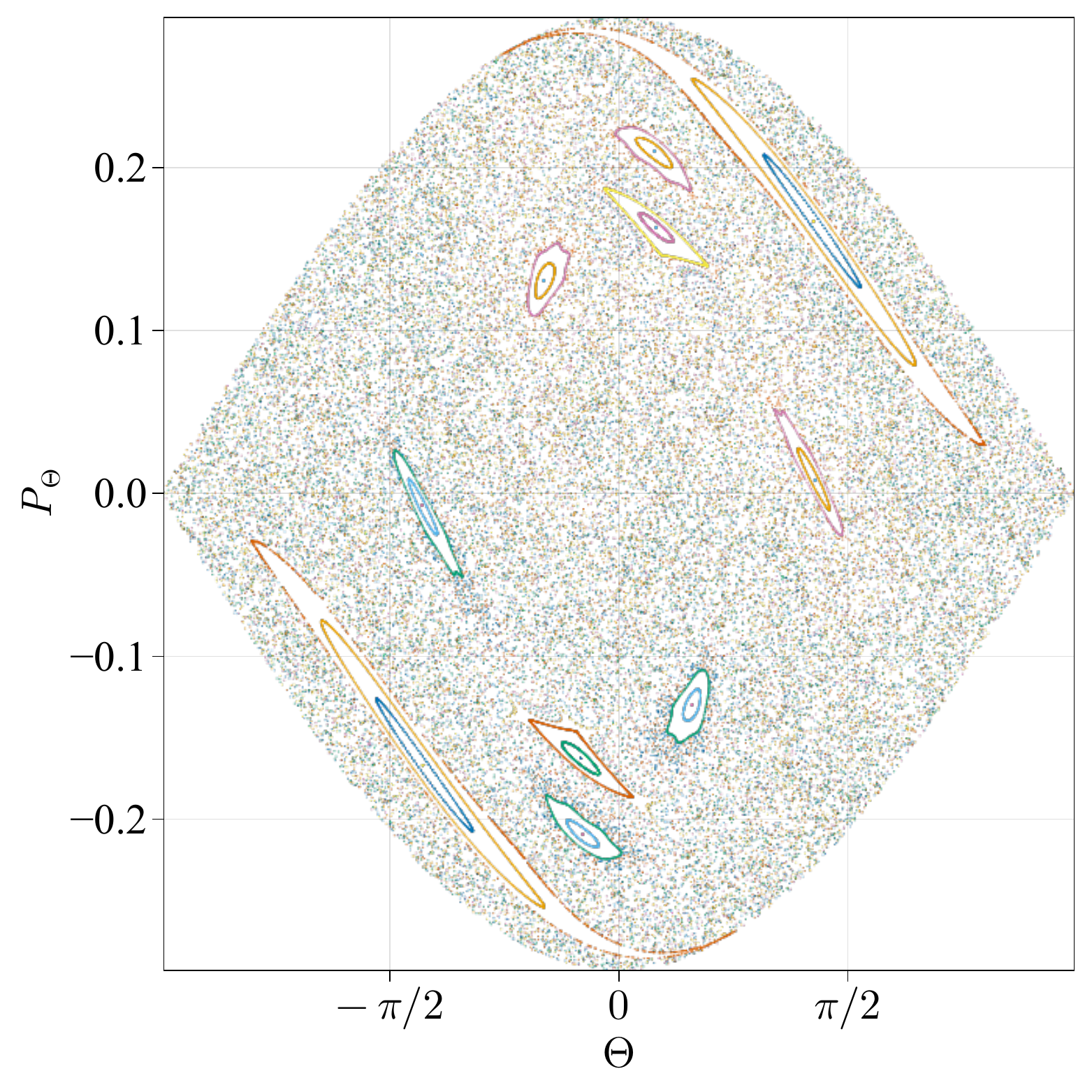}
	}

	\caption{The Poincare sections for the system~\eqref{eq:Ham_sys}, prepared for the surface defined by $ R=\eta/2$,
		with parameters set to $\eta=0.5$, $\mu=3$ and energy level given by $E=H( \eta/2,\pi,0,0)$, for varying values of the dimensionless mass ratio $\alpha $.
		Areas highlighted by the bounding boxes presented on Fig.~\ref{fig:psections}(c) can be examined in more detail on Fig.~\ref{fig:pszooms}.
		\label{fig:psections}}
\end{figure}

\begin{figure}[t]
	\centering
	\includegraphics[width=0.4\linewidth]{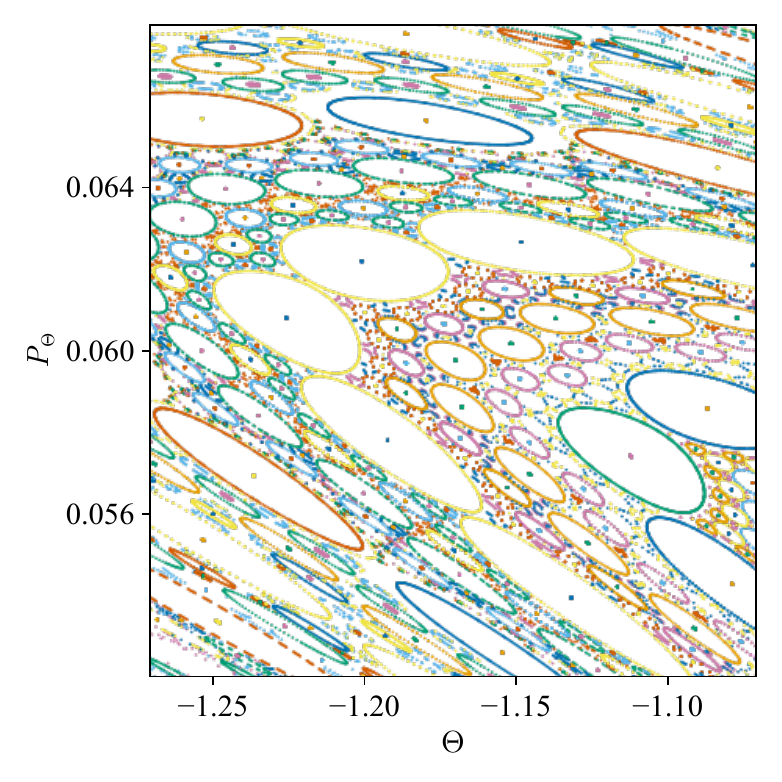}\hspace{5pt}	\includegraphics[width=0.4\linewidth]{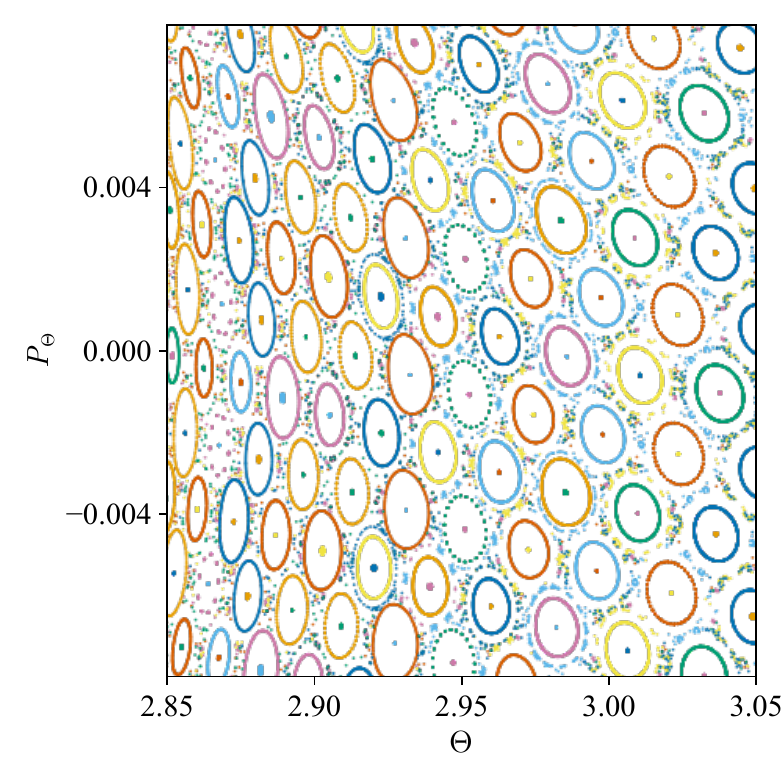}
	\caption{Magnification of chosen areas from the central part of the Poincare section presented on Fig.~\ref{fig:psections}(c).}
	\label{fig:pszooms}
\end{figure}

To visualize the structure of the phase space and to detect the presence of
regular islands, resonances, and chaotic layers, we construct Poincar\'e
sections on suitably chosen transversal hypersurfaces.

We start by employing Poincar\'e sections for a fixed value of the parameters
and the energy levels given by
\[\eta=\frac{1}{2},\quad \mu=3,\qquad E=H(\eta/2,\pi,0,0),\]
with varying values of the dimensionless mass ratio $\alpha $.
These values of the parameter $\alpha$ will provide a comprehensive view of the system dynamics when the mass of the string matters.

We proceed as follows. For fixed values of the dimensionless parameters
$\mu$, $\eta$ and $\alpha$, we consider the Poincar\'e section defined by
$R= \eta/2$, with the direction $P_R>0$.
Here, $ R=\eta/2$ is a fixed radial coordinate corresponding to
the midpoint of the swinging branch. We then numerically integrate a large
ensemble of trajectories with initial conditions sampled from
the energy surface $E=H( \eta/2,\pi,0,0)$ and collect the intersection points of
these trajectories. The resulting Poincar\'e sections are
plotted in the $(\Theta,P_\Theta)$ plane, revealing the impact of
the perturbation parameter $\alpha$ on the phase space structure.

For $\alpha=0$, the system reduces to the classical Swinging Atwood Machine with
a massless string, which for $\mu=3$ is known to be integrable.
On Fig.~\ref{fig:psections}(a) we can see the Poincar\'e section for this case.
The phase space is filled with invariant curves corresponding to quasi-periodic
motion on two-dimensional tori. The absence of chaotic regions and the presence
of smooth, closed curves are indicative of integrable dynamics. The particular solution of the classical Swinging Atwood Machine corresponds to a periodic orbit, which appears as a pair of fixed points in the Poincar\'e section.

When a small perturbation is introduced by setting $\alpha=0.01$, as shown in Fig.~\ref{fig:psections}(b), a new hyperbolic orbit, corresponding to the radial solution~\eqref{eq:rr_1}, emerges at the centre of the Poincar\'e section. The orbit is enclosed by closed invariant curves forming a separatrix, which divides distinct regions of the phase space. A thin chaotic layer develops in its neighbourhood, providing the first numerical evidence that the separatrix is associated with the onset of chaotic dynamics.

In Fig.~\ref{fig:psections}(c), corresponding to $\alpha=0.015$, the phase space undergoes a substantial qualitative change. A proliferation of periodic orbits is accompanied by the emergence of chaotic regions. Nevertheless, the radial solution~\eqref{eq:rr_1} remains visible in the central part of the section, where it is surrounded by numerous newly formed periodic orbits. A similar structure is observed in the neighbourhood of $(\pi,0)$, where the periodic islands are associated with the rotational motion of the pendulum. For clarity, two enlarged views of these regions are shown in Fig.~\ref{fig:pszooms}.

For $\alpha=0.5$, as shown in Fig.~\ref{fig:psections}(d), the chaotic layer associated with the separatrix expands considerably, occupying a large portion of the Poincar\'e section. This reflects the increasing influence of the string inertia on the global dynamics of the system. Nevertheless, numerous stable periodic islands persist, including those corresponding to high-order resonances.

For larger values of the mass ratio $\alpha\geq 3/2$,  the system reveals highly chaotic dynamics, with chaotic orbits
filling most of the section plane, except for some regular islands. However, the
previously mentioned central periodic orbit persists. Nevertheless, we conclude
that for these parameter values with $\alpha\neq 0$, the system does not appear
to be integrable. A rigorous proof is presented in Section~7.

\section{Lyapunov exponents dynamical maps}

The qualitative analysis based on Poincar\'e sections presented in the previous section strongly indicates that the system~\eqref{eq:Ham_sys} is non-integrable. In particular, the emergence of stochastic layers and scattered intersection points on the section provides clear evidence of chaotic trajectories in the phase space. While Poincar\'e sections constitute a powerful geometric tool for detecting the breakdown of invariant tori and the onset of chaotic behavior, they do not provide quantitative information about the strength of chaos, nor do they allow for a systematic comparison across different regions of the space of initial conditions.

To obtain a more refined and quantitative insight into the system's dynamics, we employ the theory of Lyapunov exponents. These quantities measure the average exponential rate of divergence of nearby trajectories in phase space. More precisely, for two initially close trajectories, their separation typically evolves as
\begin{align*}
	\|\delta \mathbf{x}(t)\| \sim \|\delta \mathbf{x}(0)\| e^{\lambda t},
\end{align*}
where $\lambda$ denotes the largest Lyapunov exponent. The sign and magnitude of $\lambda$ carry essential dynamical information: negative values correspond to contraction, zero values to neutral (typically quasi-periodic) motion, whereas positive values indicate sensitive dependence on initial conditions and hence chaotic dynamics.

In Hamiltonian systems, the spectrum of Lyapunov exponents exhibits a characteristic structure imposed by the underlying symplectic geometry. In particular, the exponents occur in symmetric pairs $(\lambda_i,-\lambda_i)$, for $i=1,\ldots, n$, and their sum vanishes as a consequence of phase-space volume preservation (Liouville's theorem). Moreover, the existence of first integrals further constrains the spectrum by introducing additional zero exponents~\cite{Pikovsky:16::}. From the perspective of integrability, this observation is crucial: integrable systems are typically characterized by vanishing Lyapunov exponents almost everywhere, reflecting the regular (quasi-periodic) motion on invariant tori. In contrast, the presence of positive Lyapunov exponents signals the destruction of such structures and provides strong numerical evidence for non-integrability. Recent applications of Lyapunov exponents to pendulum-like systems in the systematic search for additional first integrals can be found in~\cite{Szuminski:24::,Szuminski:25::JSV_VLDP,Szuminski:23::}. A detailed description of the proposed methodology, together with several benchmark applications, is presented in the forthcoming paper~\cite{Szuminski:26::}.

Since the system under consideration evolves in a four-dimensional
phase space, and since the objective of the present work is to
distinguish between regular and chaotic motion, it is sufficient to
estimate the largest Lyapunov exponent. To this end, we employ the
standard algorithm of Benettin \emph{et al.}~\cite{Benettin:80::},
which is based on the simultaneous integration of the equations of
motion and the associated variational equations, combined with
periodic Gram--Schmidt re-orthonormalization of the deviation vectors.

Let \(T\) denote the re-orthonormalization interval and let
\(r^{(k)}\) be the norm of the deviation vector immediately before the
\(k\)-th re-orthonormalization step. The finite-time approximation of
the largest Lyapunov exponent is then given by
\begin{align*}
	\lambda^{(N)}
	=
	\frac{1}{NT}
	\sum_{k=1}^{N}
	\ln r^{(k)},
\end{align*}
where \(N\) denotes the total number of re-orthonormalization steps.
As \(N\) increases, this quantity converges to the largest Lyapunov
exponent.

The equations of motion and the associated variational equations were
integrated simultaneously using the built-in \textit{Mathematica 15}
solver \texttt{NDSolve} with an explicit Runge--Kutta method.
The numerical parameters were selected after extensive testing to
ensure both stability and computational efficiency.
Throughout all computations, the re-orthonormalization interval was
fixed at \(T=1\), and the largest Lyapunov exponent was evaluated over
\(N=3000\)--\(5000\) successive re-orthonormalization steps,
corresponding to total integration times
$
	t=NT\in[3000,5000].
$
All computations were performed using machine-precision arithmetic.
These settings were found to provide reliable convergence of the
largest Lyapunov exponent while maintaining a reasonable computational
cost.

Additionally, the conservation of the energy integral $\mathcal{H}=E$ defined in~\eqref{eq:HH} is used as an independent measure of numerical accuracy. Both relative and absolute errors remain below the tolerance of $10^{-11}$, ensuring the numerical accuracy of the computed trajectories. These settings enable efficient and reliable computation of Lyapunov exponents, providing reliable convergence throughout the computations.

Having established the numerical procedure for the computation of the largest Lyapunov exponent, we now employ it to investigate the global organization of regular and chaotic dynamics in the considered system. To this end, we construct Lyapunov diagrams in several complementary settings: the parameter space, the space of initial conditions, and the energy-constrained configuration space. Each of these representations reveals different aspects of the underlying dynamical structure.

For visualization purposes, a logarithmic color scale is applied,
\begin{equation}
	\label{eq:scale}
	\text{scaling}(\lambda) \;:=\; \left[\frac{\log(\lambda/\min)}{\log(\max/\min)}\right]^p,
\end{equation}
where $\min$ and $\max$ denote the minimal and maximal values of $\lambda$ on the grid, and $p\in\mathbb{R}^+$ is a contrast parameter.
This scaling enhances the visibility of weakly chaotic regions.
For the integration times used in the present work, the threshold
$
	\lambda_{\mathrm{thr}}=5\times10^{-3}
$
was found to provide a reliable numerical separation between regular
and chaotic trajectories. Consequently, trajectories satisfying
$\lambda>\lambda_{\mathrm{thr}}$ are classified as chaotic, whereas only
those with $\lambda\le\lambda_{\mathrm{thr}}$ are subjected to the
subsequent periodicity-detection procedure.

This approach enables not only the detection of chaotic regions, but also a quantitative assessment of their intensity, as well as the identification of fine structures such as resonance zones, chaotic layers, and sharp boundaries separating distinct dynamical regimes. In particular, Lyapunov diagrams provide a natural quantitative complement to Poincar\'e sections, offering a global characterization of the phase space beyond purely geometric observations.

\subsection{Lyapunov maps without energy constraint}

In this subsection, the global dynamics of the heavy swinging Atwood machine is investigated using Lyapunov exponent maps constructed without imposing an energy constraint. Two complementary representations are employed. The first examines the parameter plane $(\mu,\alpha)$ for prescribed initial conditions, while the second focuses on the initial-condition plane $(R_0,\Theta_0)$ for fixed system parameters. Together, these complementary perspectives provide a comprehensive description of how the global dynamics depends on both the governing parameters and the initial conditions.

\subsubsection{Dynamical maps in the $(\mu,\alpha)$ parameter space}
\begin{figure}[t]
	\centering
	\subfigure[ $\Theta_0=\pi/100$]{\includegraphics[width=0.44\linewidth]{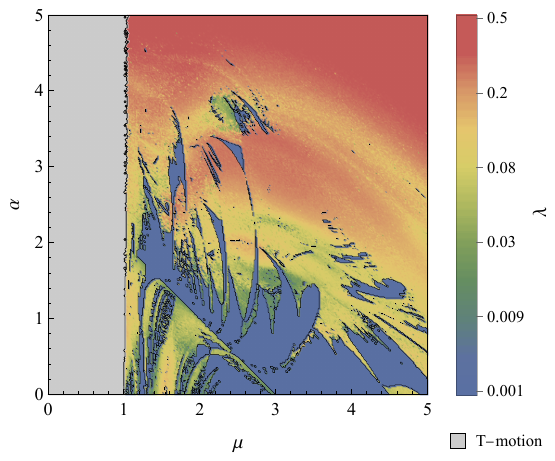}}\hspace{5pt}
	\subfigure[ $\Theta_0=\pi/4$]{\includegraphics[width=0.44\linewidth]{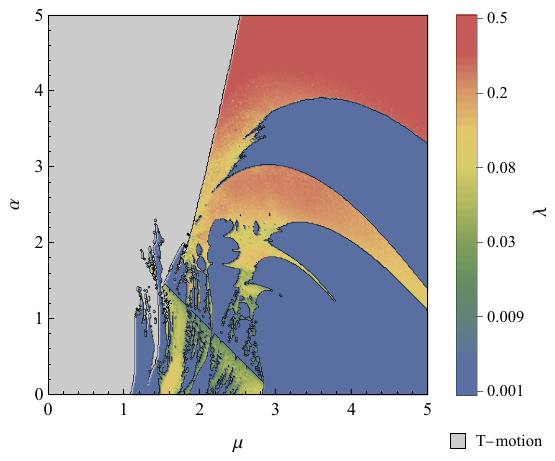}}
	\\
	\subfigure[ $\Theta_0=\pi/2$]{\includegraphics[width=0.44\linewidth]{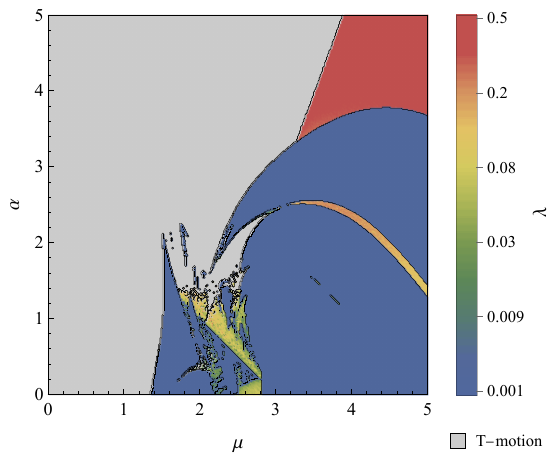}}\hspace{5pt}
	\subfigure[ $\Theta_0=3\pi/4$]{\includegraphics[width=0.44\linewidth]{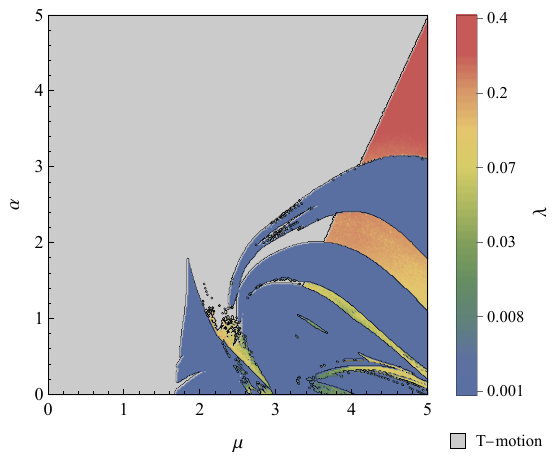}}
	\caption{(Color online) Two-dimensional dynamical maps in the
		$(\mu,\alpha)$-plane for system~\eqref{eq:Ham_sys} with fixed
		$\eta=0.5$, corresponding to the parameter-space sampling
		\eqref{eq:mu_alpha_diag}. The largest finite-time Lyapunov exponent is
		computed on a $400\times400$ grid using the initial conditions
		\eqref{eq:ini_lyap} and displayed on a logarithmic color scale. Blue regions correspond to numerically regular motion, colored regions to chaotic dynamics, and gray regions to terminating trajectories (T-motion).}
	\label{fig:parki}
\end{figure}
We fix $\eta=0.5$ and consider initial conditions of the form
\begin{align}
	\label{eq:ini_lyap}
	R(0)=\frac14,\qquad
	\Theta(0)=\Theta_0,\quad \text{where} \quad \Theta_0\in
	\left\{
	0.01\pi,\,
	0.25\pi,\,
	0.5\pi,\,
	0.75\pi
	\right\}, \qquad
	P_R(0)=P_\Theta(0)=0.
\end{align}

For each fixed value of $\Theta_0$, the largest Lyapunov exponent is
computed for every pair $(\mu,\alpha)$, provided that the corresponding
trajectory remains within the physically admissible region.

Accordingly,   we define the dynamically admissible region in the
parameter plane
\begin{align*}
	\mathcal D_{\mu,\alpha}(\Theta_0)
	:=
	\Bigl\{
	(\mu,\alpha)\in\mathbb R_+^2
	\;:\;
	0<R(\tau;\mu,\alpha,\Theta_0)<\eta
	\quad
	\text{for all }\tau\ge0
	\Bigr\}.
\end{align*}
The set $	\mathcal D_{\mu,\alpha}(\Theta_0)$
contains all parameter values for which the corresponding
trajectory remains bounded away from the geometric constraints
$R=0$ and $R=\eta$.

The corresponding parameter-space map is therefore defined by
\begin{align}
	\label{eq:mu_alpha_diag}
	(\mu,\alpha)
	\longmapsto
	\begin{cases}
		\lambda(\mu,\alpha),
		 &
		(\mu,\alpha)\in\mathcal D_{\mu,\alpha}(\Theta_0),
		\\[2mm]
		\mathrm{T\mbox{-}motion},
		 &
		(\mu,\alpha)\notin\mathcal D_{\mu,\alpha}(\Theta_0).
	\end{cases}
\end{align}
Its boundary separating bounded and terminating trajectories is given by
\begin{align*}
	\partial\mathcal D_{\mu,\alpha}(\Theta_0)
	:=
	\Bigl\{
	(\mu,\alpha)\in\mathbb R_+^2
	\;:\;
	\inf_{\tau\ge0}
	R(\tau;\mu,\alpha,\Theta_0)=0
	\;\;\text{or}\;\;
	\sup_{\tau\ge0}
	R(\tau;\mu,\alpha,\Theta_0)=\eta
	\Bigr\}.
\end{align*}
Whenever the dependence on the fixed initial angle $\Theta_0$ is clear
from the context, we simply write
$\mathcal D_{\mu,\alpha}$ and
$\partial\mathcal D_{\mu,\alpha}$.

The boundary $\partial\mathcal D_{\mu,\alpha}$  exhibits a highly intricate geometry and strong sensitivity to parameter variations, separating regions of bounded motion from those leading to escape. Although arising in an entirely different mathematical setting, this boundary plays
a conceptually similar role to the Mandelbrot set in complex dynamics~\cite{Mandelbrot:80::}, as both organize the parameter space by separating regions associated with qualitatively different dynamical behaviors.

The dynamical maps visible in Fig.~\ref{fig:parki} visualize three dynamical regimes: regular, chaotic, and terminating.
Let us focus on the first map in Fig.~\ref{fig:parki}(a).
One observes that the terminating motion (T-motion) occupies a substantial portion of the parameter space, covering approximately $20\%$ of the parameter domain $(\mu,\alpha)\in[0,5]\times[0,5]$.
This region is separated from the bounded dynamics by a relatively sharp transition, which for small values of $\alpha$ is well approximated by a nearly vertical boundary located around $\mu \approx 1$.

This behavior can be understood by noting that for the chosen initial condition, namely $\Theta_0=\pi/100$ and $R_0=\eta/2$, the system evolves in a regime corresponding to a small perturbation of the invariant manifold~\eqref{eq:M_invariant_linstab}.
In particular, for $\mu<1$, the effective restoring mechanism is insufficient to confine the motion, and the trajectory is driven towards the geometric boundary $R(t)\to \eta$, leading to termination in finite time.
This is consistent with the well-established property of the classical swinging Atwood machine with a massless string ($\alpha=0$), where terminating motion occurs when the counterweight mass is less than or equal to the oscillating mass~\cite{Tufillaro:85::, Tufillaro:88a, Tufillaro:95::}.
Consequently, the line $\mu\approx 1$ can be interpreted as a threshold separating two qualitatively distinct dynamical regimes: for $\mu<1$, the dynamics is dominated by escape towards the boundary, while for $\mu>1$, bounded motion becomes admissible and further subdivides into regular and chaotic regions depending on the value of~$\alpha$.

The parameter $\alpha$ controls the effective nonlinearity of the system and has a pronounced influence on the phase-space structure. As $\alpha$ increases, the chaotic regions gradually expand at the expense of the regular domains, while the largest Lyapunov exponent increases, reaching values up to $\lambda\approx0.55$ within the considered parameter range. This demonstrates that the string's inertia substantially enhances the dynamics' chaotic character. Moreover, the chaotic regions are not distributed uniformly; instead, they form elongated structures alternating with extended domains of regular dynamics. As will be shown in the next section, these regular domains are populated by organized families of periodic trajectories, revealing a rich resonance structure that is not visible in the Lyapunov maps alone.

The remaining panels in Fig.~\ref{fig:parki} illustrate how the
parameter-space structure changes with increasing initial angle
$\Theta_0$. For $\Theta_0=\pi/4$, shown in
Fig.~\ref{fig:parki}(b), the dynamics is no longer a small perturbation
of the invariant manifold~\eqref{eq:M_invariant_linstab}. Consequently,
the transition near $\mu\approx1$ becomes strongly deformed, and the
boundary $\partial\mathcal D_{\mu,\alpha}$ develops a much more complex  pattern.

These features become even more pronounced for
$\Theta_0=\pi/2$, shown in Fig.~\ref{fig:parki}(c). The admissible
region fragments into several disconnected components, while islands of
bounded motion appear inside the terminating region. At the same time,
the portion of the parameter space supporting chaotic bounded motion
decreases, whereas terminating motion becomes increasingly dominant.

Finally, for $\Theta_0=3\pi/4$, shown in
Fig.~\ref{fig:parki}(d), terminating motion occupies approximately
$70\%$ of the parameter domain, confining bounded trajectories to
narrow, highly structured regions. A particularly noteworthy feature is
the appearance of a large closed terminating region surrounded by bounded
motion, illustrating the increasingly complex geometry of the admissible
set $\mathcal D_{\mu,\alpha}$.

Overall, the geometry of the admissible set
$\mathcal D_{\mu,\alpha}$ depends strongly on the initial angle
$\Theta_0$, leading to qualitatively different organizations of
regular, chaotic, and terminating regions in the parameter space.
Higher values of $\Theta_0$ correspond to larger initial potential energy, which in turn allows the swinging mass to reach the geometric boundary $R=\eta$, resulting in a larger portion of the parameter space being occupied by terminating motion. As $\alpha$ increases, this effect is further amplified, as the inertia of the string enhances the ability of the mass to reach the boundary. Consequently, the interplay between $\Theta_0$ and $\alpha$ plays a crucial role in shaping the global dynamics of the system.

\subsubsection{Dynamical maps in the $(R_0,\Theta_0)$ initial-condition space}

The parameter-space maps presented above describe how the dynamics changes with the system parameters. We now complement this analysis by constructing Lyapunov maps in the initial-condition plane for fixed parameter values. The maps reveal the internal organization of the phase space, distinguishing regions of regular, chaotic, and terminating motion together with their mutual arrangement.

We now construct Lyapunov maps in the initial-condition space
$(R_0,\Theta_0)$ for the fixed parameter values
\begin{align*}
	\mu=3,\qquad
	\eta=\frac12,\qquad
	\alpha\in\left\{0,\frac12,1,2\right\},
\end{align*}
using the initial conditions
\begin{align*}
	R(0)=R_0,\qquad
	\Theta(0)=\Theta_0,\qquad
	P_R(0)=P_\Theta(0)=0,
	\qquad
	(R_0,\Theta_0)\in(0,\eta)\times\mathbb S^1.
\end{align*}

We define the dynamically admissible region in the initial-condition plane by
\begin{align*}
	\mathcal D_{R,\Theta}(\alpha)
	:=
	\Bigl\{
	(R_0,\Theta_0)\in(0,\eta)\times \mathbb{S}^1
	\;:\;
	0<R(\tau; R_0, \Theta_0,\alpha)<\eta
	\quad\text{for all }\tau\ge0
	\Bigr\},
\end{align*}
consisting of those initial conditions for which the corresponding trajectory remains within the physically admissible region. Its boundary separating bounded and terminating trajectories is given by
\begin{align*}
	\partial\mathcal D_{R,\Theta}(\alpha)
	:=
	\Bigl\{
	(R_0,\Theta_0)\in(0,\eta)\times \mathbb{S}^1
	\;:\;
	\inf_{\tau\ge0}R(\tau; R_0, \Theta_0,\alpha)=0
	\;\;\text{or}\;\;
	\sup_{\tau\ge0}R(\tau; R_0, \Theta_0,\alpha)=\eta
	\Bigr\}.
\end{align*}
This boundary exhibits a highly intricate geometry and strong sensitivity to variations in the initial conditions $(R_0,\Theta_0)$, separating bounded trajectories from those that lead to termination. In this respect, $\partial\mathcal D_{R,\Theta}$ plays a role reminiscent of Julia
sets in complex dynamics~\cite{Julia:1918,Milnor2006}, as it separates initial conditions associated with qualitatively different long-term behaviors.

Figure~\ref{fig:lyap_polar1} shows Lyapunov-based dynamical maps in the
initial-condition space $(R_0,\Theta_0)$ for
system~\eqref{eq:Ham_sys}, defined by the mapping
\begin{align}
	\label{eq:IC_map}
	(R_0,\Theta_0)
	\longmapsto
	\begin{cases}
		\lambda(R_0,\Theta_0),
		 &
		(R_0,\Theta_0)\in\mathcal D_{R,\Theta}(\alpha),
		\\[1mm]
		\mathrm{T\mbox{-}motion},
		 &
		(R_0,\Theta_0)\notin\mathcal D_{R,\Theta}(\alpha),
	\end{cases}
\end{align}
where $\lambda(R_0,\Theta_0)$ denotes the largest Lyapunov exponent
computed along the corresponding trajectory.

These maps provide a detailed representation of the phase-space structure for fixed system parameters, allowing one to distinguish regions of regular motion, chaotic dynamics, and escape.
In particular, they illustrate how variations in the initial length $R_0$ and the initial swing angle $\Theta_0$ influence the system’s behavior.
The diagrams are obtained by numerical integration of the equations of motion~\eqref{eq:Ham_sys} on a $400\times 400$ grid of initial conditions $(R_0,\Theta_0)$, for which the largest Lyapunov exponent is computed whenever the trajectory remains in the admissible region.
For better visualization, the results are presented in a polar representation of the initial-condition plane, with the radial coordinate corresponding to the initial length $R_0$ and the angular coordinate to the initial swing angle $\Theta_0$. The maps reveal three types of dynamical behavior: regular, chaotic, and terminating.  Terminating motion (T-motion), corresponding to trajectories leaving the admissible region, is shown in gray;
blue regions indicate regular bounded motion, including periodic and quasi-periodic trajectories, while regions with $\lambda>0$ correspond to chaotic dynamics with a color scale defined in~\eqref{eq:scale}.

In Fig.~\ref{fig:lyap_polar1}(a), we present the Lyapunov-based dynamical map computed for the integrable case $\alpha=0$, corresponding to a massless string.
As expected, the system exhibits either terminating or purely regular dynamics, consisting of periodic and quasi-periodic trajectories.
No chaotic regions are visible, which is  consistent with the system's integrability.
Moreover, the boundary $\partial \mathcal D_{R,\Theta}$ separating terminating and bounded motion is smooth and well-defined, and the bounded region has a characteristic cup-like shape.

The situation changes significantly when the perturbation parameter is increased to $\alpha=0.5$, meaning that the total mass of the string is half of the swinging mass.
In this case, shown in Fig.~\ref{fig:lyap_polar1}(b), the first signs of chaotic behavior appear in the form of thin layers embedded within the regular domain.
These structures indicate the onset of nonlinear effects and the gradual breakdown of integrability.
The largest Lyapunov exponent reaches maximal values up to $\lambda \approx 0.15$, particularly for initial conditions corresponding to large swing angles.
\begin{figure}[t]
	\centering
	\subfigure[ $\alpha=0$]{
		\includegraphics[width=0.48\linewidth]{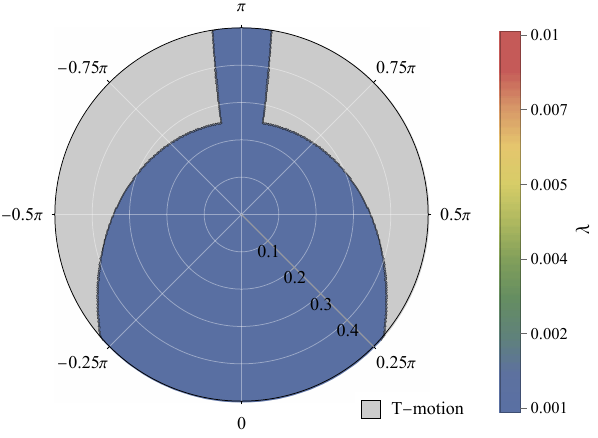}  }
	\subfigure[ $\alpha=0.5$]{\includegraphics[width=0.48\linewidth]{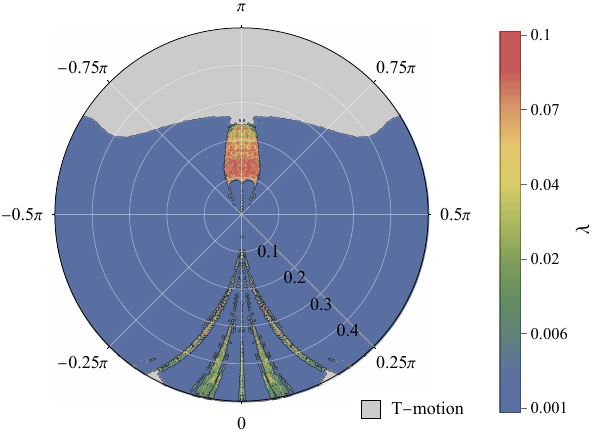}  } \\
	\subfigure[ $\alpha=1$]{\includegraphics[width=0.48\linewidth]{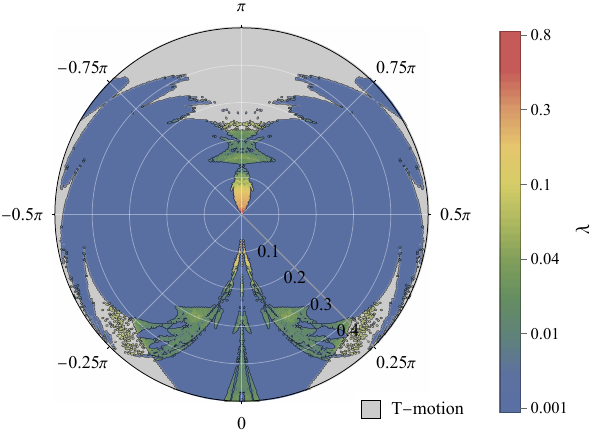}  }
	\subfigure[ $\alpha=2$]{
		\includegraphics[width=0.48\linewidth]{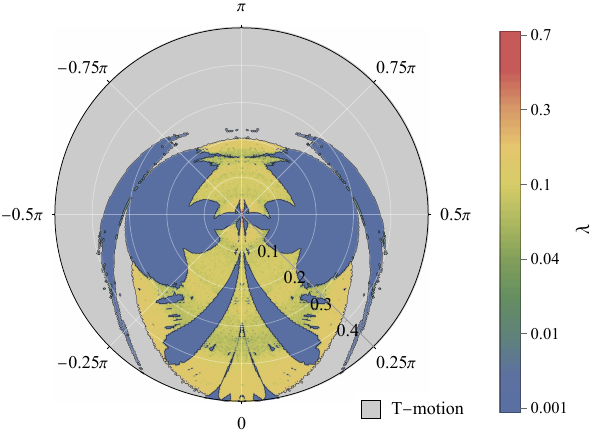}
	}	 \caption{(Color online) Lyapunov-based dynamical maps for
		system~\eqref{eq:Ham_sys} in the polar plane $(R_0,\Theta_0)$,
		constructed from the zero-velocity initial conditions
		\eqref{eq:ini_lyap} for different values of the string-to-bob mass
		ratio $\alpha$, with $\mu=3$ and $\eta=0.5$ fixed. The admissible
		domain is given by $R_0\in(0,\eta)$ and
		$\Theta_0\in(-\pi,\pi)$. The largest finite-time Lyapunov exponent is
		displayed on a logarithmic colour scale. Blue regions correspond to numerically regular motion, colored regions to chaotic dynamics, and gray regions to terminating trajectories (T-motion).
		\label{fig:lyap_polar1}}
\end{figure}

\begin{figure*}[htp]
	\centering
	\includegraphics[width=0.65
		\linewidth]{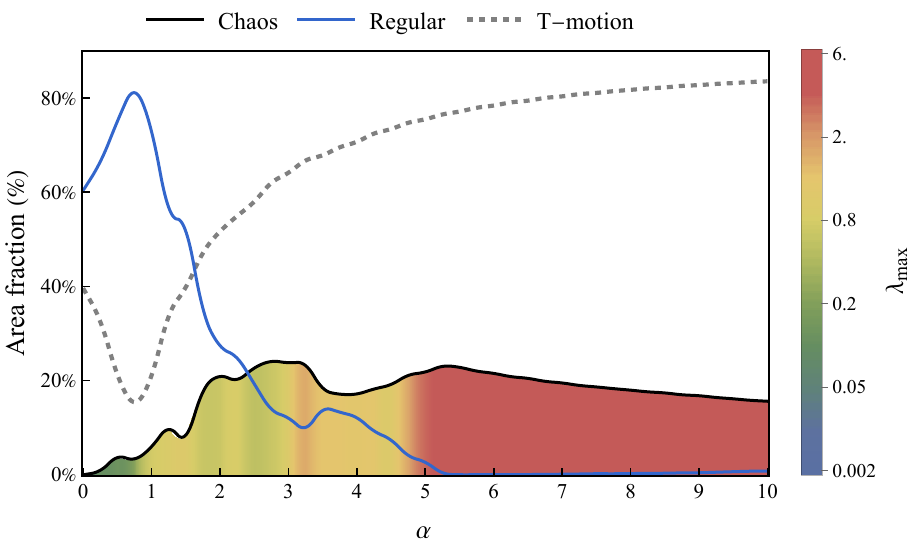}  	 \caption{
		(Color online) Fraction of the occupied area in the $(R_0,\Theta_0)$ initial-condition space as a function of $\alpha$. The black, blue, and dashed gray curves correspond to chaotic, regular, and terminating trajectories, respectively, while the shaded region under the chaotic curve represents the relative measure of chaos. The color encodes the maximal Lyapunov exponent $\lambda_{\max}$ according to the logarithmic scale shown on the right, thereby characterizing the strongest instability present in the system for a given value of $\alpha$. Consequently, the diagram simultaneously captures both the geometric extent and the dynamical intensity of chaos.
		\label{fig:percentage}}
\end{figure*}

A further increase to $\alpha=1$, where the total string mass equals the swinging mass, leads to a much more complex dynamical structure, as shown in Fig.~\ref{fig:lyap_polar1}(c).
The boundary $\partial\mathcal{D}_{R,\Theta}$ separating terminating and bounded motion develops several disconnected-looking branches, interspersed with regions of bounded dynamics. The previously smooth, cup-like boundary becomes highly irregular.
At the same time, the region of chaotic dynamics expands significantly, and the intensity of chaos increases, with the largest Lyapunov exponent reaching values up to $\lambda \approx 0.85$, indicating a strongly chaotic regime.
A particularly notable feature is the coexistence of numerous chaotic islands embedded within large regions of terminating motion.

For $\alpha=2$, corresponding to a string twice as massive as the swinging mass, the system exhibits predominantly chaotic dynamics, as shown in Fig.~\ref{fig:lyap_polar1}(d).
The terminating region occupies a substantial portion of the admissible domain, while the remaining bounded region is largely dominated by strong chaos.
These observations provide further evidence that increasing $\alpha$ enhances both the system's nonlinear character and its tendency toward escape.
Moreover, the boundaries separating regular, chaotic, and terminating regions become smoother and more coherent compared to the intermediate case $\alpha=1$.


\subsubsection{Statistical analysis of the $\alpha$-dependent global   dynamics}
With the help of Lyapunov exponents, we can also  estimate the fractions of
chaotic, regular, and terminating motion in the
$(R_0,\Theta_0)$ initial-condition plane as functions of the parameter
$\alpha$. For each fixed value of $\alpha$, the largest Lyapunov
exponent is computed over a large ensemble of uniformly distributed
initial conditions. The chaotic fraction is then defined as the ratio
of trajectories with a positive largest Lyapunov exponent to the total
number of sampled initial conditions. At the same time, terminating
trajectories are identified and included in the statistics, allowing
us to determine the  fractions of chaotic, regular, and
terminating motion. In addition, for each fixed value of $\alpha$, we
compute the maximal values of the largest Lyapunov exponent over the dynamically admissible
region,
\begin{align*}
	\lambda_{\max}(\alpha)
	=
	\max_{(R_0,\Theta_0)\in\mathcal D_{R,\Theta}(\alpha)}
	\lambda(R_0,\Theta_0;\alpha),
\end{align*}
where $\lambda(R_0,\Theta_0;\alpha)$ denotes the largest Lyapunov
exponent associated with the trajectory starting from the initial
condition $(R_0,\Theta_0)$.

The results for $\alpha \in [0,10]$ are shown in Fig.~\ref{fig:percentage}. The black, blue, and dashed gray curves correspond to chaotic, regular, and terminating trajectories, respectively, while the shaded region under the chaotic curve represents the relative measure of chaotic initial conditions. The color of this region encodes the maximal Lyapunov exponent $\lambda_{\max}$ according to the logarithmic scale shown on the right. Consequently, the diagram simultaneously captures both the geometric extent of chaos and the strongest dynamical instability present in the system for a fixed value of $\alpha$.

For small values of $\alpha$, the dynamics is predominantly regular, with only a weak presence of chaos; in particular, regular motion occupies more than $80\%$ of the available $(R_0,\Theta_0)$ area. As $\alpha$ increases, a transition to a mixed regime occurs, in which the fraction of chaotic trajectories grows rapidly and reaches its maximum at intermediate values of $\alpha$. For larger $\alpha$, the dynamics becomes increasingly dominated by terminating motion, while the regular component gradually disappears and the chaotic fraction stabilizes at an approximately constant level. Simultaneously, the intensity of chaos increases significantly with $\alpha$. In particular, the maximal Lyapunov exponent reaches values of approximately $\lambda_{\max}\approx 6$, indicating the highly chaotic nature of the system dynamics. Moreover, for $\alpha \gtrsim 5$, the value of $\lambda_{\max}$ becomes nearly constant, suggesting the emergence of an asymptotic strongly chaotic regime in which the characteristic instability saturates despite further increase of the parameter $\alpha$.

\subsection{Energy-constrained dynamical maps}
The Lyapunov maps presented in the previous subsection provide a global
overview of the dynamics over the entire space of initial conditions.
However, since trajectories corresponding to different initial conditions
generally evolve on different energy levels, such maps combine motions with
different values of the Hamiltonian. From the viewpoint of Hamiltonian
dynamics, it is therefore natural to complement this analysis by restricting
the initial conditions to a fixed energy surface.

This energy-constrained approach not only enables a direct comparison of
trajectories with the same total energy, but also reveals the geometric structure of the energetically accessible region through the associated Hill regions and zero-velocity curves.
As we shall see, these geometric objects provide a natural framework for interpreting the Lyapunov exponent maps and for understanding the transition between regular, chaotic, and terminating motions.

\subsubsection{Hill regions and zero-velocity curves}

We now turn to a complementary description of the dynamics based on Lyapunov exponent maps constructed on a fixed energy level. In contrast to the previous analysis, where the initial conditions were sampled in the $(R_0,\Theta_0)$ plane with vanishing momenta, we now impose the energy constraint and consider initial conditions of the form $(R_0,\Theta_0,P_R,0)$, where the radial momentum $P_R$ is determined from the energy integral $H=E$. This restricts the motion to a given energy surface and provides a physically consistent framework for the analysis.

This approach is closely related to the classical study of zero-velocity curves (ZVC), extensively investigated by Tufillaro in~\cite{Tufillaro:88a} for the integrable case $\alpha=0$. In that setting, the Hill curves offer a clear geometric characterization of the accessible region in the configuration space. Here, we extend this perspective to the non-integrable case $\alpha=1/2$ and examine how the geometry of the accessible region and the associated Lyapunov maps evolve with energy.

The accessible configurations are determined by the effective potential
\[
	\mathcal{V}(R,\Theta)
	=R(\mu+\alpha\eta)-\frac{\alpha}{2}R^{2}-\left(R+\frac{\alpha}{2}R^{2}\right)\cos\Theta.
\]
In contrast to the case $\alpha=0$, the quadratic dependence on $R$ complicates the analytical determination of the zero-velocity curves in polar coordinates, as the equation $V(R,\Theta)=E$ leads to nonlinear relations.

In the studied model, for a fixed energy level $E$,  the Hill region   is defined as the set of configurations for which the kinetic energy remains nonnegative,
\begin{align*}
	\mathcal{H}_E
	:=
	\Big\{
	(R,\Theta)\in (0,\eta)\times \mathbb{S}^1
	\;:\;
	V(R,\Theta)\le E
	\Big\}.
\end{align*}
Its boundary, known as the zero-velocity curve (ZVC),
\begin{align*}
	\partial \mathcal{H}_E
	:=
	\Big\{
	(R,\Theta)\in (0,\eta)\times \mathbb{S}^1
	\;:\;
	V(R,\Theta)=E
	\Big\},
\end{align*}
separates admissible configurations from energetically forbidden ones.

Lyapunov exponent dynamical maps are then constructed by sampling initial conditions $(R_0,\Theta_0)\in \mathcal{H}_E$, with $P_\Theta=0$ and $P_R$ determined from the energy constraint, ensuring that all trajectories evolve on the same energy level.

However, the Hill region $\mathcal{H}_E$ provides only a local, energy-based constraint and does not guarantee that trajectories remain in the admissible domain for all times. In particular, even if $(R_0,\Theta_0)\in \mathcal{H}_E$, the corresponding trajectory may reach the boundary $R=0$ or $R=\eta$ in finite time, leading to terminating motion.
To distinguish such cases, we introduce the set of dynamically admissible initial conditions
\begin{align*}
	\mathcal{D}_{R,\Theta}(E)
	:=
	\Big\{
	(R_0,\Theta_0)\in \mathcal{H}_E
	\;:\;
	0<R(\tau;R_0,\Theta_0,E)<\eta\quad\text{for all }\tau\ge0
	\Bigr\},
\end{align*}
which consists of those initial configurations for which the motion remains bounded away from the geometric constraint.

Thus, while the Hill region $\mathcal{H}_E$ determines the set of
energetically accessible configurations, the subset
$\mathcal{D}_{R,\Theta}(E)$ consists of those initial conditions whose
trajectories remain dynamically admissible for the prescribed energy
level.
In particular, the complement
$\mathcal{H}_E\setminus\mathcal{D}_{R,\Theta}(E)$ corresponds to initial
conditions leading to terminating trajectories.
This distinction plays
a crucial role in the interpretation of the Lyapunov dynamical maps, as
it allows us to separate regular and chaotic dynamics from terminating
motion within the energetically accessible region.

The geometry of the Hill region depends crucially on the value of the
energy. In particular, the first contact of the energetically accessible
region with the upper geometric constraint $R=\eta$ occurs at the critical
energy level
\begin{align*}
	E_{\mathrm{crit}}
	=
	\min_{\Theta} V(\eta,\Theta)
	=
	\eta(\mu-1).
\end{align*}
This minimum is attained at $\Theta=0$. Hence, at $E=E_{\mathrm{crit}}$,
the zero-velocity curve reaches the boundary $R=\eta$ for the first time.

For $E<E_{\mathrm{crit}}$, the Hill region remains separated from the
upper boundary $R=\eta$, so that this boundary is energetically
inaccessible. At $E=E_{\mathrm{crit}}$, the zero-velocity curve touches
$R=\eta$, producing a transition in the accessible configuration region.
For $E>E_{\mathrm{crit}}$, the energetically accessible region intersects
the boundary $R=\eta$, and trajectories may therefore reach this geometric
constraint in finite time, leading to terminating motion.

It should be noted, however, that $R=0$ is not a zero-velocity boundary
generated by the energy constraint. Rather, it corresponds to a geometric
singularity of the system and may also lead to terminating motion. Thus,
the Hill region determines the energetically accessible part of the
configuration space, whereas the dynamically admissible set
$\mathcal{D}_{R,\Theta}(E)$ is obtained by further excluding initial conditions whose
trajectories reach either $R=0$ or $R=\eta$ in finite time.
Consequently, the zero-velocity curve provides the natural geometric framework for constructing and interpreting the Lyapunov exponent maps on
a fixed energy level, while the distinction between $\mathcal{H}_E$ and
$\mathcal{D}_{R,\Theta}(E)$ allows us to separate regular and chaotic dynamics from
terminating motion.

\subsubsection{Dynamical maps on energy levels for fixed $\alpha$}
Figure~\ref{fig:lyap_polar3} presents Lyapunov exponent maps computed
on several fixed energy levels for $\alpha=1/2$, with $\mu=3$ and
$\eta=0.5$. The maps reveal a strong interplay between the geometry of
the energetically accessible region and the dynamical organization of
the phase space.
Thus, the diagrams presented in Fig.~\ref{fig:lyap_polar3} can be interpreted as an initial-condition space map defined by
\begin{align}
	\label{eq:IC_energy_diag}
	(R_0,\Theta_0)
	\;\longmapsto\;
	\begin{cases}
		\lambda(R_0,\Theta_0)
		 & \text{if } (R_0,\Theta_0)\in \mathcal{D}_{R,\Theta}(E),                         \\[0.2cm]
		\text{T-motion}
		 & \text{if } (R_0,\Theta_0)\in \mathcal{H}_E \setminus \mathcal{D}_{R,\Theta}(E), \\[0.2cm]
		\text{Forbidden}
		 & \text{if } (R_0,\Theta_0)\notin \mathcal{H}_E,
	\end{cases}
\end{align}
where $0 < R_0 < \eta$ and  $\mathcal{D}_{R,\Theta}(E)$ denotes the set of dynamically admissible initial conditions on the fixed energy level $E$.

\begin{figure}[t]
	\centering
	\subfigure[ $E=0.5$]{
		\includegraphics[width=0.465\linewidth]{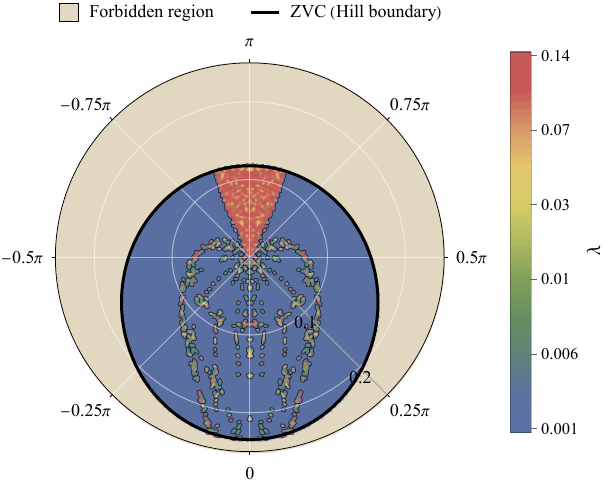}
	}
	\subfigure[ $E=1$]{
		\includegraphics[width=0.465\linewidth]{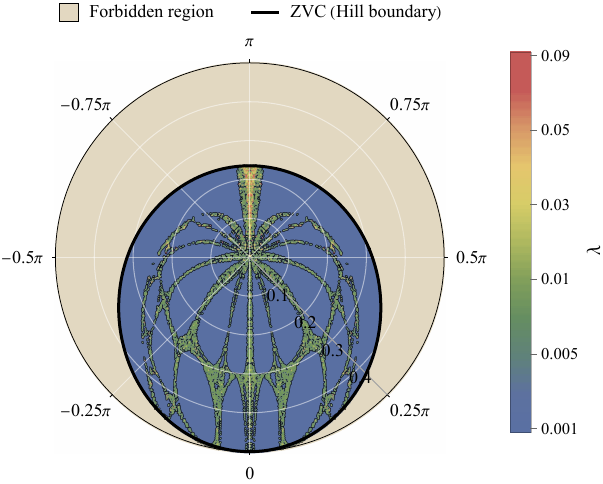}
	} \\
	\subfigure[ $E=1.025$]{
		\includegraphics[width=0.465\linewidth]{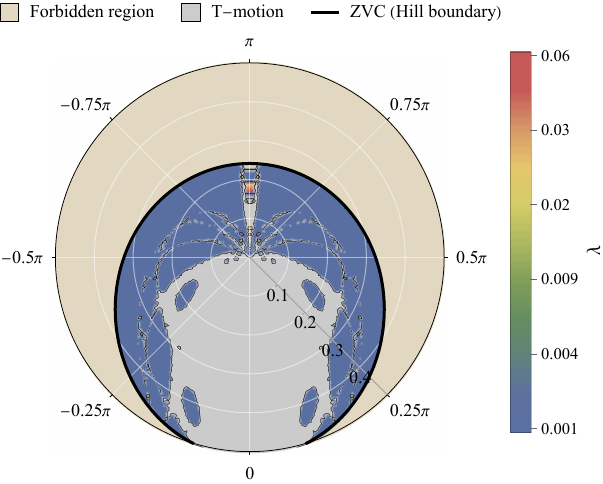}
	}
	\subfigure[ $E=1.1$]{
		\includegraphics[width=0.465\linewidth]{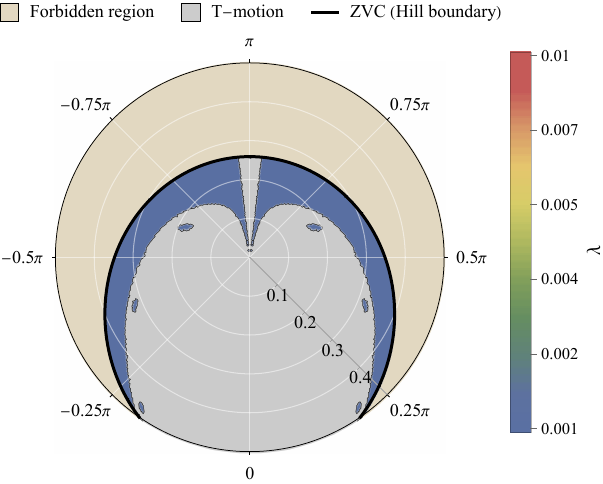}
	}
	\caption{(Color online) Lyapunov exponent maps in the polar plane $(R_0,\Theta_0)$ for $\alpha=0.5$, $\mu=3$, and $\eta=0.5$, computed on constant energy levels. The black curve denotes the zero-velocity curve (ZVC), separating the energetically accessible region from the forbidden region (light brown). The logarithmic color scale represents the largest Lyapunov exponent $\lambda$. Blue regions correspond to numerically regular motion, colored regions to chaotic dynamics, and gray regions to terminating trajectories (T-motion).
		\label{fig:lyap_polar3}}
\end{figure}

For low energy, $E=0.5$ (Fig.~\ref{fig:lyap_polar3}a), the accessible region is relatively small and bounded away from the geometric constraint $R=\eta$. The corresponding zero-velocity curve is smooth and only slightly deformed from a circular shape, reflecting the angular dependence of the potential. In particular, the radial extent of the accessible region varies with $\Theta$, attaining its maximum near $\Theta=0$ and its minimum near $\Theta=\pi$. The deformation remains moderate at this energy level.
The dynamics is predominantly regular, with only a small fraction of chaotic trajectories. The chaotic behavior is mainly localized in a narrow angular sector near $\Theta=\pi$, forming a characteristic wedge-shaped structure. Outside this region, only a few weakly chaotic trajectories are observed within the regular domain, indicating that regular dynamics still dominates the accessible phase space.

As the energy increases to $E=1$ (Fig.~\ref{fig:lyap_polar3}b), the Hill region expands and reaches the critical configuration. In fact, $E=1$ corresponds to the critical energy level $E_{\mathrm{crit}}$, at which the zero-velocity curve becomes tangent to the boundary at the point $(R,\Theta)=(\eta,0)$.
At this energy, the dynamical structure undergoes a qualitative change. The previously scattered regions of weakly chaotic behavior merge into coherent, extended structures forming distinct branches. These branches originate near the map's center and extend outward toward the zero-velocity curve, forming a characteristic web-like pattern spanning a large portion of the accessible region.
Although the dynamics is still predominantly regular, chaotic trajectories are now more widespread and organized along these elongated structures. The overall picture reflects a transition from a near-integrable regime to a more structured and globally distributed weak chaos, closely correlated with the geometry of the Hill boundary.

At energies slightly above the critical value, $E=1.025$ (Fig.~\ref{fig:lyap_polar3}c), the topology of the Hill region changes, and a large portion of initial conditions leads to terminating motion. The dynamically admissible region $\mathcal{D}_{R,\Theta}$ shrinks significantly, and the remaining non-terminating trajectories are confined to a thin region near the upper boundary of the Hill region. The chaotic structures are strongly suppressed, and most of the surviving motion appears to be regular.
Interestingly, in contrast to many Hamiltonian systems, where increasing the energy leads to a gradual growth of global chaos and a transition from regular to strongly chaotic dynamics, no such behavior is observed here. Instead, the increase in energy primarily results in the expansion of terminating trajectories, effectively removing a large portion of phase space from the set of dynamically admissible motions, thereby suppressing the development of widespread chaos.

For larger energy, $E=1.1$ (Fig.~\ref{fig:lyap_polar3}d), terminating motion dominates almost the entire accessible region. In particular, trajectories initiated near $\Theta=0$ terminate directly, forming a narrow wedge-shaped region extending from the center toward the zero-velocity curve.
The remaining dynamically admissible set is confined to a thin crescent-shaped region located away from $\Theta=0$ and bounded by the zero-velocity curve. Within this region, the motion is predominantly regular, as indicated by the uniformly small values of the Lyapunov exponent.
Interestingly, inside the large terminating domain, a few small isolated islands of regular motion are still visible. At this energy level, no significant chaotic behavior is observed, and the dynamics is effectively reduced to a coexistence of regular and terminating trajectories.

Overall, the sequence of plots demonstrates that the energy-driven deformation of the Hill region decisively controls the transition from predominantly regular dynamics, through a mixed regime with significant chaos, to a regime dominated by terminating trajectories.

\subsubsection{Statistical analysis of the energy-dependent global dynamics}
\begin{figure*}[t]
	\centering 			\includegraphics[width=0.6	\linewidth]{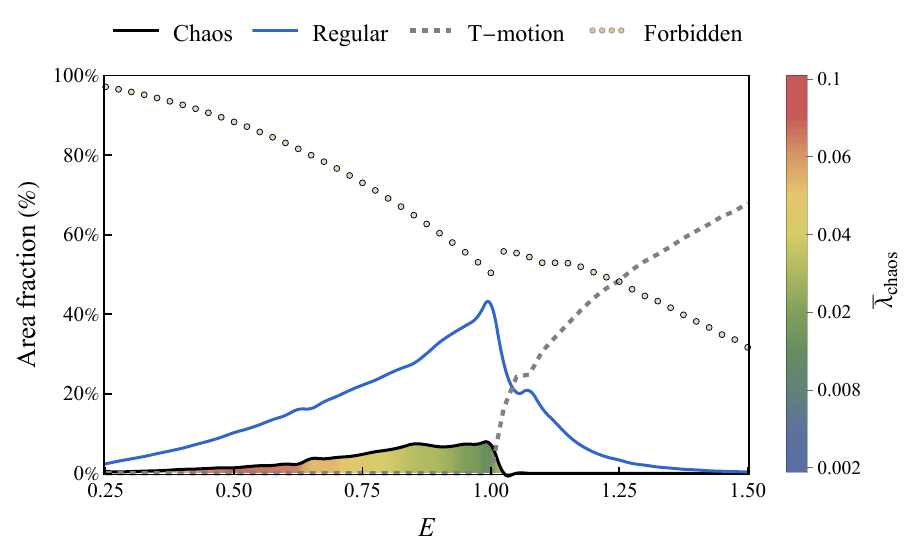} \caption{
(Color online) Fraction of the occupied area in the $(R_0,\Theta_0)$ initial-condition space as a function of the energy $E$, for fixed parameters $\mu=3$, $\eta=0.5$, and $\alpha=0.5$. The black, blue, dashed gray, and dotted light-brown curves correspond to the fractions of chaotic, regular, terminating, and forbidden trajectories, respectively. The shaded region under the chaotic curve highlights the fraction of chaotic trajectories. The color of the chaotic curve encodes the mean Lyapunov exponent $\overline{\lambda}_{\mathrm{chaos}}(E)$, computed over all chaotic trajectories at the corresponding energy level according to the logarithmic scale shown on the right, thereby combining information about both the extent and the characteristic intensity of chaos.
\label{fig:percentage2}}
\end{figure*}
In contrast to the previous analysis performed for varying $\alpha$ in the $(R_0,\Theta_0)$ initial-condition plane with vanishing initial velocities, we now investigate the dynamics as a function of the energy $E$, while keeping the parameters $\mu=3$, $\eta=0.5$, and $\alpha=0.5$ fixed. In this case, the initial radial momentum $P_R$ is no longer prescribed, but is determined from the energy constraint $H=E$. Consequently, for each energy level, the admissible initial conditions are restricted by the geometry of the corresponding Hill region. Since the main objective is now to characterize the global dynamical behavior associated with a given energy level, rather than the strongest local instability, it is more natural to employ the mean Lyapunov exponent instead of the maximal one used in the previous parameter-dependent analysis.

Fig.~\ref{fig:percentage2} presents the fraction of the occupied area
in the $(R_0,\Theta_0)$ initial-condition space as a function of the
energy $E$. The black, blue, dashed grey, and dotted light-brown
curves correspond to chaotic, regular, terminating, and forbidden
regions, respectively, while the shaded region under the chaotic curve
represents the   fraction of chaotic trajectories. The color encodes the mean Lyapunov exponent
\begin{align*}
	\overline{\lambda}_{\mathrm{chaos}}(E)
	=
	\frac{1}{N_{\mathrm{chaos}}}
	\sum_{\substack{(R_0,\Theta_0)\in\mathcal{D}_{R,\Theta}(E) \\
			\lambda(R_0,\Theta_0;E)>0}}
	\lambda(R_0,\Theta_0;E),
\end{align*}
where $\lambda(R_0,\Theta_0;E)$ denotes the largest Lyapunov exponent of the trajectory starting from $(R_0,\Theta_0)$, $\mathcal{D}_{R,\Theta}(E)$ is the set of dynamically admissible initial conditions at the energy level $E$, and $N_{\mathrm{chaos}}$ is the number of chaotic trajectories satisfying $\lambda>0$. Thus, the average is computed exclusively over chaotic trajectories, avoiding the artificial reduction that would result from including regular trajectories with $\lambda\approx0$. Consequently, the diagram provides two complementary statistical measures: the shaded area represents the fraction of chaotic trajectories, whereas the color indicates their mean Lyapunov exponent.

For low energies, only a relatively small portion of the $(R_0,\Theta_0)$ plane belongs to the dynamically accessible region. Most of the configuration space remains forbidden due to the energy constraint, and therefore the dynamics is dominated by the forbidden area. In this regime, the surviving admissible trajectories are predominantly regular, while the fraction of chaotic motion is negligible. As the energy increases, the Hill region gradually expands, allowing a larger portion of the $(R_0,\Theta_0)$ space to become dynamically accessible. Consequently, a coexistence of regular and chaotic dynamics emerges. The fraction of regular trajectories grows steadily, whereas the chaotic component develops mainly near the boundaries separating qualitatively different types of motion. Simultaneously, the mean Lyapunov exponent also increases, indicating not only a larger chaotic area but also a stronger degree of instability.

A particularly important transition occurs near the critical energy level $E_{\mathrm{crit}}=1$. At this value, the pattern of the Hill region changes, corresponding to the first contact of the zero-velocity curve with the geometric boundary $R=\eta$. Around this threshold, the chaotic component attains its maximal relative extent, while the regular bounded dynamics still occupies a significant portion of the admissible phase-space region. However, once the energy exceeds the critical value, the structure of the dynamics changes drastically. Even for slightly supercritical energies, the chaotic component rapidly disappears and the fraction of regular bounded trajectories decreases sharply. At the same time, terminating motion grows very rapidly, progressively replacing the forbidden, regular, and chaotic regions. This behavior is physically natural, since sufficiently large energies allow trajectories to reach the geometric boundary of the system, leading to termination of the motion. In fact, for $E=2$ (not shown in the figure), nearly $90\%$ of the accessible initial-condition space corresponds to terminating trajectories, indicating that bounded long-term dynamics survives only on a very small subset of the configuration space.

\subsubsection{Lyapunov maps at the critical energy $E=E_{\mathrm{crit}}$}

We conclude the analysis by considering Lyapunov maps at the critical energy $E=E_{\mathrm{crit}}=1$, fixing $\mu=3$ and $\eta=1/2$, while treating $\alpha$ as the control parameter. This allows us to isolate the influence of the string mass on both the geometry of the Hill region and the resulting dynamics.

We consider initial conditions $(R_0,\Theta_0,P_R,0)$ constrained by the energy relation $H=E$. The corresponding Lyapunov exponent maps are shown in Fig.~\ref{fig:lyap_alpha}. For $\alpha=0$, the system is integrable: the Hill region is nearly circular, and the dynamics is entirely regular, with $\lambda\approx 0$. For $\alpha=1$, integrability is lost and a coexistence of regular and chaotic dynamics emerges as shown in Fig.~\ref{fig:lyap_alpha}(b). Simultaneously, the Hill boundary becomes slightly deformed. Regular and chaotic dynamics coexist, with the latter forming extended filamentary structures strongly correlated with the geometry of the zero-velocity curve.

A qualitative transition occurs near $\alpha=4$. Fig.~\ref{fig:lyap_alpha}(c) shows that the dynamics becomes predominantly chaotic, while the zero-velocity curve undergoes a critical geometric transformation: it develops a cone-like singularity at $\Theta=0$, where it touches the boundary $R=\eta$ in a non-transversal manner. The transition near $\alpha=4$ is therefore associated not only with a rapid increase of the chaotic component, but also with a substantial reorganization of the energetically accessible region.

This behavior can be understood by restricting the effective potential
to the invariant line $\Theta=0$, for which
\begin{align*}
	\mathcal{V}(R,0)
	=
	R\left(2+\frac{\alpha}{2}\right)-\alpha R^2.
\end{align*}
For the critical energy $E=1$, the intersections of the corresponding
zero-velocity curve satisfy
\begin{align*}
	\alpha R^2
	-
	\left(2+\frac{\alpha}{2}\right)R
	+1=0,
	\qquad
	(R_1=\frac12=\eta),
	\qquad
	R_2=\frac{2}{\alpha}.
\end{align*}
For $\alpha<4$, the second root lies outside the physical interval
$0<R<\eta$, so the zero-velocity curve intersects the boundary only at
$R_1=\eta$ . At the critical value $\alpha=4$, the two roots merge at the equlibrium~\eqref{eq:delta},  $R_*=\eta$. Since
\begin{align*}
	\mathcal{V}(R_*,0)=E,
	\qquad
	\mathcal{V}_R(R_*,0)=0,
\end{align*}
the point $(R,\Theta)=(\eta,0)$ is the degenerate point of the
energy level. Consequently, the zero-velocity curve develops a
singularity at this point, giving rise to the cone-like structure
visible in Fig.~\ref{fig:lyap_alpha}.

For $\alpha>4$, the second root enters the interval $0<R<\eta$, so the zero-velocity curve closes strictly inside the domain and no longer reaches the boundary $R=\eta$. At the same time, the  equilibrium radius $R_*$ lies between $R_1$ and $R_2$ and for this case
becomes energetically inaccessible, since
$
	\mathcal{V}(R_*,0)>E.
$
Consequently, the Hill region shrinks, as visible in Fig.~\ref{fig:lyap_alpha}(d), and becomes progressively more rounded as $\alpha$ increases. This reflects the increasing dominance of the quadratic terms in $R$, which reduce the angular anisotropy of the potential and shift the energetically accessible region toward smaller radii.

Therefore, the qualitative changes observed in the Lyapunov maps are
a direct consequence of the $\alpha$-dependent deformation of the
effective potential, which modifies the geometry and topology of the
corresponding zero-velocity curves.

\begin{figure}[t]
	\centering
	\subfigure[ $\alpha=0$]{
		\includegraphics[width=0.465\linewidth]{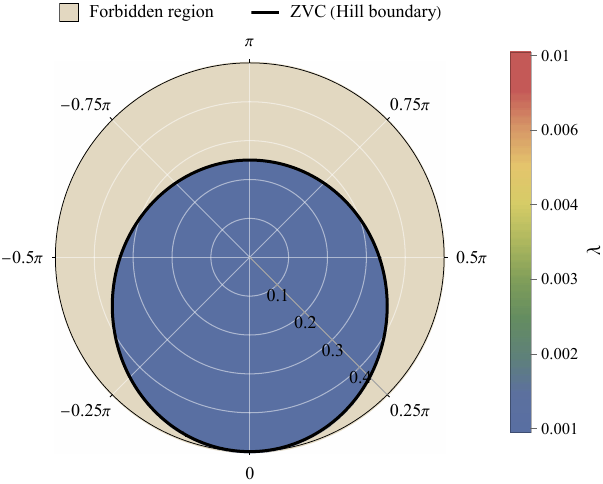}
	}
	\subfigure[ $\alpha=1$]{
		\includegraphics[width=0.465\linewidth]{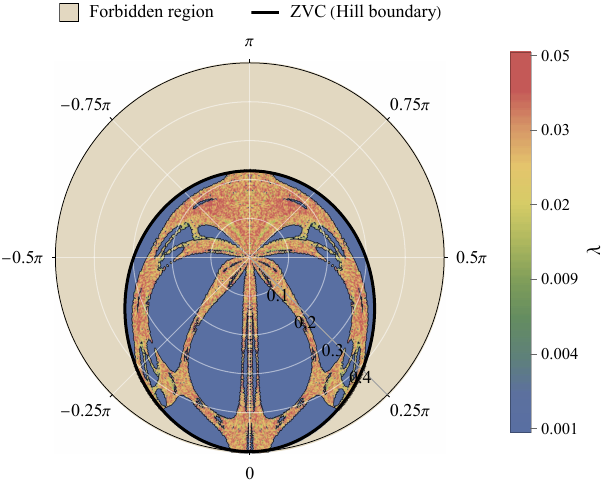}
	} \\
	\subfigure[ $\alpha=4$]{
		\includegraphics[width=0.465\linewidth]{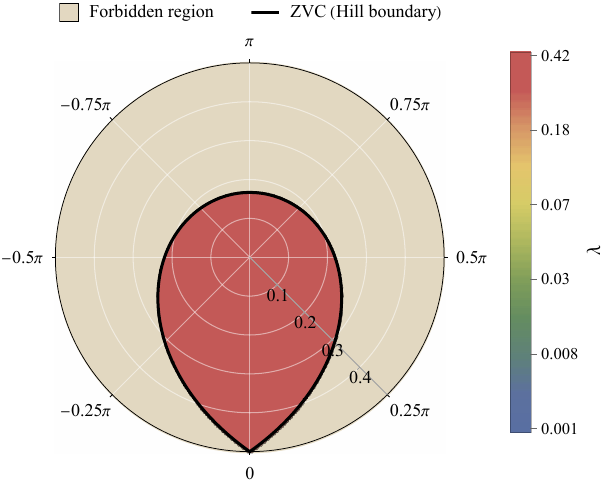}
	}
	\subfigure[ $\alpha=20$]{
		\includegraphics[width=0.465\linewidth]{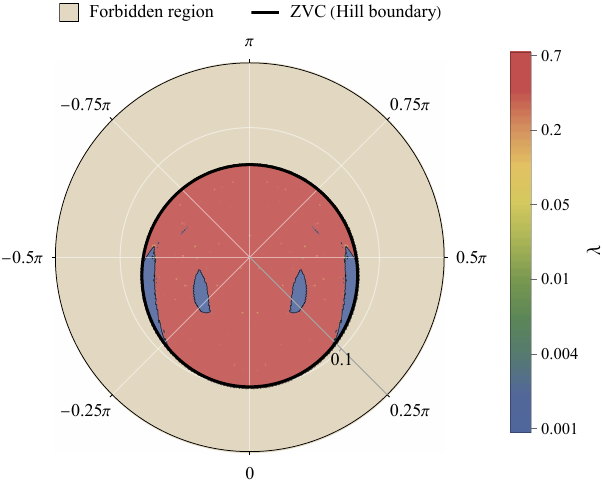}
	}
	\caption{(Color online) Lyapunov exponent maps in the polar plane $(R_0,\Theta_0)$ at the fixed critical energy level $E=E_{\mathrm{crit}}=1$, for $\mu=3$ and $\eta=0.5$, while varying the parameter $\alpha$.   The black curve denotes the zero-velocity curve (ZVC), separating the energetically accessible region from the forbidden region (light brown). The color scale (logarithmic) corresponds to the largest Lyapunov exponent~$\lambda$. Blue regions indicate regular motion, whereas colored regions correspond to chaotic dynamics.
	\label{fig:lyap_alpha}}
\end{figure}

\section{Lyapunov Refined Maps (LRM)}
The method of Lyapunov exponents is one of the fundamental tools in the
study of nonlinear Hamiltonian systems, providing a quantitative measure
of dynamical instability and the intensity of chaos. In the previous
sections, Lyapunov exponent maps enabled us to identify regions of
regular, chaotic, and terminating motion, as well as to investigate the
global organization of the phase space and its dependence on the system
parameters. In particular, they revealed the influence of the Hill
regions and zero-velocity curves on the overall dynamical structure.

Despite these advantages, classical Lyapunov maps possess an important
intrinsic limitation. Although they reliably distinguish regular motion
from chaos, they provide no information about the internal organization
of the regular regions. In Hamiltonian systems, both periodic and
quasi-periodic trajectories are characterized by vanishing largest
Lyapunov exponents and therefore appear indistinguishable in standard
Lyapunov maps. Consequently, resonance islands, periodic windows,
bifurcation structures, and other fine topological features remain
hidden inside the uniformly blue regions corresponding to
$\lambda\approx0$.

To overcome this limitation, we introduce the
\emph{Lyapunov Refined Map} (LRM), which augments the classical
Lyapunov map with additional topological information obtained from the
associated Poincar\'e map. Instead of assigning only the largest
Lyapunov exponent to each sampled point, the LRM additionally records
the geometrical period of the corresponding periodic orbit whenever
such an orbit is detected. Consequently, periodic and quasi-periodic
motions, which are indistinguishable in classical Lyapunov maps,
become naturally separable within the regular regions.

Unlike classical bifurcation diagrams, which represent one-dimensional
slices obtained by varying a single control parameter, the LRM is a
genuinely two-dimensional representation. Each sampled point is
simultaneously assigned the largest Lyapunov exponent together with the
detected geometrical period of the corresponding Poincar\'e orbit.
Thus, the LRM preserves the global information on dynamical stability
provided by Lyapunov exponents while revealing the resonance structure
hidden inside the regular regions.

The construction of the LRM consists of two ingredients. First, we
recall the phase-parametric (bifurcation) diagram technique developed
in our previous works~\cite{Szuminski:23::,Szuminski:24::}, which
provides the geometric interpretation of periodic, quasi-periodic, and
chaotic trajectories. Next, we introduce the recurrence-based
algorithm used to detect the geometrical period of periodic orbits and
show how this information is incorporated into the Lyapunov Refined
Map.

\subsection{Bifurcation diagrams}

\begin{figure*}[t]
	\centering 			\subfigure[    Global bifurcation diagrams]{
		\includegraphics[width=0.85	\linewidth]{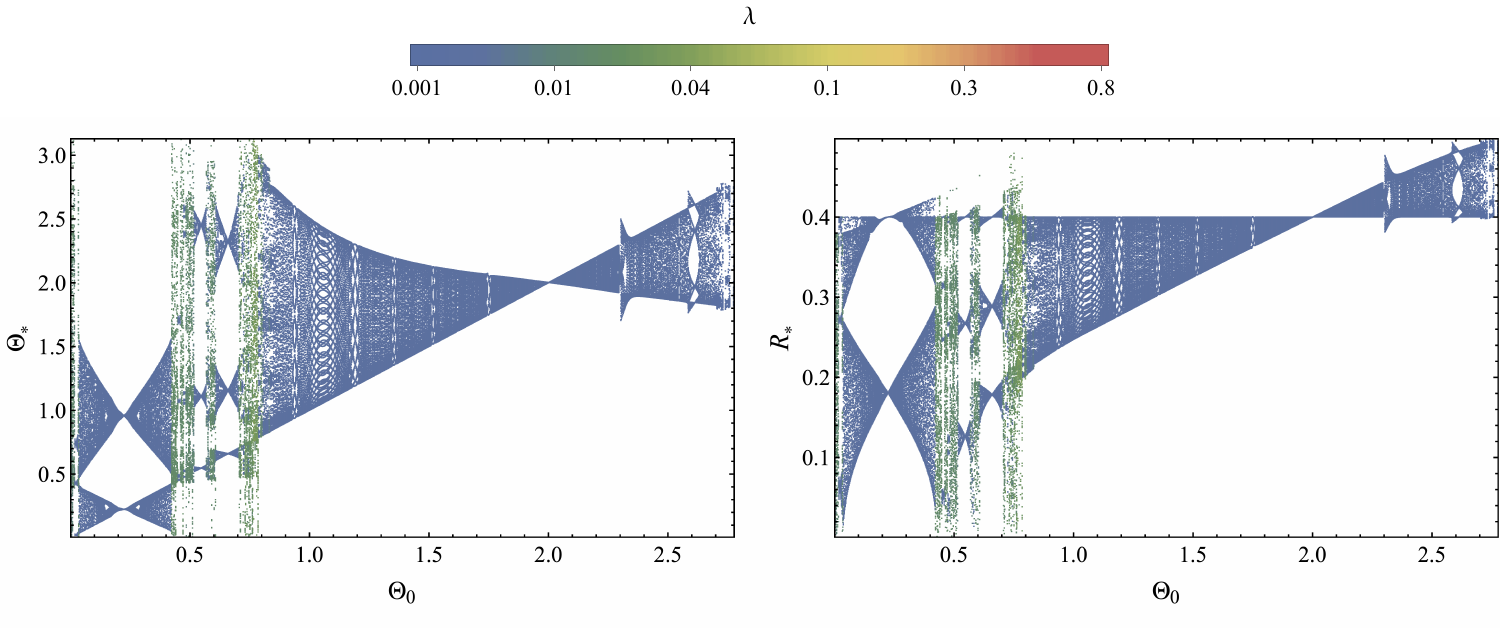} }
	\subfigure[    Enlargements of the initial interval ]{
		\includegraphics[width=0.86	\linewidth]{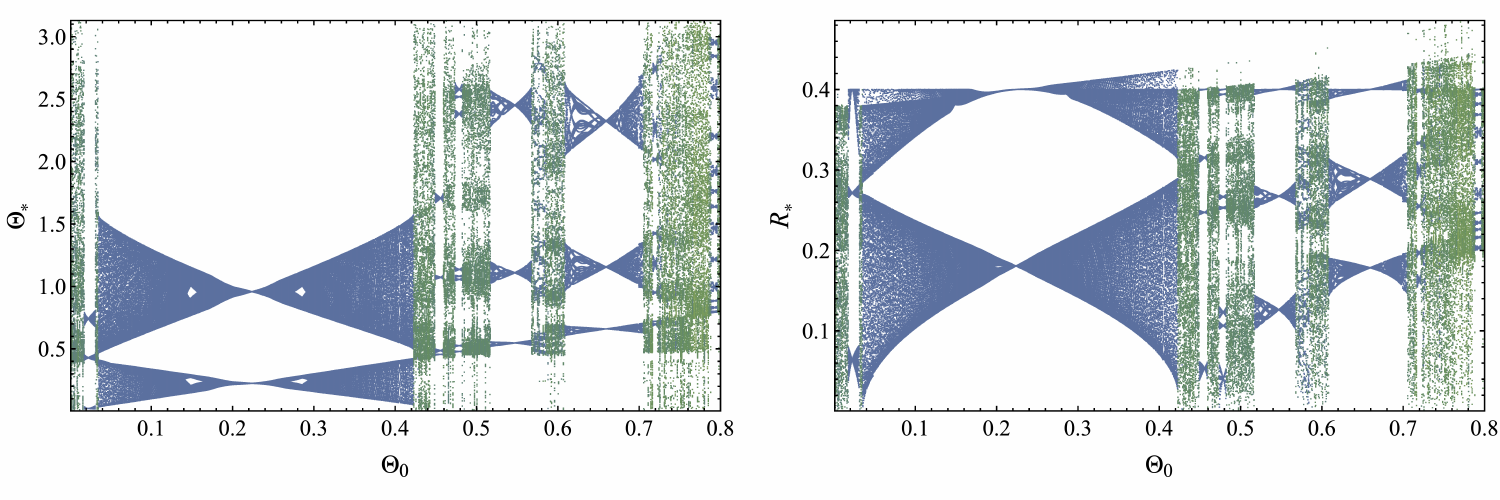} }
	\caption{Bifurcation diagram of system~\eqref{eq:Ham_sys} as a function of the initial swing angle $\Theta_{0}\in(0,\pi)$ together with its magnification over the interval $\Theta_{0}\in(0,\pi/4)$. The initial conditions and parameter values are the same as those used in Fig.~\ref{fig:lyap_polar1}(c), i.e., the diagram corresponds to the angular direction of the Lyapunov map for the fixed initial radial coordinate $R_0=0.4$. Here, $\Theta(t_*)$ and $R(t_*)$ denote the values of the state variables at the instants $t_*$ satisfying $P_{\Theta}(t_*)=0$ and $\dot P_{\Theta}(t_*)<0$. The bifurcation diagram is combined with the largest Lyapunov exponent $\lambda$, whose logarithmic color scale corresponds to the values shown in the Lyapunov map of Fig.~\ref{fig:lyap_polar1}(c).  The figure reveals the coexistence of periodic, quasi-periodic, and chaotic motions together with periodic windows embedded within chaotic layers. Representative examples of periodic, quasi-periodic, and chaotic trajectories are shown in Fig.~\ref{fig:polar}.}
	\label{fig:Biff}
\end{figure*}
\begin{figure*}[t]	\centering
	\subfigure[$ R_0=0.4,\, \Theta_0=0.66$]{
		\includegraphics[width=0.32\linewidth]{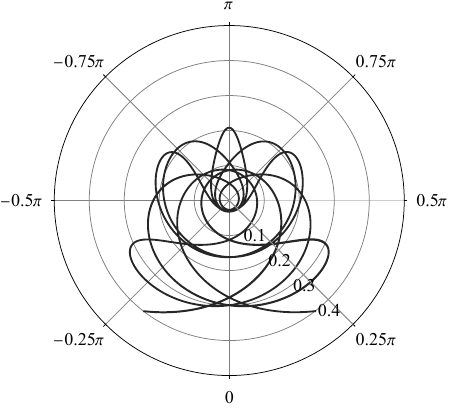}}
	\subfigure[$ R_0=0.4,\,  \Theta_0=0.4$]{
		\includegraphics[width=0.32\linewidth]{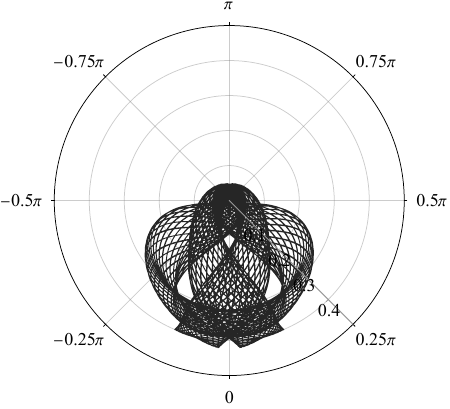}}
	\subfigure[$  R_0=0.4,\, \Theta_0=0.6$]{
		\includegraphics[width=0.32\linewidth]{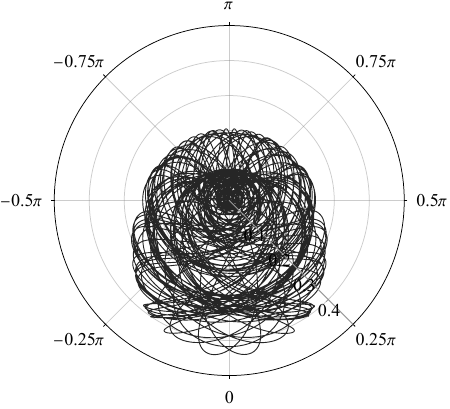}}
	\caption{Representative trajectories corresponding to selected initial conditions from the bifurcation diagrams shown in Fig.~\ref{fig:Biff}. Panels (a)–(c) illustrate characteristic examples of periodic, quasiperiodic, and chaotic motion, respectively. The corresponding initial conditions are indicated below each panel.
		\label{fig:polar}}
\end{figure*}
To complement the information provided by Lyapunov exponent maps, we employ phase-parametric (bifurcation) diagrams constructed from a chosen Poincar\'e section plane. For each value of the control parameter, the trajectory is integrated over a sufficiently long time interval, and all intersections with the section are recorded after discarding an initial transient. The corresponding values of the selected section variable are then plotted as a function of the parameter. Repeating this procedure over a dense parameter grid yields a bifurcation diagram that visualizes the evolution of the underlying invariant structures.

Periodic trajectories appear as a finite number of isolated branches, whereas quasi-periodic motions generate continuous curves associated with invariant tori. Chaotic trajectories, in turn, fill bounded regions of the diagram with dense point clouds reflecting the irregular exploration of the accessible phase space. Thus, bifurcation diagrams provide a direct geometric representation of resonances, periodic windows, bifurcation cascades, and transitions between different dynamical regimes.

The bifurcation diagrams obtained for the heavy Swinging Atwood Machine governed by~\eqref{eq:Ham_sys} are presented in Fig.~\ref{fig:Biff}.
They were constructed for the fixed parameter values $\mu=3$, $\eta=1/2$ and $\alpha=1$ by choosing the initial conditions as $(0.4, \Theta_{0}, 0, 0)$ where the initial swing angle $\Theta_{0}$ is varied over the interval
$
	\Theta_{0}\in(0,\pi),
$
which corresponds to a one-dimensional slice of the Lyapunov diagram shown in Fig.~\ref{fig:lyap_polar1}(c) along the angular direction.
For each value of \(\Theta_{0}\), the equations of motion were numerically  integrated over a sufficiently long time interval, and the values of \(\Theta(t_{*})\) and \(R(t_{*})\) were recorded whenever the trajectory crossed $
	P_{\Theta}=0,
	\,
	\dot P_{\Theta}<0.
$
The resulting bifurcation diagrams are additionally coloured according to the largest Lyapunov exponent, enabling a direct comparison between the geometric structures revealed by the bifurcation diagrams and the dynamical information contained in the corresponding Lyapunov diagram presented in Fig.~\ref{fig:lyap_polar1}(c).

An excellent agreement between both approaches is clearly visible. The blue regions, corresponding to \(\lambda\approx0\), are associated with regular dynamics, whereas the green bands coincide with parameter intervals for which the largest Lyapunov exponent becomes positive, indicating chaotic motion. Unlike the Lyapunov exponent map, however, the bifurcation diagrams reveal the internal structure of the regular regions. In particular, they clearly distinguish isolated periodic solutions, represented by a finite number of discrete branches, from quasiperiodic trajectories, which generate continuous invariant curves corresponding to invariant tori in the phase space. Therefore, the combined use of Lyapunov maps and bifurcation diagrams provides a significantly more complete description of the phase-space organization than either method alone.

The global bifurcation diagrams shown in Fig.~\ref{fig:Biff}(a) demonstrate that the most intricate dynamics is concentrated within the relatively narrow interval
$
	\Theta_{0}\lesssim\frac{\pi}{4},
$
while for larger initial swing angles the dynamics is dominated by broad regular regions interrupted only by isolated resonance structures. For this reason, Fig.~\ref{fig:Biff}(b) presents a magnification of the initial interval, revealing the fine organization of the phase space. The enlarged view reveals periodic windows embedded within chaotic layers, quasiperiodic regions bounded by resonant periodic orbits, and transitions between regular and chaotic dynamics. The excellent correspondence with the Lyapunov diagram confirms the reliability of the proposed visualization and demonstrates that the chaotic layers detected by positive Lyapunov exponents are organized around intricate networks of periodic and quasiperiodic solutions.

Representative trajectories corresponding to selected points of the bifurcation diagrams are presented separately in Fig.~\ref{fig:polar}. They illustrate the three characteristic types of motion identified by the combined analysis, namely periodic, quasiperiodic, and chaotic dynamics, and provide a direct geometric interpretation of the structures observed in Fig.~\ref{fig:Biff}. It should be emphasized that the number of branches visible in a bifurcation diagram should not, in general, be identified with the period of the corresponding orbit of the Poincar\'e map. Since the bifurcation diagrams represent projections of the dynamics onto selected observables, several distinct points of a periodic orbit may project onto the same branch. Consequently, the apparent number of branches may differ from the actual period detected by the Lyapunov Refined Map algorithm developed in the next subsection.

\subsection{Numerical algorithm for detecting periodic trajectories}

To enrich the dynamical information contained in the Lyapunov maps, the
computation of the largest Lyapunov exponent is supplemented by a
numerical search for trajectories exhibiting geometrical periodicity. The algorithm
is formulated below for the initial-condition space
$\mathcal D_{R,\Theta}(\alpha)$ introduced in the previous section,
although the same procedure applies to any other admissible
two-dimensional sampling domain, including the parameter plane
$\mathcal D_{\mu,\alpha}(\Theta_0)$.

Starting from the previously computed Lyapunov map, the analysis is first
restricted to the numerically regular part of the admissible domain,
\[
	\mathcal D_{R,\Theta}^{\,\mathrm{reg}}(\alpha)
	=
	\left\{
	(R_0,\Theta_0)\in
	\mathcal D_{R,\Theta}(\alpha)
	\;:\;
	\lambda(R_0,\Theta_0)
	\leq
	\lambda_{\mathrm{thr}}
	\right\},
\]
where $\lambda_{\mathrm{thr}}$ is a prescribed threshold for the largest
finite-time Lyapunov exponent. In the computations presented below, this
threshold is fixed at
$
	\lambda_{\mathrm{thr}}=0.005.
$
Thus, only trajectories classified as numerically non-chaotic are
subjected to the periodicity-detection procedure.

For each initial condition
\[
	(R_0,\Theta_0)
	\in
	\mathcal D_{R,\Theta}^{\,\mathrm{reg}}(\alpha),
\]
the equations of motion are integrated and intersections with the
oriented Poincar\'e section
\[
	\Sigma_{+}
	=
	\left\{
	(R,\Theta,P_R,P_\Theta):\,
	P_\Theta=0,\, \,
	\dot P_\Theta<0
	\right\}
\]
are recorded. Since
$
	\dot\Theta=P_\Theta/D(R),
$
the first two conditions select angular turning points at which
$P_\Theta$ crosses zero in the prescribed direction.

Although the Hamiltonian dynamics evolves in the full phase space, the
present algorithm is designed to detect the geometrical periodicity of
the swinging mass in the configuration plane. Accordingly, each
intersection with $\Sigma_{+}$ is represented by the Cartesian position
\[
	X_n=(x_n,y_n)
	=
	\left(
	R_n\cos\Theta_n,
	R_n\sin\Theta_n
	\right).
\]
The use of Cartesian coordinates avoids angular wrapping and the
coordinate singularity inherent in polar coordinates, while allowing
the recurrence test to be formulated in terms of the standard Euclidean
distance.

An initial transient of duration
$
	T_{\mathrm{tr}}=300
$
is discarded before the section points are collected. For each
non-terminating trajectory, the integration is continued until either
$
	N_{\max}=1000
$
intersections with $\Sigma_{+}$ have been obtained or the maximal
integration time
$
	T_{\max}=5000
$
has been reached. Trajectories producing fewer than
$
	N_{\min}=300
$
section points are excluded from the periodicity analysis. The use of
long intersection sequences considerably reduces the probability of
false periodicity detection caused by transient near-recurrences or
slowly drifting quasi-periodic trajectories.

The numerical parameters used in the periodicity-detection procedure
were selected after an extensive preliminary analysis involving
convergence tests, sensitivity studies, and comparisons of the detected
periodic structures for different integration lengths and tolerance
values. The final parameter set was chosen to provide a robust
classification of periodic candidates while keeping the computational
cost manageable.

Let
\[
	X_1,X_2,\ldots,X_N\in\mathbb R^2
\]
denote the successive Cartesian positions recorded on the oriented
section. For a fixed integer $k\geq1$, the recurrence error is defined
as
\[
	\varepsilon_k
	=
	\max_{1\leq n\leq N-k}
	\left\|
	X_{n+k}-X_n
	\right\|,
\]
where $\|\cdot\|$ denotes the Euclidean norm in the Cartesian
configuration plane. The integer $k$ therefore measures the geometrical
period of the projected orbit on the oriented section $\Sigma_{+}$.

The maximal tested period is restricted to
\[
	k_{\max}=30.
\]
The restriction \(k_{\max}=30\) was adopted as a compromise between
resolving high-order resonance structures and maintaining a reliable
number of recurrence samples per residue class.
The integers $k=1,\ldots,k_{\max}$ are examined in increasing order, so
that the first accepted value is the smallest detected geometrical
period.

All geometric tolerances are expressed relative to the characteristic
diameter \(2\eta\) of the configuration domain. For the computations
reported here, \(\eta=1/2\), so that \(2\eta=1\).
A preliminary recurrence candidate is identified whenever
\[
	\varepsilon_k<\varepsilon_{\mathrm{rep}},\qquad \varepsilon_{\mathrm{rep}}=0.02.\]
The candidate is subsequently subjected to more restrictive geometric
tests.
For a candidate period \(k\), the sequence of section points is partitioned
into \(k\) subsequences,
\[
\mathcal C_j^{(k)}
=
\left\{
X_j,
X_{j+k},
X_{j+2k},
\ldots
\right\},
\qquad
j=1,\ldots,k.
\]
Each subsequence collects points that are expected to converge to the same
periodic point if the trajectory is periodic with period \(k\).

The spread of these recurrence classes is quantified by the effective
diameter
\[
d_k
=
\max_{1\leq j\leq k}
\left[
2\max_{X\in\mathcal C_j^{(k)}}
\left\|
X-\overline X_j
\right\|
\right],
\]
where \(\overline X_j\) denotes the arithmetic mean of the points in
\(\mathcal C_j^{(k)}\).
A candidate is retained only if
\[
	d_k<d_{\mathrm{tol}},
	\qquad
	d_{\mathrm{tol}}=0.005.
\]

To distinguish true periodic recurrence from a slowly drifting
quasi-periodic trajectory, the growth of the recurrence-class diameter
is also monitored. Let $d_k^{(1/2)}$ denote the effective diameter
computed from the first half of the intersection sequence and
$d_k^{(1)}$ the corresponding value computed from the entire sequence.
The candidate is accepted only if
\[
	\frac{d_k^{(1)}}{d_k^{(1/2)}}<g_{\mathrm{tol}},
	\qquad
	g_{\mathrm{tol}}=1.1,
\]
where \(d_k^{(1/2)}\) is computed from the first half of the recorded
section points and \(d_k^{(1)}\) from the complete sequence.
Thus, the effective diameter of every recurrence class is required to
remain nearly unchanged as additional intersections are included.
The tolerance \(g_{\mathrm{tol}}=1.1\) permits at most a \(10\%\)
increase in the estimated diameter, thereby rejecting slowly drifting
quasi-periodic trajectories while retaining numerically stable periodic
recurrences.

Finally, the accepted candidate must satisfy the stricter recurrence
condition
\[
	\varepsilon_k<\varepsilon_{\mathrm{tol}},
	\qquad
	\varepsilon_{\mathrm{tol}}=0.005.
\]
Consequently, an orbit is classified as geometrically periodic of
period $k$ only if $k$ is the smallest positive integer satisfying the
recurrence, diameter, and diameter-growth criteria.

For compactness, let
\(\mathscr A_k(R_0,\Theta_0)\) denote the conjunction
\[
	\varepsilon_k<\varepsilon_{\mathrm{tol}},
	\qquad
	d_k<d_{\mathrm{tol}},
	\qquad
	\frac{d_k^{(1)}}{d_k^{(1/2)}}<g_{\mathrm{tol}}.
\]
The set of
detected periodic candidates may then be written as
\begin{align}
	\label{eq:LRM_set}
	\mathcal P_{\mathrm{raw}}
	=
	\Bigl\{
	(R_0,\Theta_0,k)
	:\; &
	(R_0,\Theta_0)
	\in
	\mathcal D_{R,\Theta}^{\,\mathrm{reg}}(\alpha),\qquad
	k
	=
	\min
	\left\{
	j\in\{1,\ldots,k_{\max}\}:
	\mathscr A_j(R_0,\Theta_0)
	\right\}
	\Bigr\}.
\end{align}
Here, $k$ is the detected geometrical period of the configuration-space
projection sampled on the oriented Poincar\'e section. It should
therefore not be interpreted as the period of an unoriented projection
that identifies section points with opposite radial momenta.

Periodic candidates with the same detected period are subsequently
grouped into families according to their geometrical similarity. For
each candidate, the center and radius of the complete Cartesian
intersection set are defined by
\[
	C
	=
	\frac{1}{N}
	\sum_{n=1}^{N}X_n,
	\qquad
	\rho
	=
	\max_{1\leq n\leq N}
	\|X_n-C\|.
\]
Two candidates with the same period are assigned to the same family
whenever
\[
	\|C_1-C_2\|<\varepsilon_{\mathrm{fam}},
	\qquad
	|\rho_1-\rho_2|<\rho_{\mathrm{tol}},
\]
where
\[
	\varepsilon_{\mathrm{fam}}
	=
	\rho_{\mathrm{tol}}
	=
	0.005.
\]
From each family, only the representative with the smallest recurrence
error $\varepsilon_k$ is retained. The final output of the algorithm is
the set
\[
	\mathcal P
	=
	\left\{
	(R_0,\Theta_0,k)\in\mathcal P_{\mathrm{raw}}
	:
	\varepsilon_k
	=
	\min_{(R,\Theta,k)\in\mathcal F}
	\varepsilon_k
	\right\},
\]
where $\mathcal F$ denotes the numerical family containing
$(R_0,\Theta_0,k)$. Thus, each numerical family contributes exactly one
representative to the set $\mathcal P$, namely the candidate with the
smallest recurrence error among all candidates belonging to the same numerical family. Each element of $\mathcal P$ consists of an
initial condition together with the detected geometrical period of the
corresponding periodic trajectory.

Although the construction has been formulated for the
initial-condition domain
$\mathcal D_{R,\Theta}^{\,\mathrm{reg}}(\alpha)$, the same procedure can
be applied to any admissible two-dimensional sampling space. In
particular, in the next subsection it is used in the parameter plane
$\mathcal D_{\mu,\alpha}^{\,\mathrm{reg}}(\Theta_0)$ to construct the
corresponding Lyapunov Refined Map. A detailed mathematical and computational description of the general
LRM framework, together with its numerical properties and applications
to Hamiltonian systems, will be presented in a forthcoming
paper~\cite{Szuminski:26LRM::}. The corresponding open-source
implementation of the algorithms used for the present Heavy Atwood
Machine model is publicly available at
\url{https://doi.org/10.18150/7INU21}.

\subsection{Construction of Lyapunov Refined Maps}

\begin{figure*}[t]	\centering
	\includegraphics[width=0.75	\linewidth]{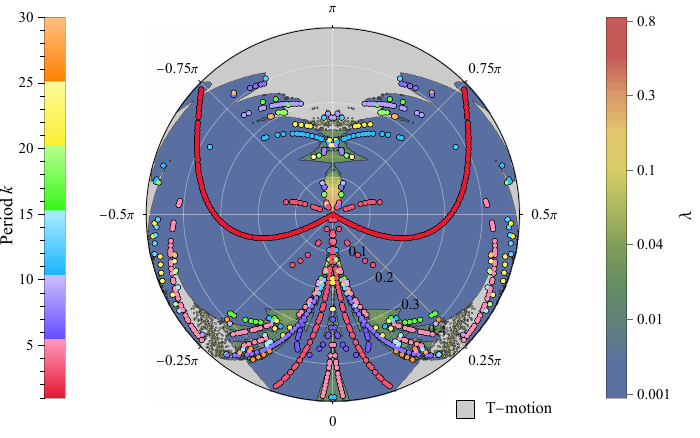}
	\caption{(Color online) Lyapunov Refined Map (LRM) constructed from the Lyapunov diagram shown in Fig.~\ref{fig:lyap_polar1}(c).
		The background color map represents the largest Lyapunov exponent $\lambda$ in the $(R_0,\Theta_0)$-plane, while gray regions denote terminating motion (T-motion).
		Superimposed colored points correspond to elliptic periodic solutions embedded within the regular regions of the phase space.
		The color of each point encodes the geometrical period $k$ of the corresponding orbit, as indicated by the accompanying color bar.
		The figure reveals organized families of resonance curves and accumulation structures inside the regular domains.
		The polar plots of exemplary periodic orbits are depicted in Fig.~\ref{fig:periodicki_polar}.
		\label{fig:periodicki}}
\end{figure*}

Applying the periodicity detection procedure described in the previous
subsection to the initial-condition plane $(R_0,\Theta_0)$ yields the
Lyapunov Refined Map shown in Fig.~\ref{fig:periodicki}. The LRM combines
the Lyapunov map defined by~\eqref{eq:IC_map} with the set
$\mathcal P$ of detected periodic trajectories given by
\eqref{eq:LRM_set}. Accordingly, the background represents the largest
Lyapunov exponent, whereas the superimposed colored points indicate the
locations of the detected periodic trajectories. Their colors encode the
geometrical period $k$ of the corresponding orbit of the Poincar\'e map.

Compared with the classical Lyapunov map, the proposed visualization
provides substantially richer dynamical information. While the Lyapunov
exponent distinguishes regular, chaotic, and terminating motion, the
additional periodicity analysis resolves the internal organization of the
regular regions. In particular, the uniformly blue domains of the
Lyapunov map are shown to contain an intricate network of periodic
orbits organized into resonance curves, branches, and compact clusters.

As anticipated, all detected periodic trajectories lie inside the
regular regions of the Lyapunov map,  where the largest Lyapunov exponent remains below the prescribed
regularity threshold. More importantly, the refined
representation demonstrates that these regular regions possess a rich
internal organization rather than forming homogeneous domains.
Periodic trajectories are arranged into well-defined resonance curves,
chains, and clusters embedded within the surrounding quasi-periodic
regions. In this way, the LRM uncovers the resonance structure hidden
inside the uniformly blue regions of the classical Lyapunov diagram and
provides a considerably more detailed picture of the regular dynamics.

The LRM reveals several characteristic features of the resonance
organization. Low-period periodic orbits are organized into smooth
resonance curves extending over large portions of the phase space,
whereas higher-period orbits form increasingly intricate families
accumulating near the boundaries between regular and chaotic motion.
Resonance branches frequently split, merge, and terminate, producing a
hierarchical web of periodic structures. Dense clusters of high-order
periodic trajectories are especially visible near the edges of chaotic
layers, indicating the presence of complicated resonance interactions.
Examples of representative periodic trajectories corresponding to
selected points of the LRM are presented in
Fig.~\ref{fig:periodicki_polar}. They illustrate the geometric
appearance of several resonance families detected by the proposed
algorithm.

The Lyapunov Refined Map therefore provides a direct geometric
visualization of the resonance structure of the phase space. Unlike classical Lyapunov maps, which classify
trajectories solely according to their dynamical stability, the proposed
representation simultaneously identifies the location, geometrical
period, and organization of periodic trajectories embedded inside the
regular domains. In this way, the LRM uncovers the fine topological
structure of the phase space while preserving the global information on
chaos contained in the Lyapunov exponent.

\begin{figure*}[t]	\centering
	\subfigure[$  k=2,\,  R_0=0.22,\, \Theta_0=1.29$]{
		\includegraphics[width=0.32\linewidth]{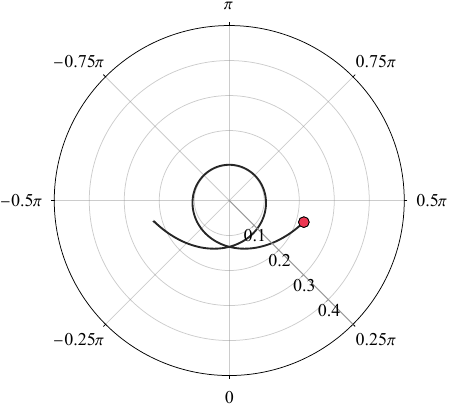}}
	\subfigure[$ k=4,\,R_0=0.44,\,  \Theta_0=0.26$]{
		\includegraphics[width=0.32\linewidth]{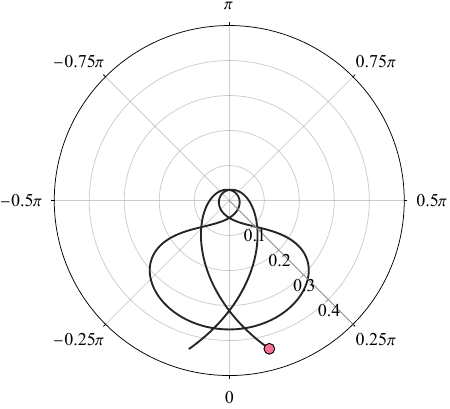}}
	\subfigure[$  k=6,\,  R_0=0.48,\, \Theta_0=0.05$]{
		\includegraphics[width=0.32\linewidth]{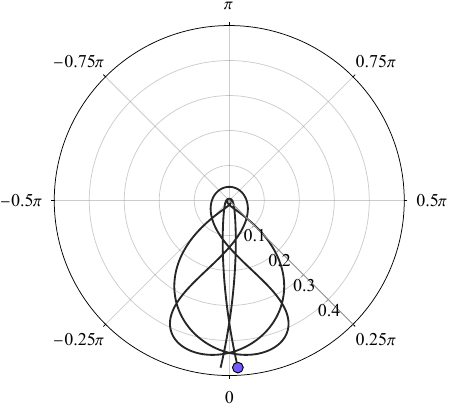}}\\
	\subfigure[$ k=10,\,R_0=0.45,\,  \Theta_0=0.57$]{
		\includegraphics[width=0.32\linewidth]{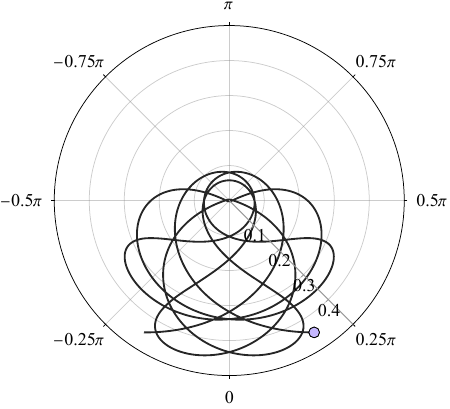}}
	\subfigure[$  k=16,\,  R_0=0.45,\, \Theta_0=1.3$]{
		\includegraphics[width=0.32\linewidth]{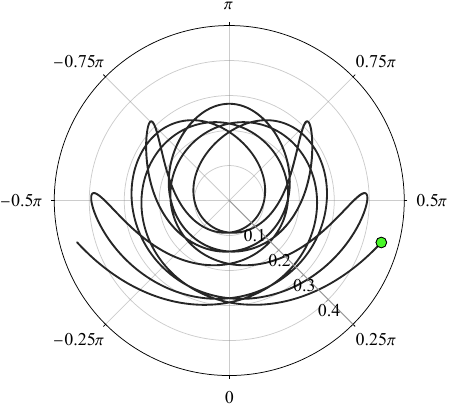}}
	\subfigure[$ k=28,\,R_0=0.31,\,  \Theta_0=2.86$]{
		\includegraphics[width=0.32\linewidth]{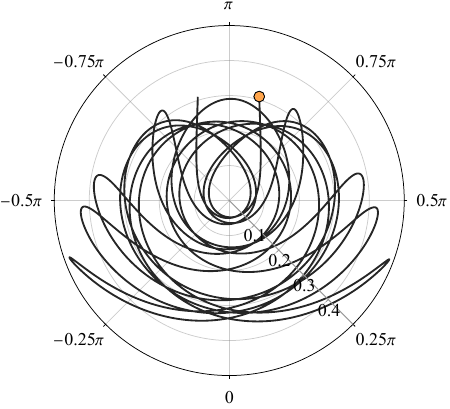}}
	\caption{(Color online)    Representative periodic trajectories identified within the Lyapunov Refined Map presented in Fig.~\ref{fig:periodicki}. The panels show periodic solutions of different periods $k$, illustrating the progressive increase in the geometrical complexity of the dynamics. Each colored point denotes the corresponding initial condition $(R_0,\Theta_0)$ selected from the LRM in Fig.~\ref{fig:periodicki}. The trajectories are displayed in polar coordinates $(R,\Theta)$ for the fixed parameters $\mu=3$, $\eta=1/2$, and $\alpha=1$.
		\label{fig:periodicki_polar}}
\end{figure*}

\begin{figure*}[t]	\centering
	\subfigure[Global view]{
		\includegraphics[width=0.465	\linewidth]{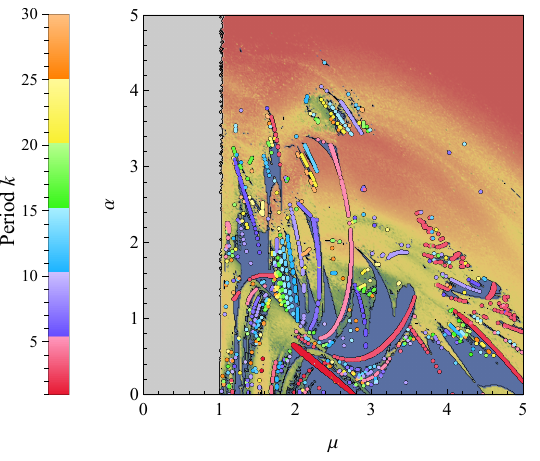}}
	\subfigure[Magnification of the central part]{
		\includegraphics[width=0.51\linewidth]{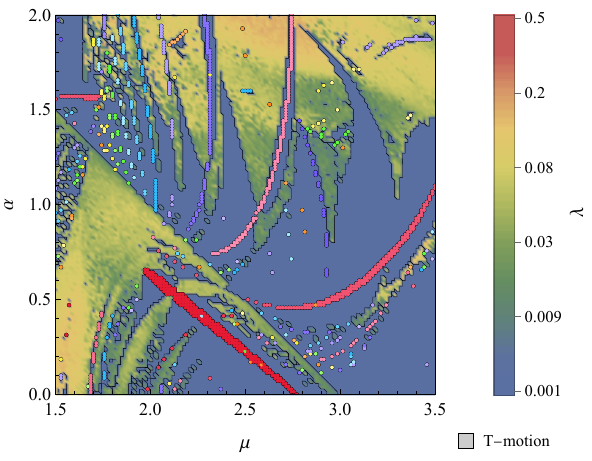}}  \caption{(Color online) Global view and magnification of the Lyapunov Refined Map (LRM) constructed from the Lyapunov diagram shown in Fig.~\ref{fig:parki}(a).
		The background color map represents the largest Lyapunov exponent $\lambda$ in the $(\mu,\alpha)$-plane, using the same logarithmic color scale as in Fig.~\ref{fig:parki}(a); blue regions correspond to regular dynamics, warmer colors indicate increasing chaoticity, while gray regions denote terminating motion (T-motion).
		Superimposed colored points correspond to detected periodic solutions embedded within the regular regions of the phase space.
		The color of each point encodes the geometrical period $k$ of the corresponding orbit, as indicated by the accompanying period color bar.
		\label{fig:LRM_parki}}
\end{figure*}

The same periodicity detection procedure can be applied to the parameter
plane $(\mu,\alpha)$ while keeping the initial conditions fixed, yielding
the LRM shown in Fig.~\ref{fig:LRM_parki}. As in the initial-condition
map, the background represents the largest Lyapunov exponent, whereas
the superimposed colored points indicate the detected periodic
trajectories, with the color encoding their geometrical period $k$.

The parameter-space LRM presented in Fig.~\ref{fig:LRM_parki}(a) shows
that periodic solutions occupy only selected regions of the parameter
space and are organized into well-defined resonance curves embedded
within the regular domains. Its magnification shown in
Fig.~\ref{fig:LRM_parki}(b) demonstrates that these structures possess a
remarkably rich hierarchical organization. Individual resonance curves
split, merge, and accumulate, producing dense resonance webs that remain
completely hidden in the corresponding Lyapunov diagram.

A clear correlation is observed between the distribution of periodic
trajectories and the geometry of the regular regions. Broad regular
domains are typically populated by low-period resonances, whereas narrow
filamentary regions contain predominantly higher-period solutions. As a
result, the LRM provides a much finer characterization of the parameter
space than the Lyapunov map alone, revealing the resonance organization
underlying the regular dynamics and identifying parameter values
associated with specific periodic motions.

Overall,  the parameter-space Lyapunov Refined Map demonstrates
that the resonance structure of the heavy SAM depends strongly on the
system parameters. It provides a global overview of how periodic
families are distributed throughout the parameter space and identifies
regions associated with low- and high-order resonances. In this way, the
LRM complements the classical Lyapunov map by revealing the resonance
organization hidden inside the regular parameter domains.

\section{Integrability}

The numerical investigations presented in the previous sections, including Lyapunov exponent maps, Lyapunov Refined Maps, Poincar\'e sections, and bifurcation diagrams, consistently reveal a rich mixture of regular and chaotic dynamics in the heavy swinging Atwood machine. In particular, the appearance of extended chaotic regions strongly suggests the absence of Liouville integrability for a broad range of parameter values.

Nevertheless, numerical evidence alone cannot establish non-integrability.
Since all computations are necessarily performed on finite grids in parameter
space and for finite integration times, they cannot exclude the possibility
that exceptional parameter values exist for which the system becomes
Liouville integrable. Such situations are well known in Hamiltonian dynamics.
Indeed, the classical swinging Atwood machine with a massless string provides
a remarkable example: although its dynamics is generically non-integrable, it
admits an exceptional Liouville integrable case for the mass ratio
$\mu=3$.

This observation naturally motivates a rigorous analytical study of the heavy
SAM, whose dynamics depends additionally on the string mass parameter
$\alpha$. The central question is whether the inclusion of string inertia
preserves any exceptional integrable regime or whether it destroys
integrability completely.

To answer this question, we employ the Morales--Ramis theory based on
differential Galois methods. This approach provides necessary conditions for
Liouville integrability by relating it to the differential Galois group of
the normal variational equations obtained by linearising the Hamiltonian
system along a non-stationary particular solution. If the identity component
of this group is non-Abelian, then the Hamiltonian system cannot possess a
complete set of meromorphic first integrals in involution.

An important advantage of the heavy SAM is that it admits explicit
non-stationary particular solutions, allowing the variational equations to be
derived and analyzed in a fully constructive manner. Consequently, the model
provides an excellent framework for confronting numerical signatures of chaos
with rigorous analytical obstructions to Liouville integrability. The
Morales--Ramis theory has become one of the fundamental tools in the study of
Hamiltonian systems and has been successfully applied to numerous problems in
classical mechanics and dynamical systems; see, for example,~\cite{Maciejewski:11::,Maciejewski:05::,10.1063/5.0200592,
	mp:13::e,Combot:18::,Szuminski:24::,Szuminski:25::JSV_VLDP,Szuminski2026}.
We state the following theorem.
\begin{theorem}
	Let $\mu,\eta$ and $\alpha$ be  positive parameters with $\alpha\neq 0$.
	Then the Hamiltonian system governed by Hamiltonian~\eqref{eq:HH}, describing the Swinging Atwood Machine with a
	massive string, is not Liouville integrable in the class of first integrals
	that are meromorphic functions of the phase-space variables.
\end{theorem}
We now present the main steps of the proof, omitting only the standard
technical details of the differential Galois computations. The interested reader may find a comprehensive exposition of these tools in references~\cite{Morales:99::,Put:99::}, as well as in the appendices of~\cite{mp:13::c}.

\subsection{Variational equations along $\vGamma(\tau)$}
Let $(\xi_R,\xi_\Theta,\xi_{P_R},\xi_{P_\Theta})$ denote variations of
$(R,\Theta,P_R,P_\Theta)$. Linearising the vector field associated with
the Hamiltonian system~\eqref{eq:Ham_sys} along the particular solution
$\vGamma(\tau)$, defined in~\eqref{eq:non_stat_part_sol}, yields the first variational equations (VE)
\begin{align}
	\label{eq:var1}
	\dot\xi_R
	=
	\dfrac{1}{\mathcal{M}}\,\xi_{P_R}, \qquad
	\dot\xi_\Theta
	=
	\dfrac{1}{D(R)}\,\xi_{P_\Theta},   \qquad
	\dot\xi_{P_R}
	=
	2\alpha\,\xi_R,                    \qquad
	\dot\xi_{P_\Theta}
	=
	-\,Q(R)\,\xi_\Theta.
\end{align}

The tangent direction to the particular solution $\vGamma(\tau)$ lies in the
$(R,P_R)$-subspace. The corresponding tangential variational subsystem $(\dot\xi_R, \dot\xi_{P_R})$ can be solved explicitly.
Indeed, eliminating $\xi_{P_R}$ yields
\begin{align*}
	\ddot\xi_R
	-
	\frac{2\alpha}{\mathcal{M}}\,\xi_R
	=
	0,
\end{align*}
whose general solution reads
\begin{align*}
	\xi_R(\tau)
	=
	C_1\,\mathrm{e}^{\omega_0 \tau}
	+
	C_2\,\mathrm{e}^{-\omega_0 \tau},
	\qquad \xi_{P_R}(\tau)=\mathcal{M}\dot\xi_R,\qquad
	\omega_0=\sqrt{\frac{2\alpha}{\mathcal{M}}},
\end{align*}
Hence, the tangential variational equations are explicitly integrable and do not
produce any obstruction to integrability.
Accordingly, by the Morales--Ramis theorem, it is sufficient to analyze the
normal variational equation along the particular solution
$\vGamma(\tau)$, namely
\begin{align*}
	\dot\xi_\Theta
	=
	\frac{1}{D(R)}\,\xi_{P_\Theta},
	\qquad
	\dot\xi_{P_\Theta}
	=
	-\,Q(R)\,\xi_\Theta.
\end{align*}
Eliminating $\xi_{P_\Theta}$ and using the identity $D'(R)=2Q(R)$, the normal
variational equation (NVE) can be written in the compact form
\begin{align}
	\label{eq:NVE_logD}
	\ddot{\xi}_\Theta
	+
	(\log D)' \left(\dot R\,\dot{\xi}_\Theta + \frac12\,\xi_\Theta\right)
	=
	0,\qquad \text{where}\qquad ('\equiv \rmd/\rmd R).
\end{align}

In order to extract further information on the differential Galois group of
the above NVE, it is convenient to perform a change of the independent variable.
This allows us to transform~\eqref{eq:NVE_logD} into a linear differential
equation with rational coefficients, for which the analysis of the differential
Galois group, and in particular of its identity component, is in general more
tractable.

\subsection{Differential Galois analysis}
Assume that system~\eqref{eq:Ham_sys} is Liouville integrable in the class
of meromorphic first integrals. Then, by the Morales--Ramis theorem, the
identity component of the differential Galois group of the normal
variational equation~\eqref{eq:NVE_logD} must be Abelian.
Consequently, it is sufficient to show that this identity component is
non-Abelian.

To this end, we transform the normal variational equation into a linear
equation with rational coefficients, which allows us to apply the
Kovacic algorithm~\cite{Kovacic:86::}, which classifies Liouvillian solutions of
second-order linear differential equations with rational coefficients.

For technical reasons, the subsequent analysis is split into two distinct cases.
\subsubsection{The massless-string case $\alpha=0$}

In this scenario, corresponding to the massless-string limit, the normal
variational equation~\eqref{eq:NVE_logD} simplifies considerably. Indeed, for
$\alpha=0$ one has
$D(R)=R^{2}$ and $(\log D)'=2/R$.
Substituting these expressions into~\eqref{eq:NVE_logD}, we obtain
\begin{align}
	\label{eq:NVE_alpha0_expanded}
	\ddot{\xi}_\Theta
	+
	\frac{2\dot R}{R}\,\dot{\xi}_\Theta
	+
	\frac{1}{R}\,\xi_\Theta
	=
	0.
\end{align}
Next, we introduce the change of the independent variable
\begin{align}
	\label{eq:change}
	\tau\longrightarrow z=\frac{\mu-1}{E}R(\tau),
	\qquad E\neq0,\qquad \dot{\xi}_\Theta=\dot z\,\xi_\Theta',
	\qquad
	\ddot{\xi}_\Theta=\ddot z\,\xi_\Theta'
	+\dot z^{\,2}\xi_\Theta'',\qquad '\equiv (\rmd/\rmd z).
\end{align}
For $\alpha=0$, the reduced radial motion on $\mathcal N$, defined by~\eqref{eq:ee}, satisfies
\[
	\dot R^{\,2}
	=
	\frac{2\bigl(E-(\mu-1)R\bigr)}{\mathcal M},
	\qquad
	\ddot R
	=
	-\frac{\mu-1}{\mathcal M},\quad \implies\quad \dot z^{\,2}
	=
	\frac{2(\mu-1)^2}{E\mathcal M}(1-z),
	\qquad
	\ddot z
	=
	-\frac{(\mu-1)^2}{E\mathcal M}.
\]
Substituting the above the relations into~\eqref{eq:NVE_alpha0_expanded}, we obtain the equation with rational coefficients.
\begin{align}
	\label{eq:NVE_alpha0_z_compact}
	\xi_\Theta''
	+
	\frac{A_1(z)}{A_2(z)}\,\xi_\Theta'
	+
	\frac{A_3}{A_2(z)}\,\xi_\Theta
	=
	0,\quad\text{where}\quad A_1(z)=4-5z,
	\qquad
	A_2(z)=2z(1-z),
	\qquad
	A_3=\frac{1+\mu}{\mu-1}.
\end{align}

As first shown by Morales~\cite{Morales:01::}, equation~\eqref{eq:NVE_alpha0_z_compact} reduces to the Gauss hypergeometric equation. Its differential Galois analysis implies that the identity component of
the corresponding differential Galois group is Abelian only for
$\mu=3$. Consequently, the classical swinging Atwood machine is Liouville integrable exclusively for this exceptional value of the mass ratio.
\subsubsection{The massive-string case $\alpha\neq0$}

We now consider the genuinely massive--string case, i.e.\ $\alpha\neq 0$.
In order to rationalize the coefficients of~\eqref{eq:NVE_logD}, we
introduce the new independent variable
\begin{equation}
	\label{eq:change_alpha}
	\tau\longrightarrow
	z
	=
	-\frac{\alpha}{3}\,R(\tau),
	\qquad
	\alpha\neq 0.
\end{equation}
Using the energy first integral~\eqref{eq:ee} with the relations~\eqref{eq:rel}, we obtain
\begin{align*}
	\dot z^{\,2}
	=
	\frac{2\alpha}{9\mathcal M}
	\left(
	\alpha E
	+
	3\bigl(\mu-1+\alpha\eta\bigr)z + 9z^2
	\right),\qquad \text{and}\qquad 	\ddot z
	=
	\frac{\alpha}{3\mathcal M}
	\left(
	\mu-1+\alpha\eta+6z
	\right).
\end{align*}

Substituting the above  into
\eqref{eq:NVE_logD} and dividing the resulting equation by $\dot z^{\,2}$,
we transform the normal variational equation~\eqref{eq:NVE_logD} into the rational form
\begin{align}
	\label{eq:NVE_alpha_z_final}
	\xi_\Theta''
	+
	\frac{B_1(z)}{B_2(z)}\,\xi_\Theta'
	+
	\frac{B_3(z)}{B_2(z)}\,\xi_\Theta
	=
	0,
\end{align}
where the common denominator is
\begin{align*}
	B_2(z)
	 & =
	z(z-1)\,P(z),
	\qquad
	P(z):=z^{2}+b z+c,\quad \text{with}\quad  	b=\frac{\mu-1+\alpha\eta}{3},
	\qquad
	c=\frac{\alpha E}{9}.
\end{align*}
The numerator of the first-derivative term reads
\begin{align*}
	B_1(z)
	 & =
	a_3 z^{3}+a_2 z^{2}+a_1 z+a_0,
\end{align*}
where the coefficients are explicitly given by
\begin{align*}
	a_3  =4,                                                \quad
	a_2  =\frac{7\mu+7\alpha\eta-25}{6},                      \quad
	a_1  =\frac{5(1-\mu-\alpha\eta)}{6}+\frac{\alpha E}{3},   \quad
	a_0  =-\frac{2\alpha E}{9}.
\end{align*}
Finally, the zeroth-order term is
\begin{align*}
	B_3(z)
	=
	\frac{\mathcal M}{12}\,\left(2-3z\right),
	\qquad
	\mathcal M=1+\mu+\alpha.
\end{align*}

For the further effective analysis of the rationalized variational
equation~\eqref{eq:NVE_alpha_z_final}, we perform an additional change
of the dependent variable
\begin{equation}
	\label{eq:change1}
	\xi_\Theta(z)
	=
	y(z)\,
	\exp\!\left[
		-\frac12
		\int^z \frac{B_1(s)}{B_2(s)}\,\rmd s
		\right],
\end{equation}
which transforms equation~\eqref{eq:NVE_alpha_z_final} into its reduced
(second-order normal) form
\begin{equation}
	\label{eq:rr}
	y'' = r(z)\, y\,\qquad \text{where}\qquad r(z)=
	\frac{
		2 B_1'\,B_2
		-
		2B_1\,B_2'
		+
		B_1^2
		-
		4 B_2\,B_3
	}{4\,B_2^2}.
\end{equation}

Although the change of variable~\eqref{eq:change_alpha} transforming the variational equation~\eqref{eq:NVE_logD} into the rational form \eqref{eq:NVE_alpha_z_final}, together with its subsequent reduction to the normal form~\eqref{eq:rr}, may alter the differential Galois group
of the variational equations, these transformations do not affect its identity component $\mathcal G^0$. Therefore, in order to prove the non-integrability of the original Hamiltonian system, it is sufficient to show that the identity component of the differential Galois group of the reduced equation~\eqref{eq:rr} is not Abelian.

The differential Galois group $\mathcal G$ of equation~\eqref{eq:rr} is an algebraic subgroup of
$\operatorname{SL}(2,\mathbb C)$. We now recall a classical lemma that classifies all possible types of the group $\mathcal G$ and the corresponding structures of solutions of the reduced equation.

Following Kovacic’s
approach~\cite{Kovacic:86::}, we state as follows.

\begin{lemma}[Kovacic]
	\label{thm:kovacic_a}
	Let $\mathcal G$ be the differential Galois group of
	\begin{equation}
		\label{eq:kov_reduced}
		y'' = r(z)\,y,
		\qquad r(z)\in \mathbb{C}(z).
	\end{equation}
	Then exactly one of the following cases occurs:

	\begin{enumerate}
		\item[] (Case~1)\, $\mathcal G$ is conjugate to a subgroup of the triangular group and
		      admits a Liouvillian solution
		      $y=P(z)\exp(\int\xi dz)$,
		      where $P\in\mathbb C[z]$ and $\xi\in\mathbb C(z)$.

		\item[] (Case~2)\,  $\mathcal G$ is conjugate to a subgroup of the infinite dihedral group
		      $\mathcal D^\dagger$, and admits a Liouvillian solution
		      $y=\exp(\int\xi dz)$,
		      where $\xi$ is algebraic of degree $2$.

		\item[](Case~3)\, $\mathcal G$ is finite, so all solutions are algebraic.

		\item[](Case~4)\, $\mathcal G=\operatorname{SL}(2,\mathbb C)$, and no Liouvillian solution exists.
	\end{enumerate}
\end{lemma}
\begin{remark}
	Let us write
	\[
		r(z)=\frac{p(z)}{q(z)},
		\qquad
		p(z),q(z)\in\mathbb C[z],\qquad \text{with}\qquad \gcd(p,q)=1.
	\]
	The zeros of $q$ are precisely the finite poles of $r(z)$.
	Denote by
	\[
		\Sigma=\Sigma'\cup\{\infty\},
		\qquad
		\Sigma'=\{c\in\mathbb C \mid q(c)=0\},
	\]
	the set of all singular points of equation~\eqref{eq:kov_reduced}.
	For $c\in\Sigma'$, the order $\ord(c)$ is the multiplicity of $c$ as a
	zero of $q$, while the order at infinity is  $
		\ord(\infty)=\deg q-\deg p.
	$
\end{remark}
\begin{lemma}[Kovacic]
	\label{thm:kovacic_b}
	The following conditions are necessary for the corresponding cases of
	Lemma~\ref{thm:kovacic_a}.

	\begin{enumerate}
		\item[]
		      (Case~1)\, Every pole $c\in\Sigma'$ has even order or order $1$, and
		      $\ord(\infty)$ is even or satisfies $\ord(\infty)>2$.

		\item[]
		      (Case~2)\, The set $\Sigma'$ contains at least one pole $c$ with
		      $\ord(c)=2$ or $\ord(c)>2$ odd.

		\item[]
		      (Case~3)\, One has $\ord(c)\le2$ for all poles and
		      $\ord(\infty)\ge2$. Moreover, if
		      $r(z)=\sum_i a_i/(z-c_i)^2+\sum_j b_j/(z-d_j)$,
		      then $\sqrt{1+4a_i}\in\mathbb Q$ for every $i$,
		      $\sum_j b_j=0$, and, setting
		      $G=\sum_i a_i+\sum_j b_jd_j$,
		      one has $\sqrt{1+4G}\in\mathbb Q$.
	\end{enumerate}
\end{lemma}
We now determine the singularity structure of equation~\eqref{eq:rr}, which provides the information required for the application of the Kovacic algorithm.

Equation~\eqref{eq:rr} is a linear differential equation whose finite singular points coincide with the zeros of the polynomial $B_2(z)$,
namely
\begin{align*}
	z_0=0,\qquad z_1=1,\qquad z_\star=z_\pm,
\end{align*}
where $z_\pm$ are the roots of $P(z)$, which are generically distinct,
that is, whenever the discriminant $\Delta=b^2-4c\neq0$.
The degenerate case $\Delta=0$ corresponds to the exceptional energy
value
$
	E_0=(\mu+\alpha\eta-1)^2/4\alpha,
$
which is precisely the stationary energy of the Hamiltonian
system~\eqref{eq:Ham_sys} at the equilibrium point
\eqref{eq:equilibrium_full}.
If an additional first integral existed, it would be independent of the
chosen energy level and would therefore exist for all generic values of
the energy. Consequently, to exclude the existence of such an integral,
it is sufficient to consider the generic case
$
	E\neq E_0,
$
for which the roots of the polynomial $P(z)$ are distinct and the
singular points of the reduced equation remain simple.
Accordingly, throughout the remainder of the proof we assume
$E\neq E_0$.

For the generic case $E\neq E_0$, the finite singularities of the reduced
equation~\eqref{eq:rr} coincide with the zeros of $B_2(z)=z(z-1)P(z)$.
In particular, $z=0$ is a simple pole of the coefficient $r(z)$, i.e.\
$\ord(0)=1$, whereas $z=1$ and $z=z_\pm$ are poles of order two,
$
	\ord(1)=\ord(z_\pm)=2.
$
Moreover, the point at infinity is a regular singular point of degree $\ord(\infty)=2$.
Therefore, equation~\eqref{eq:rr} is Fuchsian with the set of regular
singular points
\begin{equation}
	\label{eq:sing}
	\Sigma=\{0,1,z_+,z_-,\infty\}.
\end{equation}

Since the reduced equation~\eqref{eq:rr} has double poles at
$z=1$, $z=z_\pm$ and $z=\infty$, the principal parts of the corresponding
Laurent expansions at the singular points $c\in\{1,z_+,z_-\}$ and at
infinity take the form
\[
	r(z)=\frac{a_c}{(z-c)^2}
	+\mathcal O\!\left(\frac{1}{z-c}\right),
	\quad
	r(z)=\frac{a_\infty}{z^2}
	+\mathcal O\!\left(\frac{1}{z^3}\right),
\]
respectively.
In the present case, one finds
\[
	a_1=-\frac14,
	\qquad
	a_{z_\pm}=-\frac{3}{16},
	\qquad
	a_\infty=2.
\]  Thus,   the respective differences of exponents at these singularities are $	\Delta_c=\sqrt{1+4a_c},$
which yields
\begin{equation}
	\label{eq:exp}
	\Delta_1=0,
	\qquad
	\Delta_{z_+}=\Delta_{z_-}=\frac12,
	\qquad
	\Delta_\infty=3.
\end{equation}
Now we prove the following lemma.

\begin{lemma}
	Let us assume that the string is massive with $\alpha\neq 0$, then the differential Galois group of the reduced equation~\eqref{eq:rr} is $\operatorname{SL}(2,\mathbb{C})$.
\end{lemma}
\begin{proof}
	From~\eqref{eq:sing} and~\eqref{eq:exp} we know that the reduced equation has regular singularities with
	\[
		\ord(0)=1,
		\qquad
		\ord(1)=\ord(z_\pm)=\ord(\infty)=2.
	\]
	Moreover, the condition $\Delta_1=0$ excludes the finite and dihedral cases of the Kovacic algorithm (see, e.g.,~\cite{Maciejewski:02::,Stachowiak:15::}). Hence, the differential Galois group can only be either reducible (triangular) or equal to $\operatorname{SL}(2,\mathbb C)$. It therefore remains to exclude the reducible case by applying Case~1 of the Kovacic algorithm.

	For each double pole $c\in\{1,z_+,z_-,\infty\}$ we define
	\[
		\alpha_c^\pm
		=
		\frac12
		\pm
		\frac12\sqrt{1+4a_c}
		=
		\frac12
		\pm
		\frac12\,\Delta_c,
	\]
	while for the simple pole at $z=0$ we set
	$
		\alpha_0^\pm=1.
	$
	Accordingly, the associated sets are
	\[
		E_0=\{1,1\},
		\qquad
		E_1=\Bigl\{\frac12,\frac12\Bigr\},
		\qquad
		E_{z_+}=E_{z_-}=\Bigl\{\frac34,\frac14\Bigr\},
		\qquad
		E_\infty=\{2,-1\}.
	\]

	Next, we consider the Cartesian product
	$
		E
		=
		E_0\times E_1\times E_{z_+}\times E_{z_-}\times E_\infty
	$
	and restrict attention to those elements
	$
		e=(e_0,e_1,e_{z_+},e_{z_-},e_\infty)\in E
	$
	for which
	\[
		d(e)
		:=
		e_\infty-e_0-e_1-e_{z_+}-e_{z_-}
		\in\mathbb Z_{\ge0}.
	\]
	A straightforward inspection shows that the only admissible element is
	\[
		e_0=1,
		\qquad
		e_1=\frac12,
		\qquad
		e_{z_+}=e_{z_-}=\frac14,
		\qquad
		e_\infty=2,
	\]
	for which
	$
		d(e)=0.
	$

	We now proceed to the third step of the Kovacic algorithm. Since $d(e)=0$, the polynomial $P$ must be constant, and therefore $P\equiv1$. Consequently, a Liouvillian solution in Case~(i), if it exists, must be of the form
	\[
		y(z)
		=
		\exp\!\left(\int\omega(z)\,dz\right),\qquad \text{where}\qquad \omega(z)
		=
		\sum_{c\in\Sigma'}
		\frac{e_c}{z-c}
		=
		\frac1z
		+\frac1{2(z-1)}
		+\frac1{4(z-z_+)}
		+\frac1{4(z-z_-)}.
	\]

	By the Kovacic algorithm, such a solution exists only if $\omega$ satisfies the Riccati equation
	\begin{equation}
		\label{eq:kov_case1_id}
		\omega'(z)+\omega(z)^2=r(z).
	\end{equation}
	A direct substitution of the above expression for $\omega$ into~\eqref{eq:kov_case1_id}, followed by comparison with the rational function~\eqref{eq:rr}, shows that this identity can hold only if
	\[
		\mathcal M
		=
		1+\mu+\alpha
		=
		0.
	\]
	This contradicts the assumptions $\mu>0$ and $\alpha>0$. Hence, the Riccati equation admits no solution of the required form, and Case~1 of the Kovacic algorithm is excluded.

	Cases~2 and~3 have already been excluded by the local exponent
	analysis, while the above argument rules out Case~1. Consequently,
	Case~4 is the only remaining possibility. Therefore,
	$
		\mathcal G=\operatorname{SL}(2,\mathbb C).
	$
	Consequently, the normal variational equation admits no Liouvillian solutions. By the Morales--Ramis theorem, the Hamiltonian system cannot be Liouville integrable. This completes the proof.
\end{proof}

\section{Discussion and conclusions}

One of the long-standing questions in the theory of nonlinear Hamiltonian
systems concerns the borderline between integrable and non-integrable
dynamics. This problem becomes especially subtle in mechanical models
depending on parameters, where even small, physically natural modifications
may lead to a qualitative change of the global behavior. Although most
realistic systems are known to be non-integrable and chaotic, exceptional
integrable cases remain important, as they provide reference points for both
analytical and numerical studies.

Various types of pendulum systems form a natural testing ground for such
questions. The swinging Atwood machine is a particularly interesting example,
as it combines gravitational motion with geometric constraints in a relatively
simple Hamiltonian setting. It is well known that the classical model with a
massless string admits an exceptional integrable case for the mass ratio
$\mu=3$, even though numerical simulations indicate chaotic behavior for
generic parameter values. This naturally leads to the question of whether this
integrability survives once more realistic physical effects are taken into
account.

In this work we addressed this issue for the swinging Atwood machine with a
massive string. The presence of a nonzero linear mass density introduces a
configuration-dependent moment of inertia, which fundamentally changes the
structure of the kinetic energy and modifies the Hamiltonian formulation of the
system. Numerical experiments show that even very small deviations from the
massless-string limit, corresponding to $\alpha\ll1$, are sufficient to
destroy the regular dynamics associated with the integrable case $\mu=3$.
However, numerical evidence alone cannot provide a definitive answer
concerning Liouville integrability.

For this reason, we carried out a rigorous analytical study based on the
Morales--Ramis theory and differential Galois methods. By analyzing the normal
variational equations along explicit non-stationary radial solutions, we
showed that, for any nonzero value of the string mass parameter
$\alpha\neq0$, the system does not admit an additional first integral that is
meromorphic in the phase-space variables. The obstruction is of algebraic
nature and follows from the fact that the differential Galois group of the
normal variational equation is the full group
$\operatorname{SL}(2,\mathbb{C})$, whose identity component is non-Abelian.

The analytical results are complemented by an extensive numerical study of the
global dynamics. Using Poincar\'e sections, Lyapunov exponent maps,
bifurcation diagrams, and energy-constrained dynamical maps, we revealed a
rich coexistence of regular, chaotic, and terminating trajectories. In
particular, we demonstrated that the geometry of the Hill regions and
zero-velocity curves play a fundamental role in organizing the phase space
and governs the transitions between qualitatively different dynamical
regimes.

An important contribution of this work is the introduction and application of the  Lyapunov Refined Maps  methodology to the heavy Swinging Atwood Machine. By combining Lyapunov exponents with automatic periodic-orbit detection based on recurrence properties of the Poincaré map, this approach overcomes the intrinsic inability of standard Lyapunov diagrams to distinguish between periodic and quasi-periodic trajectories. Consequently, the proposed framework reveals resonance webs, families of periodic solutions, periodic windows embedded in chaotic layers, and high-order resonant structures hidden inside regular domains, providing a substantially more detailed description of the phase-space organization than classical numerical techniques.

Our results clearly demonstrate that the integrable case of the classical
swinging Atwood machine is structurally unstable. While the model with a
massless string is Liouville integrable for $\mu=3$, this property is
destroyed by any nonzero linear mass density of the string, no matter how
small. Therefore, the inclusion of string inertia leads to a genuine loss of
integrability.

From a dynamical point of view, this loss can be understood as a consequence
of the coupling between the radial and angular degrees of freedom induced by
the configuration-dependent effective inertia. This mechanism provides a
natural explanation for the chaotic behavior observed in numerical studies,
including parameter regimes close to the stationary energy level, and explains
the emergence of resonance structures and chaotic layers, visualized by the
Lyapunov Refined Maps.

Beyond the specific model considered here, the proposed methodology provides a
general framework for the investigation of multidimensional Hamiltonian
systems. The combination of LRM with the
Morales--Ramis differential Galois theory establishes a direct connection
between numerical exploration of global dynamics and rigorous analytical
obstructions to Liouville integrability. We expect that the proposed approach will be applicable to a broad class of nonlinear mechanical systems, extending well beyond pendulum-type models. The algorithmic formulation and detailed description of the proposed methodology are currently under development and will be presented in a separate publication.

In summary, the swinging Atwood machine with a massive string provides a clear example of how a physically realistic modification of an exceptional integrable Hamiltonian system destroys Liouville integrability while giving rise to a highly organized and remarkably rich dynamical structure. At the same time, the present work demonstrates that the combination of modern numerical diagnostics, embodied here by the Lyapunov Refined Maps, with rigorous analytical methods based on the Morales--Ramis theory constitutes a powerful and comprehensive framework for the study of global dynamics and integrability of nonlinear Hamiltonian systems.

\section*{Data and code availability}

The Wolfram Language implementation developed for this study, including the computation of Lyapunov exponents, the construction of Lyapunov Refined Maps,  and the visualization scripts used to
generate the figures, is publicly available in the  repository
at \url{https://doi.org/10.18150/7INU21}. The numerical data underlying the figures and other results presented in this article are available from the corresponding author upon reasonable request.

\section*{Acknowledgements}

For Open Access, the authors have applied a CC BY public copyright licence to any Author Accepted Manuscript (AAM) version arising from this submission.

The present work was initiated during a research internship of W.~S. at the Łódź University of Technology. A substantial part of the research was subsequently carried out during a research stay at the Department of Mathematics, Universitat Autònoma de Barcelona, hosted by Professor Jaume Llibre.

J.~B. is a doctoral student at the Doctoral School of Exact and
Technical Sciences, University of Zielona Góra.

\section*{Funding}

This research was supported by the National Science Centre, Poland (NCN), under Grant No.~2020/39/D/ST1\\ /01632, and by the Polish National Agency for Academic Exchange (NAWA) through the Bekker Programme, Fellowship No.~BPN/BEK/2025/1/00055.

\section*{Declarations}
\textbf{Compliance with ethical standards}\\
\textbf{Conflicts of interest} The authors declare no conflict of interest.\\


\clearpage
\thispagestyle{empty}
\begin{center}
    \includegraphics[
        width=0.85\textwidth,
        height=0.90\textheight,
        keepaspectratio
    ]{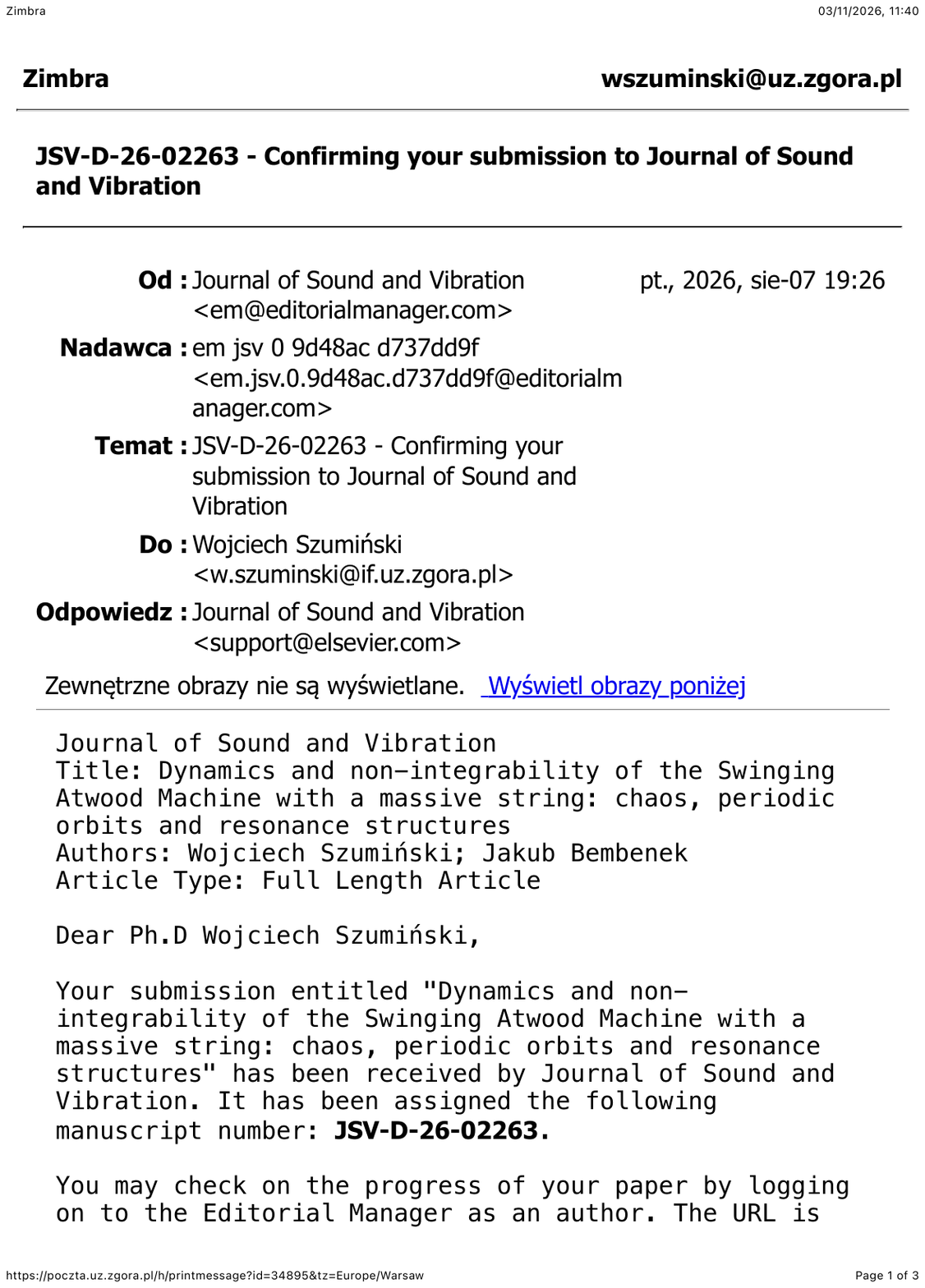}
\end{center}

\clearpage
\thispagestyle{empty}
\begin{center}
    \includegraphics[
        width=0.85\textwidth,
        height=0.90\textheight,
        keepaspectratio
    ]{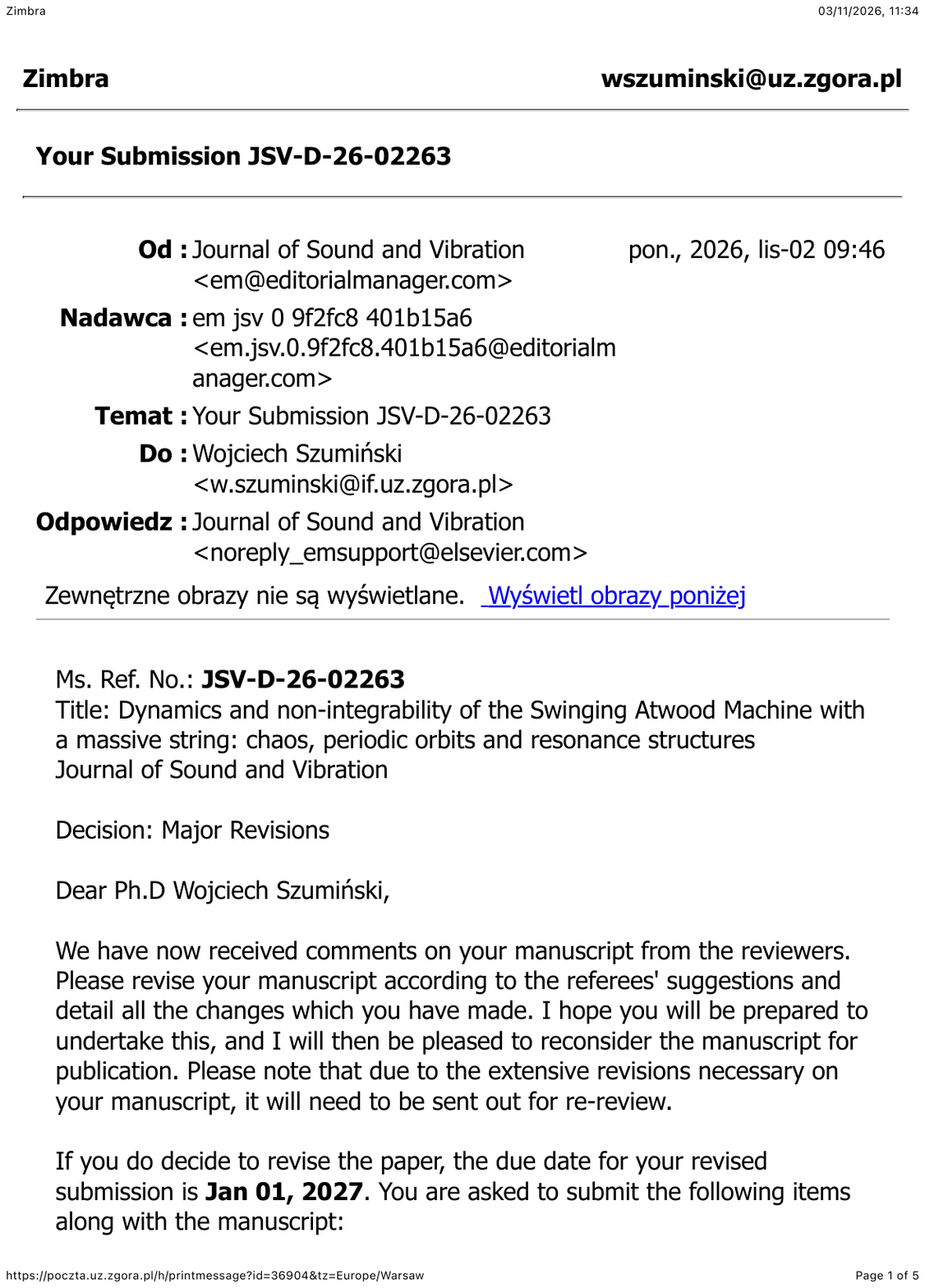}
\end{center}

\end{document}